\documentclass[letterpaper,twocolumn,10pt]{article}
\usepackage{usenix}
\usepackage{enumitem} 
\usepackage{amsthm}  
\usepackage{amsmath}
\usepackage{amssymb}  
\usepackage{enumitem}
\usepackage{bm}
\usepackage{booktabs}       
\usepackage{tabularx}       
\usepackage{amsmath}        
\usepackage{threeparttable} 
\usepackage{url} 
\usepackage{accents}
\usepackage{verbatim}
\usepackage{mdframed}
\usepackage{cleveref}
\usepackage{mathtools}
\usepackage{booktabs}
\usepackage{makecell}
\usepackage{dblfloatfix} 
\usepackage{multirow}  
\usepackage{array}     
\usepackage[skins]{tcolorbox}
\usepackage[utf8]{inputenc}
\usepackage{multicol}
\usepackage{balance}
\usepackage[available]{usenixbadges}
\usepackage{subcaption}  
\usepackage[cal=cm]{mathalfa}  

\theoremstyle{plain}      
\theoremstyle{definition} 
\theoremstyle{remark}     

\newtheorem{theorem}{Theorem}[section]    

\newtheoremstyle{definition}
{12pt}   
{12pt}   
{\it}    
{0pt}    
{\bfseries} 
{.}       
{5pt}     
{}        
\theoremstyle{definition}
\newtheorem{definition}{Definition}
\DeclareSymbolFont{letters}{OML}{cmm}{m}{it}
\SetSymbolFont{letters}{bold}{OML}{cmm}{b}{it}

\newcommand{\greyback}[1]{%
	\colorbox{gray!15}{#1}%
}

\newcommand{\greycomment}[1]{{\color{gray!70} #1}}
\usepackage{tikz}
\usepackage{amsmath}

\begin{document}
\pagestyle{empty}


\date{}

\title{\Large \bf Differential Trust: Dynamic Multi-Authority Anonymous Credentials with Epoch-Weighted Updates}

\author{
{\rm Chen Li}\\
Tianjin University\\
\texttt{lichen1232022@163.com}
\and
{\rm Jianting Ning\thanks{Corresponding author. Part of this work was done while the author was at Wuhan University.}}\\
Zhejiang Sci-Tech University\\
\texttt{jtning88@gmail.com}
\and
{\rm Xiulong Liu}\\
Tianjin University\\
\texttt{xiulong\_liu@tju.edu.cn}
\and
{\rm Yulin Liu}\\
Wuhan University\\
\texttt{liuyulin@whu.edu.cn}
}

\maketitle

\begin{abstract}
Anonymous credentials (ACs) are fundamental to privacy-preserving authentication, allowing users to prove possession of attributes without revealing their identities. State-of-the-art ACs distribute credential issuance across multiple authorities, typically employing techniques such as Shamir's secret sharing or aggregate signatures. While this approach enhances system robustness and eliminates single point of failure, it treats all authorities equally in the credential issuance phase. This uniform treatment disregards the varying levels of trustworthiness or stake held by different authorities. Such limitation has become particularly problematic in modern decentralized systems like Proof-of-Stake networks, where the inherent trust differentiation among nodes cannot be leveraged in the credential issuance process.

To address this limitation, we propose the notion of \textbf{M}ulti-\textbf{A}uthority \textbf{A}nonymous \textbf{C}redentials with \textbf{E}poch-Based \textbf{W}eights (MA-ACEW), the first Multi-Authority Anonymous Credential (MA-AC) model that considers authorities' weight distribution in credential issuance. Crucially, MA-ACEW enables efficient credential updates when authority weight distributions change across epochs. The core of MA-ACEW is our novel Epoch-Bound Pointcheval-Sanders Signature (EB-PS) primitive, which binds signatures to specific time epochs. This temporal binding enables both weight-based credential issuance within epochs and efficient non-interactive credential updates across epochs. We formalize the EUF-eCMA unforgeability requirement for EB-PS and prove our construction satisfies it under a novel STB-GPS assumption. We then prove that our MA-ACEW construction achieves unforgeability, anonymity, and blindness. Finally, we present benchmarks demonstrating the efficiency of EB-PS and MA-ACEW. Remarkably, presenting a credential aggregated from 128 partial ones takes only 10.68 ms on average. \footnote{This is the full version of the paper published in the Proceedings of the 35th USENIX Security Symposium (USENIX Security 2026).}
\end{abstract}

\section{Introduction}

Digital authentication has evolved into a fundamental security primitive bridging digital and physical identities. Such authenticated identities enable secure cross-domain access to real-world services and resources without physical presence requirements. A prominent example is Estonia's \textit{e-Residency} program, launched in 2014, which enables global users from over 170 countries to obtain digital identities for accessing e-government services (e.g., company registration, tax filing), generating over €150 million in direct economic value \cite{e-residency}. 

However, despite providing convenient service access, this approach raises significant privacy concerns as users' digital identities become vulnerable to unauthorized exposure and potential misuse. At first glance, one might perceive privacy preservation and identity authentication as inherently contradictory objectives: \textit{How can one prove their identity while maintaining privacy?}

Yet, already in the 1980s, Chaum \cite{chaum1983blind}, \cite{chaum1985security} provided an elegant solution to this challenge through his pioneering work on cryptographic techniques for creating privacy-friendly and user-centric authentication solutions. Later, anonymous credentials (ACs) emerged as a cryptographic primitive that enables users to prove specific attributes (such as being over 18) without revealing any additional personal information or creating linkable traces across different authentications. Over the past decades, ACs have attracted significant research attention, resulting in a vast body of research exploring diverse approaches \cite{DBLP:conf/scn/CamenischL02, DBLP:conf/ccs/CamenischH02, DBLP:conf/crypto/CamenischL04, DBLP:conf/ccs/BaldimtsiL13, DBLP:conf/tcc/BelenkiyCKL08, DBLP:conf/asiacrypt/CamenischDHK15, fuchsbauer2019structure, garman2013decentralized, DBLP:conf/pkc/Sanders20, hanzlik2021little}.

Beyond theoretical research, anonymous credentials have found their way into various practical applications. A notable example is Direct Anonymous Attestation (DAA), which has been integrated into the Trusted Platform Module specification \cite{DBLP:conf/ccs/BrickellCC04} and has seen continued development \cite{DBLP:conf/ccs/KumarLLAPBWZ18, DBLP:conf/ccs/WesemeyerNTCSW20}. More recent applications include PrivacyPass for anonymous web authentication \cite{DBLP:journals/popets/DavidsonGSTV18}, privacy-preserving point collection systems \cite{DBLP:conf/ccs/BlomerBDE19}, and anonymous tokens \cite{DBLP:conf/crypto/KreuterLO020, DBLP:conf/ccs/KarantaidouRBKK24} for secure authorization.

Despite significant advances in functionality and adoption, ACs face two critical challenges in today's distributed landscape. First, reliance on a single trusted issuer creates a centralization vulnerability where compromised signing keys could produce unauthorized yet valid credentials. Second, the emergence of blockchain credential platforms such as Hyperledger Indy \cite{hyperledger-ursa}, Veramo \cite{veramo}, and Okapi \cite{okapi} has created a paradigm shift toward distributed architectures that fundamentally conflicts with centralized issuance models. These challenges motivate the need for distributed anonymous credential systems that maintain traditional security and privacy properties while supporting multiple authorities. Recent research has pursued two approaches: threshold issuance schemes \cite{sonnino2018coconut, doerner2023threshold, crites2023threshold} that distribute trust among multiple parties using Shamir's secret sharing \cite{shamir1979share}, and Multi-Authority Anonymous Credentials (MA-ACs) \cite{hebant2023traceable, DBLP:conf/ccs/MirBGLS23} that enable independent issuers to coexist within a single system through aggregate signatures with randomizable tags, allowing users to combine credentials from different sources into compact proofs while preserving privacy.

Even so, these distributed credential schemes face a significant misalignment with the inherent nature of modern distributed networks. In decentralized ecosystems, participants hold different levels of authority and trust, fundamentally shaping network dynamics and consensus mechanisms. However, existing ACs, both threshold-based and multi-authority, treat all partial credentials uniformly, disregarding the inherent trust asymmetry among issuing authorities. 
This uniform treatment fails to capture the natural heterogeneity of trustworthiness in blockchain governance. For instance, in Proof-of-Stake (PoS) systems \cite{kiayias2017ouroboros, david2018ouroboros} or Decentralized Autonomous Organizations (DAOs) \cite{wang2019decentralized}, a validator with a substantial stake or a governance delegate with high reputational backing inherently warrants greater influence than a peripheral node with minimal investment. Such differential weighting is crucial for security and fairness \cite{saad2021pos}, ensuring that decision-making power accurately reflects the stakeholders' tangible commitment to the ecosystem. Moreover, these power structures are not static but evolve over time, stake distributions fluctuate due to slashing events, delegation updates, or token transfers, requiring dynamic weight adjustments. Addressing this mismatch requires a paradigm shift from static, unweighted credential schemes toward dynamic, weighted mechanisms to foster sustainable decentralized ecosystems. We provide detailed motivation in Section \ref{problem}. Our analysis of this limitation led us to explore weighted authority models, from which we derive the following research challenge:

\textit{How to design an anonymous credential system with weight distribution among different authorities and simultaneously support dynamic weight updates?}

While existing works \cite{sonnino2018coconut, doerner2023threshold, DBLP:conf/ccs/MirBGLS23, crites2023threshold} have primarily focused on systems with equal or static weights, our work extends this line of research by considering dynamic weight distribution among authorities. To this end, we introduce an epoch-based mechanism that assigns and refreshes authority weights over time, together with efficient credential updates that ensure only epoch-valid credentials can be verified. Our main contributions are summarized as follows:

\begin{itemize}[nosep, leftmargin=*]
	\item \textbf{The EB-PS Primitive}. To enforce temporal constraints on signature validity, we introduce a new primitive 
called Epoch-Bound Pointcheval-Sanders Signature (EB-PS). In this primitive, the 
signing process is split into two phases: message signing and epoch-specific binding. 
Our construction of EB-PS builds upon AtoSa~\cite{DBLP:conf/ccs/MirBGLS23}, which 
itself extends the multi-message PS signature scheme~\cite{PS16short}. Our key 
innovation is the introduction of a dynamic management process for one of the 
secret key components. Specifically, we designate a component, the epoch key $z_j$, 
which, unlike the other static key components, must be periodically updated. 
Upon a transition to a new epoch $j'$, a freshly generated $z_{j'}$ replaces the old 
$z_j$. This mandatory key rotation is the core mechanism that binds the signature to a 
specific time frame. Consequently, a valid EB-PS signature can be viewed as a 
combination of a long-term signature and a short-term, 
epoch-specific signature. While preserving the randomization properties 
of AtoSa, EB-PS achieves remarkable efficiency: a signature aggregated from $n$ 
messages remains compact at only two group elements.
	\item \textbf{The MA-ACEW Construction}. Based on our EB-PS construction and the Weighted Threshold Signature scheme proposed by Das et al. \cite{das2023threshold}, we present a new AC construction, named Multi-Authority Anonymous Credentials with Epoch-Based Weights (MA-ACEW). To the best of our knowledge, MA-ACEW is the first anonymous credential system that enables weighted trust evaluation across multiple authorities (i.e., multiple issuers). MA-ACEW addresses practical scenarios where each issuer is assigned different weights according to the issuer's trust level or computational capability. Additionally, the redistribution or adjustment of issuer weights is allowed during epoch transitions. For credential updates during epoch transitions, we eliminate the need for users to re-interact with each issuer. Instead, users perform a one-time interactive obtaining process with an issuer, and subsequent operations involve the issuer computing epoch-specific partial credentials for users holding long-term credentials. When a user wants to present their credential to a verifier during epoch $j$, they only need to compute the corresponding proofs and combine the partial credentials into one. Finally, we extend the core framework to natively support three crucial features: (i) multi-attribute credentials, (ii) issuer hiding, and (iii) selective disclosure of attributes under arbitrary predicates. Collectively, these properties make MA-ACEW particularly suitable for practical scenarios, especially in distributed networks.
	\item \textbf{New Assumption and Formal Security Proofs}. We introduce and formalize a new hardness assumption, the Separable Time-Bound Generalized PS (STB-GPS) assumption. This assumption extends the well-established Generalized PS (GPS) assumption \cite{PS16short, kim2021practical} to more accurately capture the security requirements of epoch-based signature schemes where signature components are separable and time-bound, and we analyze its hardness in the Generic Group Model (GGM). We then establish a formal security model for our epoch-based primitive, defining Existential Unforgeability under chosen-message and Epoch-corruption Attack (EUF-eCMA). This model is formulated in the chosen-key setting \cite{boneh2003aggregate,lysyanskaya2004sequential,DBLP:conf/ccs/MirBGLS23} and is specifically designed to handle the dynamics of epoch transitions and, crucially, allows the adversary to perform epoch-secret corruptions. We then rigorously prove that our EB-PS construction achieves EUF-eCMA security under the STB-GPS assumption. Furthermore, for our  MA-ACEW construction, we provide formal security definitions for its three key properties: unforgeability, anonymity, and blindness. Building upon the security models from \cite{fuchsbauer2019structure}, \cite{hebant2023traceable}, and \cite{DBLP:conf/ccs/MirBGLS23}, our models introduce additional oracles to capture epoch transitions and partial credential issuance, which are crucial for our system. We then provide rigorous proofs demonstrating that MA-ACEW satisfies all these specified security features.
    
    \item \textbf{Implementation}. We implement both the EB-PS and MA-ACEW constructions in Golang. Additionally, we develop a smart contract for credential issuance and verification operations in MA-ACEW. For credential presentation, users need just 10.68 ms to show credentials aggregated from 128 partials. On-chain verification requires 1077K gas on Ethereum with pre-computed $\mathbb{G}_2$ exponentiation.
    
\end{itemize}

\section{Problem Statement}\label{problem}

\subsection{Problem Description}
In modern decentralized ecosystems, PoS has established itself as the fundamental mechanism for securing trust and consensus. Unlike traditional identity-based systems, authority in a PoS setting is derived strictly from economic backing, where entities exercise power for governance or resource allocation proportional to their held stake. A valid authorization is defined by accumulating a super-majority of the total stake, rather than a simple majority of entities. However, when applying decentralized privacy-preserving credentials to this setting, a structural misalignment arises. Existing decentralized anonymous credentials are inherently stake-agnostic, they treat every issuer's signature as identical, disregarding the fact that authority in PoS is strictly stake-dependent. Consequently, these schemes limit the applicability of privacy-preserving mechanisms in scenarios that require fine-grained, stake-based governance. To bridge this gap, it is necessary to construct a credential system that is compatible with the PoS trust model while strictly preserving user privacy. However, realizing such a system presents two core challenges:
\begin{description}
    \item[\textbf{A.1:}] \textit{The scheme must support a trust model where issuers possess differential weights corresponding to their specific stake or institutional credibility.}
    \item[\textbf{A.2:}] \textit{The scheme must support dynamic weight adjustments to reflect that issuer trust levels change over time, allowing weights to increase or decrease in response to changes in stake or confidence.}
\end{description}

\noindent\textbf{Limitation of Existing Works}. Current distributed anonymous credential systems adopt two main approaches. The first approach, threshold issuance for ACs \cite{sonnino2018coconut, doerner2023threshold}, typically employs a $(t,n)$ Shamir's secret sharing scheme \cite{shamir1979share}, where each of the $n$ issuers holds a share of the signing key, and any subset of $t$ issuers can collaboratively generate a credential. Alternatively, MA-ACs \cite{DBLP:conf/ccs/MirBGLS23, hebant2023traceable} are based on aggregate signatures with randomizable tags, allowing users to aggregate showings of credentials from different issuers (with respect to the same tag) into one compact showing, and can be viewed as a multi-signature scheme. While these approaches effectively distribute the issuing responsibility, they all adopt a uniform-weight paradigm, assigning equal importance to all credential issuers. 
	
    \noindent\textbf{Naive Solutions}. 
    A straightforward approach to adapt existing ACs to support arbitrary weights is through virtualization. In this model, an issuer with weight $w$ simply holds $w$ distinct signing keys and emulates $w$ separate virtual issuers. However, this method suffers from several critical limitations: (i) The signing cost, partial signature size, and user's computational load all scale linearly with the total weight $\mathrm{W}$, imposing severe scalability constraints. (ii) Any modification to the weight vector $\bm{w}$ necessitates costly key regeneration and credential re-issuance across the entire system. (iii) Requiring issuers to maintain numerous signing keys significantly increases system complexity and the risk of key compromise.

    The severe scalability issues of virtualization, particularly the linear growth in cost and data size, naturally lead to considering more specialized solutions like Weighted Secret Sharing (WSS) schemes. While the concept dates back to Shamir's work \cite{shamir1979share}, a notable recent development is the Weighted Ramp Secret Sharing (WRSS) scheme by Garg et al. \cite{garg2023cryptography}. Their construction, which is based on the Chinese Remainder Theorem (CRT), appears highly promising as it directly addresses the primary scaling limitation of virtualization: the share size for a party with weight $w$ is only $O(w)$ bits. This efficiency is achieved in the ramp setting \cite{10.1007/3-540-39568-7_20}, which allows for a gap between the privacy and reconstruction thresholds, seemingly providing a direct path to our desired functionality.

    Despite its elegance in solving the scaling problem, this solution is ultimately not a robust one. The core efficiency of WRSS is predicated on a fundamental constraint: its reliance on a ramp setting. This introduces a predefined gap between the reconstruction threshold $T$, the minimum combined weight of participants needed to recover the secret, and the privacy threshold $t$, the maximum combined weight that is guaranteed to learn no information about it. This gap (i.e., $T - t = \Omega(\lambda)$), while enabling efficiency, also introduces its own severe limitations for practical systems: i) Its reconstruction is set-dependent, meaning the algorithm to recover the secret changes based on the exact set of participants. This requires a preliminary coordination step to identify all active members, creating a bottleneck that is impractical for dynamic or decentralized networks. ii) It introduces security vulnerabilities for small weights, as an $O(w)$-bit share can be discovered via a brute-force attack if a party's weight $w$ is small. iii) It imposes a rigid and static weight structure, since a party's weight is intrinsically tied to its public parameter, and any change requires a costly, system-wide reset.

   \noindent\textbf{Summary of Challenges.} In summary, naive approaches are unsuitable for dynamic ACs: virtualization incurs prohibitive linear costs, while CRT-based WRSS rigidly binds weights to public parameters. Neither supports weight changes without costly resets. A practical solution must enable efficient updates and decouple weights from credentials—our work is the first to achieve both.
\subsection{Solution Overview}
Here we present a high-level overview of our solution, which is efficient and conceptually simple. To construct our anonymous credential system, we partition the system into epochs. Each epoch typically spans a fixed duration (e.g., one day or one week), during which issuers' weights remain constant. Our approach is structured into two phases: embedding time epochs into the signature scheme, followed by integrating weight settings based on this modified signature scheme.

\noindent\textbf{Basic Signature Scheme}. Our anonymous credential system leverages the Pointcheval-Sanders (PS) signature scheme. 
In its general form, a signer can sign a vector of messages $(m_1, \dots, m_n)$ using a secret key $\mathsf{sk} = (x, y_1, \dots, y_n)$. 
A signature is a pair of group elements $\sigma = (h, h^{x + \sum_{i=1}^n y_i m_i})$ for a random $h \in \mathbb{G}$. 
This structure supports efficient randomization for anonymity: a user can re-randomize a signature using a random $r \in \mathbb{F}_p$ to obtain $\sigma' = (h^r, (h^{x + \sum_{i=1}^n y_i m_i})^r)$, which is a valid signature that cannot be linked to the original.

Our innovation is to enable dynamic credential updates by partitioning the secret key vector. One component, $z_j$, serves as the epoch-specific key for epoch $j$, while the remaining $(y_1,\dots,y_k)$ sign user attributes $(m_1,\dots,m_k)$. The update protocol works as follows:
\begin{enumerate}[leftmargin=*]
    \item \textbf{Key Updates:} At each epoch $j$, the issuer replaces $z_{j-1}$ with $z_j$, keeping $(x, y_1,\dots,y_k)$ unchanged.
    \item \textbf{Signature Updates:} With a small public update token, users transform prior signatures into ones valid for the new epoch without re-issuance.
\end{enumerate}

    \noindent \noindent To build intuition for our construction, we begin by describing a basic version. In this simplified exposition, the signature base is the fundamental public parameter $h$, devoid of any user-specific tag, and the credential contains only a single attribute $m$. The signature form thus simplifies to $\sigma = (h, h^{x+y\cdot m+z_j\cdot F(\mathsf{ctx},j)})$. The term $z_j \cdot F(\mathsf{ctx},j)$ functions as a pseudo random function (PRF) output to provide domain separation for different credential contexts (identified by $\mathsf{ctx}$). 

     While the element $h$ is a randomly chosen parameter in the original PS signature scheme, it is often adapted in ACs to serve as a user-specific base. Our design leverages this concept by having the issuer store this base $h$ after a user's initial enrollment. Building on this, the issuer can unilaterally issue credential updates for each epoch $j$ by computing and distributing the temporal component $h^{F(\mathsf{ctx}, j)\cdot z_j}$. The exponent is carefully constructed: the epoch-specific secret $z_j$ ensures the value is fresh for each period, while the function $F$ binds the update to the specific context $\mathsf{ctx}$ and epoch $j$. Consequently, as $z_j$ is regenerated for each epoch, the credential update for users is a simple, non-interactive process that only requires fetching this new component.

\noindent\textbf{Incorporating weight setting.} Now we first consider the weight distribution across all issuers. In each epoch $j$, we denote this distribution by a vector $\bm{w}_j$. Weight vector adjustments occur only during epoch transitions, based on various factors beyond the scope of this paper. 

We note that each issuer is assigned a specific weight under epoch $j$. Consider a set of $n$ issuers, we use a bit vector $\bm{b}$ to represent the participating issuers in the multi-issuer credential $\mathsf{cred}$ for a user, where $b[i] = 1$ denotes participation and $b[i] = 0$ represents non-participation. Consequently, the total weight of the credential $\mathsf{cred}$ can be expressed as $\mathrm{W} = \langle \bm{w}_j, \bm{b} \rangle$, which is the inner product of vectors $\bm{w}_j$ and $\bm{b}$. Recall that each issuer's secret key is composed of $(x_i, y_i, z_i)$, where $y_i$ and $z_i$ are used for signing a message and incorporating epoch information, respectively. The public key component $\mathrm{X}_i = v^{x_i}$ corresponds to private key element $x_i$ and correlates with weight $w_{i,j}$. Thus, $\mathrm{X} = \langle \bm{pk}_x, \bm{b}\rangle$ represents the aggregated keys of participating issuers, where $\bm{pk}_x = (\mathrm{X}_1, \ldots, \mathrm{X}_n)$.

To prove the validity of a weight-based credential, we simultaneously demonstrate three crucial properties: i) the user's aggregated credential is validly obtained from the participating issuers represented by the vector $\bm{b}$; ii) the user's provided combined weight sum $\mathrm{W}$ exceeds the threshold defined by the verifier; and iii) the corresponding aggregated verification key $\mathrm{X}$ is correctly computed using the same vector $\bm{b}$. By verifying these properties, we confirm both the credential's validity and its compliance with the weight requirements, where issuers are weighted rather than treated equally.

For the proofs of inner products $\langle \bm{pk}_x, \bm{b}\rangle$ and $\langle \bm{w}_j, \bm{b}\rangle$, Das et al. \cite{das2023threshold} provide an elegant solution that substantially reduces verification complexity. We will demonstrate in Section \ref{MA-ACEW} how to integrate their approach with several minor but important modifications tailored to our construction. During epoch transitions, our solution requires recomputing only the critical components $(h^{F(\mathsf{ctx},j')\cdot z_{j'}})$ and the proof $\langle \bm{w}_{j'}, \bm{b} \rangle$, which significantly reduces the computational overhead associated with authority set updates. 

\section{Preliminaries}\label{prel}
\subsection{Assumptions}

\begin{definition}[Separable Time-Bound Generalized PS (STB-GPS) Assumption]
Given an asymmetric pairing setting $\mathcal{S} = (p, \mathbb{G}_1, \mathbb{G}_2, \mathbb{G}_T, u, v, e)$, a challenger chooses random master secrets $x, y \in \mathbb{F}_p$ and a set of independent random epoch secrets $\{z_j\}_{j=1}^T \subset \mathbb{F}_p$. The challenger initializes empty lists $\mathcal{Q}_h, \mathcal{Q}_{\text{sign}}, \mathcal{Q}_{\text{corrupt}}$. An adversary $\mathcal{A}$ is given the public parameters $(u,v, u^y, v^x, v^y)$ and access to the following oracles:

\begin{itemize}[leftmargin=*, topsep=3pt, itemsep=2pt, partopsep=0pt]
    \item \textbf{Oracle $\mathcal{O}_h(\cdot)$:}
    Outputs a uniformly distributed element $h \in \mathbb{G}_1$ and adds $h$ to a list $\mathcal{Q}_h$.

    \item \textbf{Oracle $\mathcal{O}_{\text{sign}}(j, m, h, \mathsf{ctx})$:}
    If $h \notin \mathcal{Q}_h$, or $j \in \mathcal{Q}_{\text{corrupt}}$, or the identity $(m, h, \mathsf{ctx})$ has already been signed for any epoch (i.e., $(m, h, \mathsf{ctx}, \star) \in \mathcal{Q}_{\text{sign}}$), it returns $\perp$. Otherwise, it computes:
    $$s_{\mathsf{lt}} = h^{x + m \cdot y}
    , s_{\mathsf{ep},j} = h^{F(\mathsf{ctx}, j) \cdot z_j}$$
    It adds the tuple $(m, h, \mathsf{ctx}, j, s_{\mathsf{lt}}, s_{\mathsf{ep}, j})$ to $\mathcal{Q}_{\text{sign}}$ and returns $(s_{\mathsf{lt}}, s_{\mathsf{ep},j})$.

    \item \textbf{Oracle $\mathcal{O}_{\text{update}}(j, m, h, \mathsf{ctx})$:}
    Upon input $(j, m, h, \mathsf{ctx})$, if the tuple for epoch $j$, $(m, h, \mathsf{ctx}, j, \star)$, is not in $\mathcal{Q}_{\text{sign}}$, or if a tuple for epoch $j+1$ already exists, or if the target epoch $j+1$ is corrupted (i.e., $j+1 \in \mathcal{Q}_{\text{corrupt}}$), the oracle returns $\perp$. Otherwise, it computes $s_{\mathsf{ep}, j+1} \leftarrow h^{F(\mathsf{ctx},j+1) \cdot z_{j+1}}$. It finds the corresponding long-term part $s_{\mathsf{lt}}$ from the record for epoch $j$, adds the new tuple $(m, h, \mathsf{ctx}, j+1, s_{\mathsf{lt}}, s_{\mathsf{ep},j+1})$ to $\mathcal{Q}_{\text{sign}}$, and returns $s_{\mathsf{ep}, j+1}$.

    \item \textbf{Oracle $\mathcal{O}_{\text{corrupt}}(j)$:}
    Returns the secret key $z_j$ for epoch $j$ and adds $j$ to the corruption list $\mathcal{Q}_{\text{corrupt}}$.
\end{itemize}

The STB-GPS assumption holds if for any PPT adversary $\mathcal{A}$, the probability of outputting a valid forgery tuple $(j^*, m^*, h^*, \mathsf{ctx}^*, s_{\mathsf{lt}}^*, s_{\mathsf{ep},j^*}^*)$ is negligible. A tuple is a valid forgery if it satisfies all of the following conditions:
\[
\left\{
\begin{aligned}
   & e(s_{\mathsf{lt}}^*, v) = e(h^*, v^{x + m^* y}), \quad
     e(s_{\mathsf{ep},j^*}^*, v) = e(h^*, v^{F(\mathsf{ctx}^*, j^*) z_{j^*}}) \\
   & h^* \in \mathcal{Q}_h,\;\; j^* \notin \mathcal{Q}_{\text{corrupt}} \\
   & (m^*, h^*, \mathsf{ctx}^*, j^*, s_{\mathsf{lt}}^*, s_{\mathsf{ep},j^*}^*) 
       \notin \mathcal{Q}_{\text{sign}}
\end{aligned}
\right.
\]
\end{definition}

\begin{theorem}\label{theo}
The STB-GPS assumption holds in the generic group model. After an adversary makes a total of $q$ queries to the assumption's oracles ($\mathcal{O}_h, \mathcal{O}_{\text{sign}}, \mathcal{O}_{\text{update}}, \mathcal{O}_{\text{corrupt}}$) and $q_G$ queries to the group operation oracles, its probability of producing a valid forgery is no more than $(2+ q + q_G)^2/p$, where $p$ is the prime order of the groups.
\end{theorem}

We defer the proof of the theorem to Appendix \ref{hardp}.

\subsection{Weighted Threshold Signature}
Our work builds on the weighted threshold signature (WTS) scheme by Das et al.~\cite{das2023threshold}, which is based on the inner product arguments of Campanelli et al.~\cite{campanelli2022linear}. For our purposes, we adapt the original construction with a key structural modification.
Specifically, we decompose the monolithic $\mathsf{Combine}$ algorithm into three more granular components: $\mathsf{CombPk}$ for public key aggregation, $\mathsf{CombWt}$ for weight combination, and $\mathsf{MergePf}$ for aggregating proofs.
The complete formal definition of our adapted scheme is provided in Appendix~\ref{WTS def}.

\section{Epoch-Bound Pointcheval-Sanders Signature}\label{EB-PS}

\subsection{Syntax and Security Definitions}
We now introduce the syntax of EB-PS, which extends the framework proposed by Mir et al.~\cite{DBLP:conf/ccs/MirBGLS23}. Formally, an EB-PS consists of the following algorithms:
\vspace{0.01in}

\noindent$\mathsf{Setup}(1^\lambda, \mathrm{T}) \rightarrow \mathsf{pp}$: On input the security parameter $\lambda$ and the total number of time periods $\mathrm{T}$, this algorithm outputs the public parameters $\mathsf{pp}$.
\vspace{0.05in}

\noindent$\mathsf{KGen}(\mathsf{pp}, n) \rightarrow \{\mathsf{lsk}_i, \mathsf{lvk}_i\}_{i \in [n]}$: On input the public parameters $\mathsf{pp}$ and the total number of signers $n$, the algorithm outputs a set of long-term key pairs $\{\mathsf{lsk}_i, \mathsf{lvk}_i\}_{i \in [n]}$.
\vspace{0.05in}

\noindent$\mathsf{TKGen}(\mathsf{pp}, n, j) \rightarrow \{\mathsf{tsk}_{i,j}, \mathsf{tvk}_{i,j}\}_{i \in [n]}$: On input the public parameters $\mathsf{pp}$, the total number of signers $n$, and an epoch index $j \in [\mathrm{T}]$, it outputs a set of $n$ epoch-specific key pairs where $\mathsf{tsk}_{i,j}$ denotes the signing key and $\mathsf{tvk}_{i,j}$ denotes the verification key for signer $i$ under epoch $j$.
\vspace{0.05in}

\noindent$\mathsf{GenAuxTag}(\mathrm{S}) \rightarrow (\mathsf{aux}, \mathsf{tg})$: On input a message-key set $\mathrm{S}$, the algorithm outputs auxiliary information $\mathsf{aux}$ associated with set $\mathrm{S}$ and a tag $\mathsf{tg}$.
\vspace{0.05in}

\noindent$\mathsf{SignLt}(\mathsf{lsk}_i, m_i, \mathsf{aux}, \mathsf{tg}) \rightarrow \sigma_{\mathsf{lt},i}$: On input the long-term secret key $\mathsf{lsk}_i$, message $m_i$, auxiliary data $\mathsf{aux}$, and tag $\mathsf{tg}$, the algorithm outputs a long-term signature $\sigma_{\mathsf{lt},i}$.
\vspace{0.05in}

\noindent$\mathsf{SignEp}(j, \mathsf{tg}, \mathsf{tsk}_{i,j}, F(\mathsf{ctx}, j)) \rightarrow \sigma_{\mathsf{ep},i,j}$: Under epoch $j$, given a tag $\mathsf{tg}$, an epoch-specific secret key $\mathsf{tsk}_{i,j}$, and a time-dependent function evaluation $F(\mathsf{ctx}, j)$ where $\mathsf{ctx}$ remains constant across epochs and serves solely for epoch binding, the algorithm outputs an epoch-specified signature $\sigma_{\mathsf{ep},i,j}$.
\vspace{0.05in}

\noindent$\mathsf{CombSig}(j, \mathsf{tg}, \sigma_{\mathsf{lt},i}, \sigma_{\mathsf{ep}, i, j}) \rightarrow \sigma_{i,j}$: Under epoch $j$, and the same $\mathsf{tg}$, given a long-term signature $\sigma_{\mathsf{lt},i}$ and an epoch-specified signature $\sigma_{\mathsf{ep}, i, j}$ for signer $i$, the algorithm outputs a combined signature $\sigma_{i,j}$ under epoch $j$.
\vspace{0.05in}

\noindent$\mathsf{Verify}(j, \mathsf{tg}, \mathsf{lvk}_i, \mathsf{tvk}_{i,j}, m_i, F(\mathsf{ctx}, j), \sigma_{i,j}) \rightarrow \{0,1\}$: Under epoch $j$, given a tag $\mathsf{tg}$, long-term verification key $\mathsf{lvk}_i$, an epoch-specified verification key $\mathsf{tvk}_{i,j}$, message $m_i$, the common time-dependent function evaluation $F(\mathsf{ctx}, j)$, and a combined signature $\sigma_{i,j}$ for signer $i$, the algorithm outputs 1 if the signature is valid and 0 otherwise.
\vspace{0.05in}

\noindent$\mathsf{AggSigLt}(\mathsf{tg}, \{(\mathsf{lvk}_i, m_i, \sigma_{\mathsf{lt},i})\}_{i=1}^\ell) \rightarrow (\sigma_{\mathsf{agg, lt}}, \mathbb{M}, \mathsf{avk}_{\mathsf{lt}})$: Under the same tag $\mathsf{tg}$, given a set of $\ell$ long-term signatures $\sigma_{\mathsf{lt}, i}$ for the messages $\{m_ i\}_{i \in [\ell]}$ under the verification keys $\{\mathsf{lvk}_i\}_{i \in [\ell]}$, the algorithm outputs an aggregate signature $\sigma_{\mathsf{agg, lt}}$, a message set $\mathbb{M}$, and an aggregated verification key $\mathsf{avk}_{\mathsf{lt}}$.
\vspace{0.05in}

\noindent$\mathsf{AggSigEp}(j, \mathsf{tg}, \{\mathsf{tvk}_{i,j}, \sigma_{\mathsf{ep}, i, j}\}_{i = 1}^\ell, F(\mathsf{ctx},j)) \rightarrow (\sigma_{\mathsf{agg, ep},j}, \mathsf{avk}_{\mathsf{ep},j})$: Under epoch $j$ and the same $\mathsf{tg}$, given a set of $\ell$ epoch-specified signatures $\sigma_{\mathsf{ep}, i, j}$ under the verification keys $\{\mathsf{tvk}_{i,j}\}_{i \in [\ell]}$ with respect to the common time-dependent function evaluation $F(\mathsf{ctx},j)$, the algorithm outputs an aggregate epoch signature $\sigma_{\mathsf{agg, ep},j}$ and an aggregated epoch verification key $\mathsf{avk}_{\mathsf{ep},j}$.
\vspace{0.05in}

\noindent$\mathsf{CombAggSig}(j, \mathsf{tg}, \sigma_{\mathsf{agg, lt}}, \sigma_{\mathsf{agg, ep},j}) \rightarrow \sigma_{\mathsf{agg}, j}$: Under epoch $j$ and the same $\mathsf{tg}$, given an aggregated long-term signature $\sigma_{\mathsf{agg, lt}}$ and an aggregated epoch signature $\sigma_{\mathsf{agg, ep},j}$, the algorithm outputs a combined aggregate signature $\sigma_{\mathsf{agg}, j}$.
\vspace{0.05in}

\noindent$\mathsf{AggVerify}(j, \mathsf{tg}, \mathsf{avk}_j, \mathbb{M}, \sigma_{\mathsf{agg},j}, F(\mathsf{ctx}, j)) \rightarrow \{0,1\}$: Under epoch $j$, given a tag $\mathsf{tg}$, an aggregated verification key $\mathsf{avk}_j = (\mathsf{avk}_{\mathsf{lt}}, \mathsf{avk}_{\mathsf{ep},j})$, a message set $\mathbb{M}$, a combined aggregate signature $\sigma_{\mathsf{agg},j}$, and the common time-dependent function evaluation $F(\mathsf{ctx}, j)$, the algorithm outputs 1 if the aggregate signature is valid and 0 otherwise.
\vspace{0.05in}

\noindent$\mathsf{RndSigTag}(\mathsf{avk}_j, \mathsf{tg}, \sigma_{\mathsf{agg}, j}, r) \rightarrow (\sigma_{\mathsf{agg},j}', \mathsf{tg}')$: Under aggregated verification key $\mathsf{avk}_j$, given a tag $\mathsf{tg}$, an aggregate signature $\sigma_{\mathsf{agg}, j}$, and a random value $r$, the algorithm outputs a randomized aggregate signature $\sigma_{\mathsf{agg},j}'$ and a randomized tag $\mathsf{tg}'$.
\vspace{0.05 in}

\noindent\textbf{Correctness.} EB-PS ensures both basic and aggregation correctness, as formalized in Appendix \ref{def corr}.

\noindent\textbf{Unforgeability}. To formalize the security of our epoch-based scheme, we work within the chosen-key model of security, following the line of work in \cite{boneh2003aggregate,lysyanskaya2004sequential, DBLP:conf/ccs/MirBGLS23}.  Our security notion, Existential Unforgeability under chosen-message and Epoch-corruption Attack (EUF-eCMA), by building an interactive game where the adversary is given a single challenge public key $\mathsf{vk}'$ and access to a corresponding signing oracle. 
\vspace{-1em}
\begin{definition}[Existential Unforgeability under chosen-message and Epoch-corruption Attack (EUF-eCMA)]
An epoch-bound digital signature scheme $\Sigma = (\mathsf{Setup}, \mathsf{KGen}, \mathsf{TKGen}, \mathsf{SignLt}, \mathsf{SignEp}, \mathsf{AggVerify})$ is EUF-eCMA secure if for any PPT adversary $\mathcal{A}$, the probability of winning the following game is negligible.
\end{definition}
\vspace{-1em}
\noindent\textit{Setup.} The challenger $\mathcal{C}$ runs $\mathsf{pp} \leftarrow \mathsf{Setup}(1^\lambda, \mathrm{T})$ to generate the public parameters. Then, $\mathcal{C}$ generates a long-term key pair $(\mathsf{lvk}^*, \mathsf{lsk}^*) \leftarrow \mathsf{KGen}(\mathsf{pp})$, where the secret key $\mathsf{lsk}^*$ is never revealed. For each epoch $j \in [1, \mathrm{T}]$, $\mathcal{C}$ computes the epoch-specified key pair $(\mathsf{tvk}^*_{j}, \mathsf{tsk}^*_{j}) \leftarrow \mathsf{TKGen}(\mathsf{pp}, j)$. $\mathcal{C}$ initializes an empty query log $\mathcal{Q} \leftarrow \emptyset$ and an empty set of corrupted epochs $\mathcal{Q}_{CE} \leftarrow \emptyset$. Finally, $\mathcal{C}$ provides the adversary $\mathcal{A}$ with the public parameters $\mathsf{pp}$, the long-term public key $\mathsf{lvk}^*$, and all epoch public keys $\{\mathsf{tvk}^*_j\}_{j \in [1, \mathrm{T}]}$.
    
    \noindent\textit{Queries.} The adversary $\mathcal{A}$ can adaptively issue the following queries:
    \begin{itemize}[nosep,leftmargin=*]
        \item \textbf{Initial Signing Query:} On input an epoch $j \in [1, \mathrm{T}]$, a tag $\mathsf{tg}$, a message $m$, auxiliary info $\mathsf{aux}$, and a context $\mathsf{ctx}$:

        \begin{itemize}[nosep, leftmargin=*]
    \item If epoch $j \in \mathcal{Q}_{CE}$, return $\perp$.
    \item Else compute $\sigma_{\mathsf{lt}} \leftarrow \mathsf{SignLt}(\mathsf{lsk}^*, m, \mathsf{aux}, \mathsf{tg})$ and $\sigma_{\mathsf{ep}, j} \leftarrow \mathsf{SignEp}(j, \mathsf{tg}, \mathsf{tsk}^*_j, F(\mathsf{ctx}, j))$.
    \item Return $(\sigma_{\mathsf{lt}}, \sigma_{\mathsf{ep}, j})$ to $\mathcal{A}$ and update $\mathcal{Q} \leftarrow \mathcal{Q} \cup \{(j, \mathsf{tg}, m)\}$.
\end{itemize}

        \item \textbf{Epoch Update Query}: For a transition from epoch $j$ to $j+1$ ($j < \mathrm{T}$), on input a tag $\mathsf{tg}$ and context $\mathsf{ctx}$:
    \begin{itemize}[nosep,leftmargin=*]
    \item If $j+1 \in \mathcal{Q}_{CE}$, return $\perp$.
    \item Otherwise, compute $\sigma_{\mathsf{ep}, j+1} \leftarrow \mathsf{SignEp}(j+1, \mathsf{tg}, \mathsf{tsk}^*_{j+1}, F(\mathsf{ctx}, j+1))$ and return it to $\mathcal{A}$.
\end{itemize}

    \item \textbf{Epoch Corruption Query}: On input an epoch index $j \in [1, \mathrm{T}]$, add $j$ to $\mathcal{Q}_{CE}$ and return $\mathsf{tsk}^*_{j}$.

    \end{itemize}
        
    \noindent\textit{Winning Condition.} The adversary $\mathcal{A}$ outputs a forgery tuple 
    $(j^*, \mathsf{tg}^*, \mathbb{M}^*, \sigma_{\mathsf{agg}, j^*}^*, \mathsf{avk}_{j^*})$, where $\mathbb{M}^* = \{m_k^*\}_{k=1}^\ell$.  $\mathcal{A}$ wins if all of the following conditions hold:
    \[
\left\{
\begin{aligned}
   & \mathsf{AggVerify}(j^*, \mathsf{tg}^*, \mathsf{avk}_{j^*}, 
      \mathbb{M}^*, \sigma_{\mathsf{agg},j^*}^*, F(\mathsf{ctx}, j^*)) = 1 \\[4pt]
   & (\mathsf{lvk}^*, \mathsf{tvk}^*_{j^*}) \in \mathsf{avk}_{j^*},
     \quad j^* \notin \mathcal{Q}_{CE} \\[4pt]
   & \exists m_k^* \in \mathbb{M}^* \;\; \text{s.t.}\;\;
     (j^*, m_k^*, \mathsf{tg}^*) \notin \mathcal{Q}
\end{aligned}
\right.
\]

\subsection{The EB-PS Construction}\label{ebpscons}
In this subsection, we present our EB-PS construction. Inspired by \cite{crites2023threshold,DBLP:conf/ccs/MirBGLS23}, we derive the auxiliary data \( \mathsf{aux} \) by hashing relevant information to ensure a uniform \( h \) component. Following \cite{DBLP:conf/ccs/MirBGLS23}, we adopt a unique tag \( \mathsf{tg} \) for partial aggregation, requiring all partial signatures to share the same tag. The EB-PS scheme proceeds as follows:
\vspace{0.05in}

\noindent$\mathsf{Setup}(1^\lambda, \mathrm{T}) \rightarrow \mathsf{pp}$: On input a security parameter $\lambda$ and the number of epochs $\mathrm{T}$, this algorithm generates bilinear groups $\mathsf{BG} = (p, \mathbb{G}_1, \mathbb{G}_2, \mathbb{G}_T, u, v, e) \leftarrow \mathsf{BGGen}(1^\lambda)$, where $p$ is a prime order, $u$ is a generator of $\mathbb{G}_1$, and $v$ is a generator of $\mathbb{G}_2$. Then, it selects a hash function $\mathrm{H}: \{0,1\}^* \rightarrow \mathbb{G}_1$ and outputs the public parameters $\mathsf{pp} = \{\mathsf{BG}, \mathrm{T}, \mathrm{H}\}$.

\vspace{0.05in}
\noindent$\mathsf{KGen}(\mathsf{pp}, n) \rightarrow \{\mathsf{lsk}_i, \mathsf{lvk}_i\}_{i \in [n]}$: For the $i$-th signer, this algorithm randomly samples $(x_i, y_i) \in \mathbb{F}_p^2$, sets the long-term secret key as $\mathsf{lsk}_i = (x_i, y_i)$ and computes the corresponding verification key as $\mathsf{lvk}_i = (\mathrm{X}_i, \mathrm{Y}_i) = (v^{x_i}, v^{y_i})$.

\vspace{0.05in}
\noindent$\mathsf{TKGen}(\mathsf{pp}, n, j) \rightarrow \{\mathsf{tsk}_{i,j}, \mathsf{tvk}_{i,j}\}_{i \in [n]}$: For epoch $j \in [\mathrm{T}]$ and the $i$-th signer, this algorithm randomly samples $z_{i,j} \in \mathbb{F}_p$, sets the epoch-specific secret key as $\mathsf{tsk}_{i,j} = z_{i,j}$ and computes the corresponding verification key as $\mathsf{tvk}_{i,j} = \mathrm{Z}_{i,j} = v^{z_{i,j}}$.

\vspace{0.05in}
\noindent$\mathsf{GenAuxTag}(\mathrm{S}) \rightarrow (\mathsf{aux}, \mathsf{tg})$: Given a message-key set $\mathrm{S} = \{(m_i,\mathsf{lvk}_i)_{i\in[\ell]}\}$, the algorithm randomly samples $(\gamma, \delta) \in \mathbb{F}_p^2$, sets $\mathsf{aux} =u^{\gamma} \parallel u^{\delta} \parallel \{(m_i, \mathsf{lvk}_i)\}_{i \in [\ell]}$, computes $h = \mathsf{H}(\mathsf{aux})$, and outputs auxiliary data $\mathsf{aux}$ and tag $\mathsf{tg} = (\Gamma
=h^{\gamma}, \Delta = h^{\delta})$, where all verification keys $\mathsf{lvk}_i$ must be distinct.

\vspace{0.05in}
\noindent$\mathsf{SignLt}(\mathsf{lsk}_i, m_i, \mathsf{aux}, \mathsf{tg}) \rightarrow \sigma_{\mathsf{lt},i}$: Given a long-term secret key $\mathsf{lsk}_i = (x_i, y_i)$ for the $i$-th signer, message $m_i$, auxiliary data $\mathsf{aux}$, and tag $\mathsf{tg}$, the algorithm checks that $\mathsf{aux}$ has the form $(u^\gamma \parallel u^\delta \parallel \{(m_i, \mathsf{lvk}_i)\}_{i \in [\ell]})$, and verifies that $(m_i, \mathsf{lvk}_i) \in \mathsf{aux}$; if not, it outputs $\bot$. Otherwise, the algorithm computes $\sigma_{\mathsf{lt},i} = (h', s_{\mathsf{lt}, i})$ where:
$$h' = h^{\gamma}, \quad s_{\mathsf{lt}, i} = (h^{\gamma})^{x_i + m_i \cdot y_i}$$

\noindent$\mathsf{SignEp}(j, \mathsf{tg}, \mathsf{tsk}_{i,j}, F(\mathsf{ctx}, j)) \rightarrow \sigma_{\mathsf{ep},i,j}$: Under the epoch $j$, given a tag $\mathsf{tg}$, an epoch-specific secret key $\mathsf{tsk}_{i,j} = z_{i,j}$ for the $i$-th signer, and a time-dependent function evaluation $F(\mathsf{ctx}, j)$, the algorithm outputs:
$$\sigma_{\mathsf{ep},i,j} = s_{\mathsf{ep},i,j} = (h^{\delta})^{F(\mathsf{ctx}, j) \cdot {z_{i,j}}}$$

\noindent$\mathsf{CombSig}(j, \mathsf{tg}, \sigma_{\mathsf{lt},i}, \sigma_{\mathsf{ep}, i, j}) \rightarrow \sigma_{i,j}$: Under epoch $j$ and the same tag $\mathsf{tg}$, given a long-term signature $\sigma_{\mathsf{lt},i}$, and an epoch-specific signature $\sigma_{\mathsf{ep},i,j}$, the algorithm outputs a combined signature $\sigma_{i,j}$ computed as:
$$\sigma_{i,j} = (h', s_i = s_{\mathsf{lt},i} \cdot s_{\mathsf{ep},i,j})$$

\noindent$\mathsf{Verify}(j, \mathsf{tg}, \mathsf{lvk}_i, \mathsf{tvk}_{i,j}, m_i, F(\mathsf{ctx}, j), \sigma_{i,j}) \rightarrow \{0,1\}$: Under epoch $j$, given a $\mathsf{tg}$, a long-term verification key $\mathsf{lvk}_i = (\mathrm{X}_i, \mathrm{Y}_i)$ and an epoch-specific verification key $\mathsf{tvk}_{i,j} = \mathrm{Z}_{i,j}$ for the $i$-th signer, messages $m_i$ and $F(\mathsf{ctx},j)$, and a signature $\sigma_{i,j}$, this algorithm parses $\sigma_{i,j}$ as $(h', s_i)$ and outputs 1 if the equation holds:
\[e(h', \mathrm{X}_i \cdot \mathrm{Y}_i^{m_i})e((h^{\delta})^{F(\mathsf{ctx}, j)},\mathrm{Z}_{i, j}) = 
e(s_i,v) \wedge  h' \neq 1_{\mathbb{G}}\]

\noindent$\mathsf{AggSigLt}(\mathsf{tg}, \{(\mathsf{lvk}_i, m_i, \sigma_{\mathsf{lt},i})\}_{i=1}^\ell) \rightarrow (\sigma_{\mathsf{agg, lt}}, \mathbb{M}, \mathsf{avk}_{\mathsf{lt}})$:  Under the same $\mathsf{tg}$, given a set of $\ell$ long-term signatures $\sigma_{\mathsf{lt}, i}$ for the messages $\{m_{1, i}\}_{i \in [\ell]}$ under the verification keys $\{\mathsf{lvk}_i\}_{i \in [\ell]}$, this algorithm outputs the set of messages $\mathbb{M} = \{m_i\}_{i=1}^\ell$, the aggregated verification key $\mathsf{avk}_{\mathsf{lt}} = \{\mathsf{lvk}_i\}_{i \in [\ell]}$, and an aggregate long-term signature: $$\sigma_{\mathsf{agg, lt}} = (h', \prod_{i =1 }^\ell s_{\mathsf{lt},i})$$

\noindent$\mathsf{AggSigEp}(j, \mathsf{tg}, \{\mathsf{tvk}_{i,j}, \sigma_{\mathsf{ep}, i, j}\}_{i = 1}^\ell, F(\mathsf{ctx},j)) \rightarrow (\sigma_{\mathsf{agg, ep},j}, \mathsf{avk}_{\mathsf{ep},j})$: Under epoch $j$ and the same $\mathsf{tg}$, given a set of $\ell$ epoch-specified signatures $\sigma_{\mathsf{ep}, i, j}$ under the verification keys $\{\mathsf{tvk}_{i,j}\}_{i \in [\ell]}$ with respect to the common time-dependent function evaluation $F(\mathsf{ctx},j)$, this algorithm outputs the aggregated epoch verification key $\mathsf{avk}_{\mathsf{ep},j} = \{\mathsf{tvk}_{i,j}\}_{i \in [\ell]}$ and an aggregate epoch signature: $$\sigma_{\mathsf{agg, ep}, j} = \prod_{i=1}^\ell s_{\mathsf{ep},i,j}$$

\noindent$\mathsf{CombAggSig}(j, \mathsf{tg}, \sigma_{\mathsf{agg,lt}}, \sigma_{\mathsf{agg,ep},j}) \rightarrow \sigma_{\mathsf{agg}, j}$: Under epoch $j$ and the same $\mathsf{tg}$, given an aggregated long-term signature $\sigma_{\mathsf{agg, lt}}$ and an aggregated epoch signature $\sigma_{\mathsf{agg,ep},j}$, this algorithm combines them to output an aggregate signature $\sigma_{\mathsf{agg}, j}$, where:
\[\sigma_{\mathsf{agg}, j} = \left(h',s = \sigma_{\mathsf{agg, lt}} \cdot \sigma_{\mathsf{agg,ep},j}\right)\]

\noindent$\mathsf{AggVerify}(j, \mathsf{tg}, \mathsf{avk}_j, \mathbb{M}, \sigma_{\mathsf{agg},j}, F(\mathsf{ctx}, j)) \rightarrow \{0,1\}$: Under epoch $j$, given an aggregate verification key $\mathsf{avk} = \{(\mathsf{lvk}_i, \mathsf{tvk}_{i,j})\}_{i=1}^\ell$ where $\mathsf{lvk}_i = (\mathrm{X}_i, \mathrm{Y}_i)$ and $\mathsf{tvk}_{i,j} = \mathrm{Z}_{i,j}$, a message set $\mathbb{M} = \{m_i\}_{i=1}^\ell$, a time-dependent function evaluation $F(\mathsf{ctx},j)$, and an aggregate signature $\sigma_{\mathsf{agg}, j} = (h', s)$, outputs 1 if $ h' \neq 1_{\mathbb{G}}$ and the following equation holds:
\[\small e\left(h', \prod_{i=1}^\ell \mathrm{X}_i \cdot \mathrm{Y}_i^{m_i}\right) \cdot 
e\left((h^{\delta})^{F(\mathsf{ctx},j)}, \prod_{i=1}^\ell \mathrm{Z}_{i,j}\right) = e(s,v) 
\]

\noindent$\mathsf{RndSigTag}(\mathsf{avk}_j, \mathsf{tg}, \sigma_{\mathsf{agg}, j}, r) \rightarrow (\sigma_{\mathsf{agg},j}', \mathsf{tg}')$: Under a given aggregated verification key $\mathsf{avk}_j$ and with a tag $\mathsf{tg}$, an aggregate signature $\sigma_{\mathsf{agg}, j}$, and a random value $r$, the algorithm simultaneously randomizes the signature and tag as:$$\sigma'_{\mathsf{agg}, j} = ((h')^r, s^r),  \quad \mathsf{tg}' = ((h^{\gamma})^r, (h^{\delta})^r)$$

\noindent\textbf{Correctness}. We refer readers to Appendix \ref{proof of corr} for the correctness.

\noindent\textbf{Unforgeability.} The formal security proof of unforgeability for our EB-PS scheme is presented in Appendix \ref{unforge of ebps}.
\begin{theorem}\label{unfeb}
	The EB-PS scheme is EUF-eCMA secure under the STB-GPS assumption in the random oracle model.
\end{theorem}

\section{Multi-Authority Anonymous Credentials with Epoch-Based Weights}\label{MA-ACEW}
In this section, we present Multi-Authority Anonymous Credentials with Epoch-Based Weights (MA-ACEW), the first AC scheme supporting weighted trust across distributed authorities. Existing multi-authority schemes \cite{sonnino2018coconut, hebant2023traceable, DBLP:conf/ccs/MirBGLS23} treat issuers uniformly, lacking the ability to capture varying trust levels. The weighted threshold signature (WTS) scheme of Das et al.~\cite{das2023threshold} suggests a direction but depends on BLS signatures, making it unsuitable for standard ACs and dynamic weight updates. MA-ACEW overcomes this by integrating a modified WTS with our EB-PS framework.
\subsection{Formal Definition}\label{madef}
Suppose there are $n$ credential issuers $\mathsf{CI}_i$ $(i \in [n])$. In the first epoch, each $\mathsf{CI}_i$ generates long-term keys $(\mathsf{lsk}_i,\mathsf{lvk}_i)$ via $\mathsf{KGen}$ and epoch-specific keys $(\mathsf{tsk}_{i,1},\mathsf{tvk}_{i,1})$ via $\mathsf{TKGen}$, together with an initial weight vector $\bm{w}_1=(w_{1,1},\ldots,w_{n,1})$. In epoch $j$, users interacting with $\mathsf{CI}_i$ obtain a long-term credential $\mathsf{cred}_{\mathsf{lt},i}$ signed under $\mathsf{lsk}_i$ and an epoch-specific credential $\mathsf{cred}_{\mathsf{ep},i,j}$ signed under $\mathsf{tsk}_{i,j}$.

At transition $j \rightarrow j+1$, the weight vector updates to $\bm{w}_{j+1}$ (based on external factors), each issuer erases $(\mathsf{tsk}_{i,j},\mathsf{tvk}_{i,j})$, and generates $(\mathsf{tsk}_{i,j+1},\mathsf{tvk}_{i,j+1})$. Using the new key, $\mathsf{CI}_i$ non-interactively derives fresh epoch-specific credentials $\mathsf{cred}_{\mathsf{ep},i,j+1}$ for all users with valid long-term credentials, ensuring uniqueness via tag $\mathsf{tg}$. In the context of credentials, we treat messages as attributes, while retaining the original notation \(m\).

\vspace{-1em}
\begin{definition}
	\noindent \textbf{(MA-ACEW).} A Multi-Authority Anonymous Credentials with Epoch-Based Weights (MA-ACEW) is defined by the following algorithms/protocols:
\end{definition}
\vspace{-1em}
\noindent$\mathsf{Setup}$: On input a security parameter $\lambda$, output public parameter $\mathsf{pp}$.
\vspace{0.05in}

\noindent$\mathsf{IniEpKGen}$: Initialize the system by generating epoch-1 state and key materials for each issuer $\{\mathsf{CI}_i\}_{i \in [n]}$, including long-term key pair $(\mathsf{lsk}_i, \mathsf{lvk}_i) \leftarrow \mathsf{KGen}(1^\lambda)$, epoch-specific key pair $(\mathsf{tsk}_{i,1}, \mathsf{tvk}_{i,1}) \leftarrow \mathsf{TKGen}(1^\lambda)$ for credential issuance, and aggregation key $\mathsf{ak}_i$ for inner-product argument generation. A global verification key $\mathsf{vk}$ is derived from all issuers' aggregation keys, along with the initial weight vector $\bm{w}_1 \in \mathbb{F}_p^n$.

\vspace{0.05in}
\noindent$\mathsf{UKGen}$: Given a message-key space $\mathrm{S}$, generate a user key pair $(\mathsf{usk}, \mathsf{uvk})$ as the user's identity along with the auxiliary information $\mathsf{aux}$.
\vspace{0.05in}

\noindent$\mathsf{Issuance}$: During epoch $j$, each user performs a one-time interaction with credential issuer $\mathsf{CI}_i$ to obtain a credential tuple $(\mathsf{cred}_{\mathsf{lt}, i}, \mathsf{cred}_{\mathsf{ep}, i, j})$ consisting of a long-term credential $\mathsf{cred}_{\mathsf{lt}, i}$ and an epoch-specific credential $\mathsf{cred}_{i, j}$, where:
\begin{equation*}
\begin{split}
    &\langle\mathsf{CredObtain}(\mathsf{tg}, \mathsf{aux}, \mathbb{M}) \leftrightarrow \mathsf{CredIssue}(j, \mathsf{lsk}_i, \mathsf{tsk}_{i, j})\rangle \\
    & \qquad \quad \quad \quad \quad \rightarrow \mathsf{cred}_{\mathsf{lt}, i}, \mathsf{cred}_{\mathsf{ep}, i, j}
\end{split}
\end{equation*}

\noindent$\mathsf{EpUpdate}$: 
To initiate epoch \(j+1\), the system first updates the weight vector from \(\bm{w}_j\) to \(\bm{w}_{j+1}\). 
Then, each \(\mathsf{CI}_i\) generates a new, epoch-specific key pair \((\mathsf{tsk}_{i, j+1}, \mathsf{tvk}_{i, j+1})\).
Subsequently, for every user holding a valid long-term credential \(\mathsf{cred}_{\mathsf{lt}}\), \(\mathsf{CI}_i\) issues a fresh epoch credential \(\mathsf{cred}_{\mathsf{ep}, i, j+1}\), computed as follows:
\[
    \mathsf{CredIssueEp}(\mathsf{tsk}_{i, j+1}, \mathsf{tg}, j+1) \rightarrow 
    \mathsf{cred}_{\mathsf{ep}, i, j+1} 
\]

\noindent$\mathsf{Gen\text{-}Policies:}$
For clarity, we first introduce our protocol within a simplified framework. In this model, the verifier's policy is constrained to two components: the current epoch $j+1$ and a minimum weight threshold $\mathrm{W}_{\mathsf{acc}}$, denoted as $\mathsf{pol} = (j+1, \mathrm{W}_{\mathsf{acc}})$. To satisfy this policy, the user must present a set of certified attributes $\mathbb{M}$ whose cumulative weight meets or exceeds $\mathrm{W}_{\mathsf{acc}}$ for epoch $j+1$. Our construction readily extends to support a more expressive policy framework, 
enabling verifiers to specify arbitrary attribute disclosure requirements and complex predicates over hidden attributes. We defer the full details of this extension to Section~\ref{secadd}.
\vspace{0.05in} 

\noindent$\mathsf{Show}$: To satisfy a verifier's policy $\mathsf{pol} = (j+1, \mathrm{W}_{\mathsf{acc}})$, a user holding an aggregated credential $\mathsf{cred}_{\mathsf{agg},j+1}$ and a corresponding tag $\mathsf{tg}$ engages in an interactive protocol with the verifier. Here, $\mathsf{tg}$ serves as the user's pseudonymous identity. The interaction involves the user providing a zero-knowledge proof $\pi$ and disclosing a set of certified attributes $\mathbb{M} = \{m_k\}_{k \in [\ell]}$. The proof $\pi$ demonstrates that the credential's aggregated weight $\mathrm{W}$ satisfies the policy threshold, i.e., $\mathrm{W} \ge \mathrm{W}_{\mathsf{acc}}$. 
\begin{equation}
	\left\langle
	\begin{array}{l}
		\mathsf{CredShow}(\mathsf{tg}, \mathsf{cred}_{\mathsf{agg}, j+1}, \mathbb{M}, \pi) \leftrightarrow  \\
		\mathsf{CredVerify}(\mathsf{pp}, \mathbb{M}, \mathsf{pol}, \{\mathsf{lvk}_i, \mathsf{tvk}_{i,j+1}\}_{i \in [\ell]}, \pi)
	\end{array}
	\right\rangle \rightarrow \{0,1\} \nonumber
\end{equation}

Here, we assume the full set \(\mathbb{M}\) is disclosed for simplicity.
In Section~\ref{secadd}, we extend our construction to support flexible selective disclosure, 
which is achieved by partitioning the attribute set into disclosed (\(\mathbb{M}_{\mathcal{D}}\)) and hidden (\(\mathbb{M}_{\mathcal{H}}\)) subsets.
As a final simplification for this presentation, we also denote the time function as \(j\), using only the epoch index.

\subsection{Secuirty Definition}\label{sec52}
We formalize security in a game-based model adapted from prior work \cite{fuchsbauer2019structure, hebant2023traceable, DBLP:conf/ccs/MirBGLS23}. 
To capture the dynamics of epoch-based weights, we extend this model with three new oracles: $\mathcal{O}^{\mathsf{UpdEp}}$, $\mathcal{O}^{\mathsf{ObtIssEp}}$, and $\mathcal{O}^{\mathsf{IssEp}}$.
The full formalization, detailing the game mechanics and oracle interactions, is deferred to Appendix \ref{oracles}.

\vspace{0.01in}
\noindent\textbf{Correctness.} The correctness ensures that a credential showing for a non-empty attribute subset $\mathbb{M}$ always verifies successfully when the credential was honestly issued for an attribute set $ \{m_i\}_{i \in [\ell]}$, and the aggregate weight of the attributes in $\mathbb{M}$ satisfies the threshold specified in the policy $\mathsf{pol}$. The aggregated credential must be newly issued or properly updated within the current epoch.

\vspace{0.01in}
\noindent\textbf{Unforgeability.} The unforgeability asserts that no adversary $\mathcal{A}$ can produce a valid aggregated credential $\mathsf{cred}_{\mathsf{agg},j}$ with non-negligible probability, without possessing the required credentials from the set of accepted issuers $\mathsf{CI} = \{\mathsf{lvk}_i, \mathsf{tvk}_{i,j}\}_{i \in [n]}$. The credential must show for a policy $\mathsf{pol}$ and a set of disclosed attributes $\mathbb{M}$. Here, $\mathcal{A}$ can obtain $(\mathsf{lvk}_i, \mathsf{tvk}_{i,j}) \in \mathsf{CI}$ through the oracles $\mathcal{O}^{\mathsf{HCI}}(i)$ and $\mathcal{O}^{\mathsf{CCI}}(i)$. Furthermore, all credentials used in the aggregation must be updated to the current epoch to ensure the validity of the showing.

\vspace{0.01in}
\noindent\textbf{Anonymity.} The anonymity ensures that no adversary $\mathcal{A}$, acting as a malicious verifier, can distinguish between two users with non-negligible advantage. Furthermore, different showings of the same credential should be computationally unlinkable. In the security game, the adversary has adaptive access to an oracle that, on the input of two distinct user indexes $id_0$ and $id_1$, acts as one of the two credential owners (depending on bit $b$) in the verification. 

\vspace{0.01in}
\noindent\textbf{Blindness.} The blindness prevents the issuer from learning which user receives which credential, thus preserving unlinkability between issuance and presentation. Formally, it guarantees that a malicious issuer, $\mathcal{A}$, cannot distinguish between two honest users during the issuance protocol, thereby severing the link between a credential's issuance and its subsequent presentation. 
\vspace{-1em}
\begin{definition}[Unforgeability]
	The unforgeability is defined by the security game in Fig.\ref{Unforgeability}. MA-ACEW is unforgeable if for any PPT adversary $\mathcal{A}$, there exists a negligible function $\epsilon(\lambda)$ satisfies:
     $$\textup{Adv}_{\mathcal{A}}^{\mathsf{EU\text{-}CMA}}= \left\lvert \textup{Pr}[\mathsf{Exp}^{\mathsf{UNF}}_{\mathsf{MA\text{-}ACEW},\mathcal{A}}(\lambda) = 1] \right\rvert \leq \epsilon(\lambda)$$
   
\end{definition} 

\begin{figure}[h]
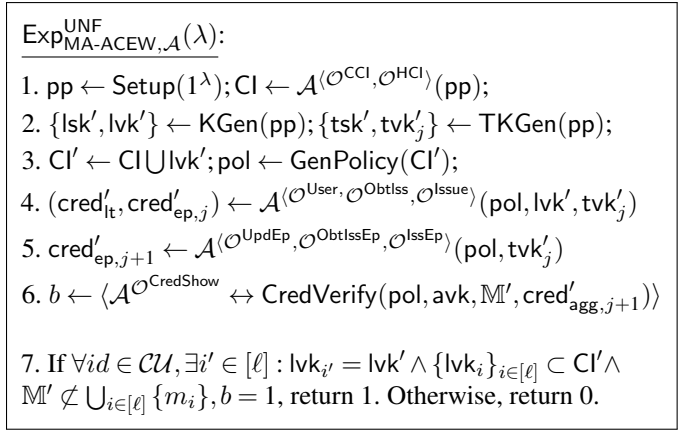

	\centering
	\fbox{
		\begin{minipage}{\linewidth}
			\vspace{0.3em}
			$\underline{\mathsf{Exp}^{\mathsf{UNF}}_{\mathsf{MA\text{-}ACEW},\mathcal{A}}(\lambda)}$:
			\vspace{0.2em}
			
			1. $\mathsf{pp} \leftarrow\mathsf{Setup}(1^\lambda); \mathsf{CI} \leftarrow \mathcal{A}^{\langle\mathcal{O}^{\mathsf{CCI}}, \mathcal{O}^{\mathsf{HCI}}\rangle}(\mathsf{pp});$
			\vspace{0.2em}
			
			2. $\{\mathsf{lsk}', \mathsf{lvk}'\} \leftarrow \mathsf{KGen}(\mathsf{pp}); \{\mathsf{tsk}', \mathsf{tvk}_j'\} \leftarrow \mathsf{TKGen}(\mathsf{pp});$
			\vspace{0.2em}
			
			3. $\mathsf{CI}' \leftarrow \mathsf{CI} \bigcup \mathsf{lvk}'; \mathsf{pol} \leftarrow \mathsf{GenPolicy}(\mathsf{CI}');$
			\vspace{0.2em}
			
			4. $(\mathsf{cred}_{\mathsf{lt}}', \mathsf{cred}'_{\mathsf{ep},j}) \leftarrow \mathcal{A}^{\langle\mathcal{O}^{\mathsf{User}, }\mathcal{O}^{\mathsf{Obtlss}}, \mathcal{O}^{\mathsf{Issue}}\rangle}(\mathsf{pol}, \mathsf{lvk}', \mathsf{tvk}_j')$
			\vspace{0.2em}
			
			5. $\mathsf{cred}_{\mathsf{ep}, j+1}' \leftarrow \mathcal{A}^{\langle\mathcal{O}^{\mathsf{UpdEp}}, \mathcal{O}^{\mathsf{ObtIssEp}}, \mathcal{O}^{\mathsf{IssEp}}\rangle}(\mathsf{pol}, \mathsf{tvk}_j')$
			\vspace{0.2em}
			
			6. $b \leftarrow \langle \mathcal{A}^{\mathcal{O}^{\mathsf{CredShow}}} \leftrightarrow \mathsf{\mathsf{CredVerify}}(\mathsf{pol}, \mathsf{avk}, \mathbb{M}', \mathsf{cred}_{\mathsf{agg}, j+1}')\rangle$
			\vspace{0.2em}
			
			7. If $\forall id \in \mathcal{CU}, \exists i' \in [\ell]: \mathsf{lvk}_{i'} = \mathsf{lvk}' \land \{\mathsf{lvk}_i\}_{i \in [\ell]} \subset \mathsf{CI}' \land$
			\\ \hspace{1em} $\mathbb{M}' \not\subset \bigcup_{i \in [\ell]} \mathcal\{{m}_{i}\}, b = 1$, return 1. Otherwise, return 0.
            \vspace{0.2em}
		\end{minipage}
	}
	\caption{Unforgeability Security Game}\label{Unforgeability}		
\end{figure}

\vspace{-1em}
\begin{definition}[Anonymity]
	The anonymity is defined by the security game in Fig.\ref{Anonymity}. MA-ACEW is anonymous if for any PPT adversary $\mathcal{A}$, there exists a negligible function $\epsilon(\lambda)$ satisfies:
	
    $$\textup{Adv}_{\mathcal{A}}^{\mathsf{ANO}\text{-}b}= \left \lvert \textup{Pr}[\mathsf{Exp}^{\mathsf{ANO}\text{-}b}_{\mathsf{MA\text{-}ACEW},\mathcal{A}}(\lambda)] - \frac{1}{2} \right \rvert \leq \epsilon(\lambda)$$

\end{definition}
\vspace{-1em}
\begin{figure}[h]
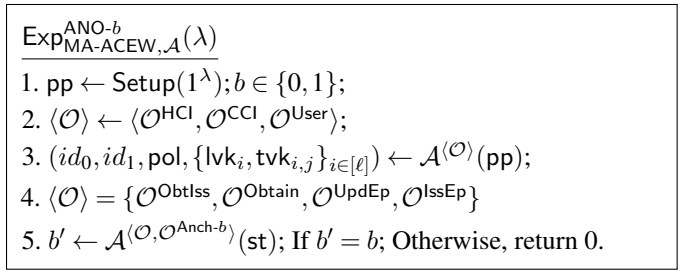

	\centering
	\fbox{
		\begin{minipage}{\linewidth}
			\vspace{0.3em}
			$\underline{\mathsf{Exp}^{\mathsf{ANO}\text{-}b}_{\mathsf{MA\text{-}ACEW},\mathcal{A}}(\lambda)}$
			\vspace{0.2em}
			
			1. $\mathsf{pp} \leftarrow \mathsf{Setup}(1^\lambda); b \in \{0,1\};$
			\vspace{0.2em}
			
			2. $\langle \mathcal{O} \rangle \leftarrow  \langle \mathcal{O}^{\mathsf{HCI}}, \mathcal{O}^{\mathsf{CCI}}, \mathcal{O}^{\mathsf{User}}\rangle;$
			\vspace{0.2em}
			
			3. $(id_0, id_1, \mathsf{pol}, \{\mathsf{lvk}_i, \mathsf{tvk}_{i,j}\}_{i \in [\ell]}) \leftarrow \mathcal{A}^{\langle \mathcal{O} \rangle}(\mathsf{pp});$
			\vspace{0.2em}
			
			4. $\langle\mathcal{O}\rangle = \{\mathcal{O}^{\mathsf{ObtIss}}, \mathcal{O}^{\mathsf{Obtain}}, \mathcal{O}^{\mathsf{UpdEp}}, \mathcal{O}^{\mathsf{IssEp}}\}$
			\vspace{0.2em}
			
			5. $b' \leftarrow \mathcal{A}^{\langle \mathcal{O}, \mathcal{O}^{\mathsf{Anch}\text{-}b}\rangle}(\mathsf{st});$ If $b' = b$; Otherwise, return 0.
			\vspace{0.2em}
		\end{minipage}
	}
	\caption{Anonymity Security Game}\label{Anonymity}
	
\end{figure}
\begin{definition}[Blindness]
	The blindness is defined by the security game in Fig.\ref{Blindness}. MA-ACEW is blind if for any PPT adversary $\mathcal{A}$, there exists a negligible function $\epsilon(\lambda)$ satisfies:
	\vspace{-0.5em}
    $$\textup{Adv}_{\mathcal{A}} ^{\mathsf{BLI}\text{-}b}= \left \lvert \textup{Pr}[\mathsf{Exp}^{\mathsf{BLI}\text{-}b}_{\mathsf{MA\text{-}ACEW},\mathcal{A}}(\lambda)] - \frac{1}{2}\right \rvert \leq \epsilon(\lambda)$$
    \vspace{-0.5em}
\end{definition}
\vspace{-1em}
\begin{figure}[h]
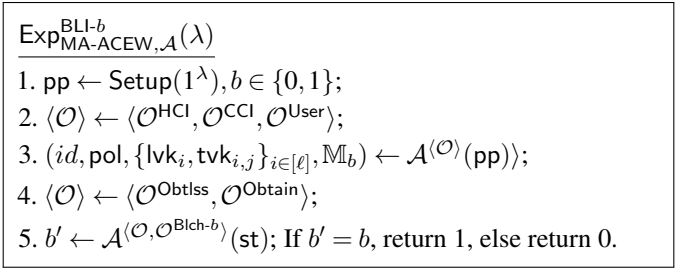

	\centering
	\fbox{    
		\begin{minipage}{\linewidth}
			\vspace{0.3em}
			$\underline{\mathsf{Exp}^{\mathsf{BLI}\text{-}b}_{\mathsf{MA\text{-}ACEW},\mathcal{A}}(\lambda)}$
			\vspace{0.2em}
			
			1. $\mathsf{pp} \leftarrow \mathsf{Setup}(1^\lambda), b \in \{0,1\};$
			\vspace{0.2em}
			
			2. $\langle \mathcal{O} \rangle \leftarrow  \langle \mathcal{O}^{\mathsf{HCI}}, \mathcal{O}^{\mathsf{CCI}}, \mathcal{O}^{\mathsf{User}}\rangle;$
			\vspace{0.2em}
			
			3. $ (id, \mathsf{pol}, \{\mathsf{lvk}_i, \mathsf{tvk}_{i,j}\}_{i \in [\ell]}, \mathbb{M}_b)\leftarrow \mathcal{A}^{\langle \mathcal{O}\rangle}(\mathsf{pp}) \rangle$;
			\vspace{0.2em}
			
			4. $\langle \mathcal{O} \rangle \leftarrow  \langle \mathcal{O}^{\mathsf{Obtlss}}, \mathcal{O}^{\mathsf{Obtain}}\rangle;$
			\vspace{0.2em}
			
			5. $b' \leftarrow \mathcal{A}^{\langle\mathcal{O}, \mathcal{O}^{\mathsf{Blch}\text{-}b}\rangle}(\mathsf{st})$; If $b' = b$, return 1, else return 0.
			\vspace{0.2em}
		\end{minipage}
	}    
	\caption{Blindness Security Game}\label{Blindness}

\end{figure}

\subsection{Building Blocks for MA-ACEW}\label{buildingblo}
Before presenting our MA-ACEW construction, we recall the weighted threshold signature (WTS) scheme of Das et al.~\cite{das2023threshold}, which ensures that aggregated credential weights in each epoch satisfy the policy. The scheme relies on the polynomial identity from \cite{ben2019aurora}, widely used in succinct SNARK constructions:
$$
x(t)b(t) = q(t) \cdot z_H(t) + t \cdot r(t) + \langle x, b \rangle \cdot n^{-1}.
$$
We adapt WTS to our framework with two changes: (i) moving from symmetric to asymmetric pairings, and (ii) splitting the $\mathsf{Combine}$ phase into $\mathsf{CombPk}$ and $\mathsf{CombWt}$, which enables WTS integration. The full construction appears in Appendix~\ref{wtscons}.

\subsection{MA-ACEW construction}\label{macons}

In this subsection, we present the concrete construction of MA-ACEW. For simplicity, we denote the EB-PS scheme as $\Psi_1$ and the WTS scheme as $\Psi_2$, and we instantiate the time-independent function $F(\mathsf{ctx},j)$ with access only to $j$ in MA-ACEW. These notations and the function are consistently used in Fig. ~\ref{fig:enter-label} and throughout this subsection. 

\begin{figure*}[!t]
	\centering
	\small\begin{mdframed}[linewidth=1pt,
		 ]     
		\noindent
		\textbf{Setup Phase:} Run $\mathsf{pp}_{\Psi_1} \leftarrow \Psi_1.\mathsf{Setup}(1^\lambda, \mathrm{T})$ $\wedge$ $\mathsf{pp}_{\Psi_2} \leftarrow \Psi_2.\mathsf{Setup}(1^\lambda)$, output $\mathsf{pp} = (\mathsf{pp}_{\Psi_1}, \mathsf{pp}_{\Psi_2}) =  (\mathsf{BG}, \mathsf{CRS}, \overrightarrow{v}_{\mathcal{L}}, \overrightarrow{\alpha}_{\mathcal{L}}, \overrightarrow{\beta}, \alpha, \theta, v^\eta).$
		
		\noindent
		\textbf{Initial Epoch IKGen Phase:} Initialize the epoch state by setting the current epoch $j := 1$.
		\begin{itemize}
			\item \greyback{IKG.1} Generate the long-term keys for each issuer $\{(\mathsf{lsk}_i, \mathsf{lvk}_i)\}_{i \in [n]} \gets \Psi_1.\mathsf{KGen}(\mathsf{pp}_{\Psi_1})$, where $\mathsf{lsk}_i = (x_i, y_i), \mathsf{lvk}_i = (\mathrm{X}_i, \mathrm{Y}_i)$.
			
			\item \greyback{IKG.2} Set initial weight vector $\bm{w}_1 := [w_{1,1}, \ldots, w_{n,1}]$, where $w_{i,1}$ corresponds to $\mathsf{CI}_i$'s public key $\mathrm{X}_i$.
			
			\item \greyback{IKG.3} Generate the epoch-specific keys for $\mathsf{CI}_i$:
			$\{\mathsf{tsk}_{i,1}, \mathsf{tvk}_{i,1}\}_{i \in [n]} \leftarrow \Psi_1.\mathsf{TKGen}(\mathsf{pp}_{\Psi_1}, n, 1)$,
			where $\mathsf{tsk}_{i,1} = z_i$, $\mathsf{tvk}_{i,1} = \mathrm{Z}_{i,1}$.
			
			\item \greyback{IKG.4} Compute the verification and aggregated keys from $ (\mathsf{ak}_i, \mathsf{vk})\leftarrow \Psi_2.\textsf{KGen}(\mathsf{pp}_{\Psi_2}, n, \bm{w}_1)$, where $$\mathsf{vk} = (v, \alpha, \beta, v^{x(\tau)}, v^{w_1(\tau)}, v^{\tau}, \alpha^{\tau}, u^{z_H(\tau)}), \quad \mathsf{ak}_i = (v^{x_i}, v_i^{x_i}, \alpha_i^{x_i}, \theta^{x_i}, v^{\eta x_i}, \{\beta_k^{x_i}\}_{k \in [n]})$$

		\vspace{-0.15in}\noindent\greycomment{// Note that $v^{w_1(\tau)}$ is only utilized in epoch $j = 1$, and $\mathsf{ak}_i$ is to assist user to  persists across all epochs. }
		\end{itemize}

		\noindent
		\textbf{UKGen phase:} Run $(\mathsf{aux}, \mathsf{tg}) \leftarrow \Psi_1.\mathsf{GenAuxTag}(\mathrm{S})$, return $(\mathsf{usk} = (\gamma, \delta), \mathsf{uvk} = (\Gamma, \Delta), \mathsf{aux} = (u^\gamma \parallel u^\delta \parallel \{c_{m_i}, c_i, \mathsf{lvk}_i\}_{i \in [\ell]}))$ to user.
		
		\noindent

        \textbf{Gen-Policies Phase:} 
		The verifier sets threshold $\mathrm{W}_{\textsf{acc}}$ and only accepts credentials of current epoch $j$. Return $\mathsf{pol} = (\mathrm{W}_{\mathsf{acc}},j)$.
		
		\smallskip
		\noindent
		
		\textbf{Issuance Phase:}
		\vspace{0.1cm}
		\noindent The user interacts with $\mathsf{CI}_i$ to obtain partial credentials at epoch $j$:
		\begin{itemize}
			\item \greyback{IS.1} The user sends $(\mathsf{tg}, \mathsf{aux}, \pi_{\mathsf{tg}}, \pi_{\mathsf{CI}_i})$ to an issuer $\mathsf{CI}_i$, where defined as:
			$$
			\begin{aligned}
				& \mathsf{aux} = (u^\gamma \parallel u^\delta \parallel \{c_{m_i}, c_i, \mathsf{lvk}_i\}_{i \in [\ell]}), \quad \pi_{\mathsf{tg}}= \mathsf{ZKPOK}\{(\gamma, \delta): \mathsf{tg}_1 = h^\gamma \land \mathsf{tg}_2 = h^\delta \land \Gamma = u^\gamma \land \Delta = u^\delta\}, \\
				& \quad \quad \pi_{\mathsf{CI}_i} = \mathsf{ZKPOK}\{(d, m_i, o_i, k_i): \zeta = u^d \land c_{m_i}= u^{m_i} h_1^{o_{i}} \land c_i= (u^{k_i}, \zeta^{k_i}\cdot (h')^{m_i})\}
			\end{aligned}
			$$
            \vspace{-0.15in}
			
			\item \greyback{IS.2}$\mathsf{CI}_i$ verifies $\pi_{\mathsf{CI}i}$ and $\pi_{\mathsf{tg}}$, checks the format of $\mathsf{aux}$, and if all pass,  then computes $(\widehat{\mathsf{cred}_{\mathsf{lt}, i}}, \mathsf{cred}_{\mathsf{ep}, i, j})$ for the user, where:
			$$\widehat{\mathsf{cred}_{\mathsf{lt},i}} = (h', c_{i,1}^{y_i}, (h')^{x_i}c_{i,2}^{y_i}), \quad \mathsf{cred}_{\mathsf{ep},i,j} = (h^\delta)^{j \cdot z_{i,j}}$$
            \vspace{-0.2in}
			
			\item \greyback{IS.3} The user unblinds the long-term credential to obtain $\mathsf{cred}_{\mathsf{lt},i} = (h^\gamma, (h^\gamma)^{x_i + m_iy_i})$ and stores $(\mathsf{cred}_{\mathsf{lt},i}, \mathsf{cred}_{\mathsf{ep},i,j})$.
			
			\noindent\greycomment{// Note that the user only interacts with $\mathsf{CI}_i$ once, as $\mathsf{cred}_{\mathsf{lt}, i}$ remains valid across all epochs.}
		\end{itemize}
		
		\smallskip
		\noindent
		\textbf{Epoch Update:} When the system transitions from epoch $j$ to epoch $j + 1$, the following operations are performed:
		\begin{itemize}
			\item \greyback{EP.1} Securely erase $\mathsf{CI}_i$'s previous $(\mathsf{tsk}_{i,j}, \mathsf{tvk}_{i,j})$, and run $\{\mathsf{tsk}_{i,j+1}, \mathsf{tvk}_{i,j+1}\}_{i \in [n]} \leftarrow \Psi_1.\mathsf{TKGen}(\mathsf{pp}, n, j+1)$ for  epoch $j+1$.
			\item \greyback{EP.2} Reset the epoch weight vector $\bm{w}_{j+1}$ and recompute the weight verification key $v^{w_{j+1}(\tau)}$, which is part of $\mathsf{vk}$.
			\item \greyback{EP.3} Each issuer $\mathsf{CI}_i$ computes $\mathsf{cred}_{\mathsf{ep}, i, j+1} = (h^\delta)^{(j+1)\cdot z_{i, j+1}}$ for its credentialed users as part of the credential update.
			\item \greyback{EP.4} The verifier updates $\mathsf{pol}$ to $(\mathrm{W}'_{\mathsf{acc}},j+1)$.
			
			\noindent \greycomment{// Upon receiving $\mathsf{cred}_{\mathsf{ep}, i, j}$ from each $\mathsf{CI}_i$, the user replaces the previous epoch signature $\mathsf{cred}_{\mathsf{ep}, i, j}$ with $\mathsf{cred}_{\mathsf{ep}, i, j+1}$.}
		\end{itemize}
		
		\smallskip
		\noindent
		\textbf{Show phase:} The user interacts with the verifier to show a credential in epoch $j+1$:
		\begin{itemize}
			\item \greyback{SH.1} The user computes the combined credential $\mathsf{cred}_{\mathsf{agg}, j+1}$ with disclose set $\mathbb{M}$ in current epoch, where:
			$$\begin{aligned}
				& \mathsf{cred}_{\mathsf{agg,lt}} \leftarrow \Psi_1.\mathsf{AggSigLt}(\mathsf{tg}, \{(\mathsf{lvk}_i, m_i, \mathsf{cred}_{\mathsf{lt},i})\}_{i=1}^\ell), \quad \mathsf{cred}_{\mathsf{agg,ep},j+1} \leftarrow \Psi_1.\mathsf{AggSigEp}(j+1, \mathsf{tg}, \{\mathsf{tvk}_{i,j+1}, \sigma_{\mathsf{ep},i,j}\}_{i=1}^\ell), \\
				& \quad \quad \quad \quad \quad \quad \quad \mathsf{cred}_{\mathsf{agg},j} \leftarrow \Psi_1.\mathsf{CombAggSig}(j+1, \mathsf{tg}, \mathsf{cred}_{\mathsf{agg,lt}}, \mathsf{cred}_{\mathsf{agg,ep},j})
			\end{aligned}$$
            \vspace{-0.15in}
			
			\item \greyback{SH.2} According to the disclosed set $\mathbb{M}$, set the bit vector $b[i] = 1$ for $i \in \mathbb{M}$, and compute the proof as:
			$$
			\begin{aligned}
				& (\pi_b, \pi_{\mathsf{IPA, pk}}, \mathrm{X}) \leftarrow \Psi_2.\mathsf{CombPk}(\bm{b}, \{\mathsf{ak}_i\}_{i \in [n]}, \mathsf{vk}), \quad (\pi_{\mathsf{IPA, wt}}, \mathrm{W}_{\mathsf{claim}}) \leftarrow \Psi_2.\mathsf{CombWt}(\bm{w}_{j+1}, \{\mathsf{ak}_i\}_{i \in [n]})
			\end{aligned}
			$$
            \vspace{-0.2in}
			\item \greyback{SH.3} Run $\pi_{\mathsf{IPA}} \leftarrow \Psi_2.\mathsf{MergePf}(\pi_{\mathsf{IPA,pk}}, \pi_{\mathsf{IPA,wt}}, \mathrm{X}, \mathrm{W}_{\mathsf{claim}}),(\mathsf{cred}_{\mathsf{agg},j+1}', \mathsf{tg}') \leftarrow \Psi_1.\mathsf{RndSigTag}(\mathsf{avk}_{j+1}, \mathsf{tg}, \mathsf{cred}_{\mathsf{agg}, j+1}, r)$.
			
			\item \greyback{SH.4} $ \text{The user} \text{ sends } (\mathsf{cred}'_{\mathsf{agg}, j+1}, \mathsf{tg}',u_b, \pi_b, \pi_{\mathsf{IPA}}, \mathrm{X}, \mathrm{W}_{\mathsf{claim}}, \mathbb{M}) \text{ to the verifier}$.
		\end{itemize}
		
		\smallskip
		\noindent
		\textbf{CredVerify Phase:} The verifier checks if the randomized credential is valid and satisfies the acceptance threshold defined in $\mathsf{pol}$.
		\begin{itemize}
			\item \greyback{CV.1} Run $\Psi_1.\mathsf{AggVerify}(j+1, \mathsf{tg}', \mathsf{avk}_{j+1}, \mathbb{M}, \mathsf{cred}'_{\mathsf{agg},j+1})$ and $\Psi_2.\mathsf{Verify}(\pi_{\mathsf{IPA}},\mathsf{vk}, \mathrm{X}, \mathrm{W}_{\mathsf{claim}})$.
			\item \greyback{CV.2} Check if $\mathrm{W}'_{\mathsf{acc}} \leq \mathrm{W}_{\mathsf{claim}}$. Return 1 if all verification checks are successful; otherwise, return 0.
		\end{itemize}
	\end{mdframed}
	\caption{MA-ACEW Construction}
	\label{fig:enter-label}
\end{figure*}

\noindent\textbf{Interactive issuing.} 
To prevent issuers from learning user attributes, we build upon techniques from \cite{sonnino2018coconut, DBLP:conf/ccs/MirBGLS23}.
For clarity, we refer to the $\mathsf{GenAuxTag}$ algorithm as $\mathsf{GenUserTag}$ and detail the credential issuance protocol.
In the original EB-PS scheme, the $\mathsf{GenAuxTag}$ algorithm exposed plaintext attributes within its output: $\mathsf{aux} = u^{\gamma} \parallel u^{\delta} \parallel \{(m_i, \mathsf{lvk}_i)\}_{i \in [\ell]}$.

Our core modification is to replace these plaintext attributes $\{m_i\}$ with their ElGamal encryptions.
Concretely, the user first generates an ElGamal key pair $(\mathsf{esk}, \mathsf{evk}) = (d, \zeta = u^{d})$. 
For each attribute $m_i$, the user then computes its ciphertext $c_i = \mathsf{Enc}((h')^{m_i},\mathsf{evk}) = (u^{k_i}, \zeta^{k_i}\cdot (h')^{m_i})$ and, in parallel, a commitment to the attribute $c_{m_i} = u^{m_i} h_1^{o_{i}}$.
The values $(d, o_i, k_i)$ are sampled uniformly at random from $\mathbb{F}_p$, and $h_1 \in \mathbb{G}_1$.
Additionally, the user needs to generate a corresponding proof:
\begin{align*}
	\pi_{\mathsf{CI}_i} = \mathsf{ZKPOK}\{ & (d, m_i, o_i, k_i): \zeta = u^d \land c_{m_i}= u^{m_i} h_1^{o_{i}} \\
	& \land c_i = (u^{k_i}, \zeta^{k_i}\cdot (h')^{m_i})\}
\end{align*}
\noindent The auxiliary information $\mathsf{aux}$ is now defined as
$$\mathsf{aux} = (u^\gamma \parallel u^\delta \parallel \{c_{m_i}, c_i, \mathsf{lvk}_i\}_{i \in [\ell]})$$

\noindent And the issuing phase between user and issuer $\mathsf{CI}_i$ as:

\begin{itemize}[leftmargin=*]
	\item The user transmits $(\mathsf{tg}, \mathsf{aux}, \pi_{\mathsf{CI}_i}, \pi_{\mathsf{tg}})$ to the credential issuer $\mathsf{CI}_i$, where 
	\begin{align*}
    \pi_{\mathsf{tg}} = \mathsf{ZKPOK}\Big\{ (\gamma, \delta) : & \; \mathsf{tg}_1 = h^\gamma \land \mathsf{tg}_2 = h^\delta \land \\
     & \qquad \Gamma = u^\gamma \land \Delta = u^\delta \Big\}
\end{align*}
	is uniformly utilized for all credential issuers.
	\item The issuer $\mathsf{CI}_i$ parses $\mathsf{aux}$ and verifies the validity of proofs $\pi_{\mathsf{CI}_i}$ and $\pi_{\mathsf{tg}}$. Upon successful verification, $\mathsf{CI}_i$ executes blind issuance on the attribute $m_i$, computing:$$\widehat{\mathsf{cred}_{\mathsf{lt},i}} = (h', c_{i,1}^{y_i}, (h')^{x_i}c_{i,2}^{y_i}), \quad \mathsf{cred}_{\mathsf{ep},i,j} = (h^\delta)^{j \cdot z_{i,j}}$$
	\item The user unblinds the partial long-term credential as $\mathsf{cred}_{\mathsf{lt},i} = (h', (h')^{x_i} c_{i,2}^{y_i} (c_{i,1}^{y_i})^{-d})$
\end{itemize}

\noindent\textbf{Epoch transition.} In MA-ACEW, the user interacts with each credential issuer $\mathsf{CI}_i$ only once to obtain the long-term credential. Additionally, $\mathsf{CI}_i$ issues epoch-based credentials $\mathsf{cred}_{\mathsf{ep},i,j}$. To streamline this process, $\mathsf{CI}_i$ maintains a list of registered users, identified by their unique tags $\mathsf{tg}$. For each new epoch $j+1$, $\mathsf{CI}_i$ computes and issues fresh epoch-based credentials $\mathsf{cred}_{\mathsf{ep}, i, j+1}$ for all registered users without requiring additional interaction, which eliminates the need for subsequent interactions between users and credential issuers.

During the system initialization phase, each credential issuer $\mathsf{CI}_i$ generates secret keys $(\mathsf{lsk}_i = (x_i, y_i), \mathsf{tsk}_{i,1} = z_i)$, where $y_i$ is used for signing attributes and $z_i$ for updating (signing epoch information for each user). The value $\mathrm{X}_i = v^{x_i}$ corresponds to $\mathsf{CI}_i$'s initial weight $w_{i,1}$. Additionally, all issuers $\{\mathsf{CI}_i\}_{i \in [n]}$ compute aggregation keys $\mathsf{ak}_i$ to assist the user. Assuming a Public Key Infrastructure (PKI), users can non-interactively retrieve $\mathsf{ak}_i$ from each $\mathsf{CI}_i$, as stated in \cite{das2023threshold}. Since the weight vector $\bm{w}_j$ is publicly known and the public key vector $\bm{pk}_x$ remains valid across all epochs, no complex operations are required for subsequent updates.

When the system transitions from epoch $j$ to epoch $j+1$, the weight vector is reset to $\bm{w}_{j+1}$, where each $w_{i,j+1}$ continues to correspond to the issuer's long-term public key $\mathrm{X}_i$. This update requires only one term in the public verification key $\mathsf{vk}$ to be recomputed: $v^{w_{j+1}(\tau)}$. Each credential issuer $\mathsf{CI}_i$ updates its epoch-specific keys by executing $(\mathsf{tsk}_{i,j+1}, \mathsf{tvk}_{i,j+1}) \leftarrow \Psi_1.\mathsf{TKGen}(\mathsf{pp}_{\Psi_1}, j+1)$. Additionally, each $\mathsf{CI}_i$ computes new epoch-based credentials $\mathsf{cred}_{\mathsf{ep},i,j+1} \leftarrow \Psi_1.\mathsf{SignEp}(\mathsf{tsk}_{i,j+1}, j+1, \mathsf{tg})$ for each registered user. As $w_{i,j}$ is publicly known, this computation can be performed independently of the credential issuers.

\vspace{0.01in}
\noindent\textbf{Interactive Showing.}
According to practical scenarios, the verifier needs to define a policy $\mathsf{pol}$ to determine whether to accept a user's obtained credential. Here, the verifier only accepts credentials updated in the current epoch, i.e., the credential issuer $\mathsf{CI}_i$ must have computed a valid epoch-specific signature for the user under $\mathsf{tg}$. Additionally, the verifier can adjust the acceptance threshold $\mathrm{W}_{\mathsf{acc}}$ according to the weight distribution scenario of nodes.

The user processes the set of disclosed attributes $\mathbb{M}$ and the corresponding partial credentials $\{\mathsf{cred}_{\mathsf{lt},i}, \mathsf{cred}_{\mathsf{ep},i,j+1}\}_{i \in [\mathbb{M}]}$ as follows. First, the user aggregates the long-term credentials set and the epoch-specific credential set separately. Then, these two aggregated credentials are combined into a single credential of constant size. Notably, the user can precompute the aggregated long-term credentials. Since epoch-specific credentials are re-issued by credential issuers at the beginning of each epoch, the user only needs to aggregate them when presenting a credential during the current epoch $j+1$. It is important to note that only one attribute per $\mathsf{lvk}_i$ can be issued (corresponding to the secret key $y_i$). However, this can be extended to support multiple attributes by generating a series of $\mathsf{lvk}_i$, as demonstrated in \cite{PS16short}, \cite{sonnino2018coconut}. We will elaborate on this extension in Section \ref{secadd}.

Additionally, the user computes the corresponding IPA proofs: $\pi_{\mathsf{IPA, pk}}$ to convince the verifier that $\langle \bm{pk}_x, \bm{b}\rangle = \mathrm{X}$, and $\pi_{\mathsf{IPA, wt}}$ to convince the verifier that $\langle \bm{w}_{j+1}, \bm{b} \rangle = \mathrm{W}_{\mathsf{claim}}$ satisfies the acceptance threshold $\mathrm{W}'_{\mathsf{acc}}$, i.e., $\mathrm{W}_{\mathsf{claim}} \geq \mathrm{W}'_{\mathsf{acc}}$. These IPA proofs are merged into a single proof by taking their random linear combination using the Fiat-Shamir heuristic \cite{fiat1986prove}. Notably, the proof $\pi_{\mathsf{IPA, pk}}$ can be precomputed in the previous epoch. During epoch $j+1$, the user only needs to compute the epoch-specific aggregated credentials and the weight proof.

However, the pre-computation is only applicable when the total weight of the user's obtained credentials exceeds the verifier-defined threshold $\mathsf{pol}$. If the total weight is insufficient, the user must obtain additional credentials and re-compute the corresponding proof. Nonetheless, this scenario is uncommon in practice, as system parameters typically remain relatively stable across epochs, and users generally maintain a total credential weight that exceeds the required threshold. 

Given the IPA proof and the combined credential, the verifier verifies the credential's validity under epoch $j+1$ by checking that: (a) The aggregated credential $\mathsf{cred}_{\mathsf{agg}, j+1}$ is a valid credential on the disclosed attribute set $\mathbb{M}$, and has been updated to the current epoch $j+1$; (b) $\pi_b$ is a correct proof of commitment $u_b$ to vector $\bm{b}$, confirming that $\bm{b}$ is a bit vector;  (c) $\pi_{\mathsf{IPA}}$ is a valid \noindent IPA proof for both the aggregated public key $\mathrm{X}$ and that the provided $\mathrm{W}_{\mathsf{claim}}$ satisfies the defined threshold in $\mathsf{pol}$.
\begin{theorem}\label{masec1}[{\textsc{U}}\textsc{nforgeability}]
	If the EB-PS signature scheme is unforgeable, the IPA protocol satisfies knowledge soundness, and the ZKPoK is simulation-sound extractable, then the MA-ACEW construction achieves unforgeability.
\end{theorem}

\begin{theorem}\label{masec2}[\textsc{A}\textsc{nonymity}]
	MA-ACEW achieves anonymity if the DDH assumption holds in $\mathbb{G}_1$ and the ZKPoK protocol is zero-knowledge.
\end{theorem}

\begin{theorem}\label{masec3}[\textsc{B}\textsc{lindness}]
	MA-ACEW achieves blindness if the ElGamal encryption is IND-CPA secure and ZKPoK protocol is zero-knowledge.
\end{theorem}

\noindent  The complete security proofs can be found in Appendix \ref{sec Ma-acew}.

\subsection{Additional Properties}\label{secadd}
In this subsection, we extend our MA-ACEW scheme to support several advanced properties, namely \textit{Multi-Attribute Credential Issuance}, \textit{Issuer Hiding}, and policies supporting the \textit{Selective Disclosure of Arbitrary Attributes}. We show that these features can be incorporated with minor modifications to our core construction. The full details are deferred to Appendix \ref{additi}.

\subsection{Application}
In this subsection, we demonstrate the practical utility of our scheme within PoS-based ecosystems, specifically instantiating it for privacy-preserving DAO governance.

Consider a DAO voting scenario where a user acts as a delegate, aggregating voting power from multiple stakeholders (who act as credential issuers). Initially, a user holding $(\mathsf{aux}, \mathsf{tg})$ interacts with these stakeholders to accumulate the necessary credentials. Here, the public key component $\mathrm{X}_i$ directly correlates  with the $i$-th stakeholder $\mathsf{CI}_i$'s voting weight $w_{i,j}$. To prove possession of sufficient voting power, the user employs a binary participation vector $\bm{b}$. This vector ensures that only the stakeholders contributing to the aggregated credential $\mathsf{cred}_{\mathsf{agg},j}$ are selected, meaning that only those entries where $b_i = 1$ are included in the inner products for the aggregated key $\langle \bm{pk}_x, \bm{b} \rangle$ and the total accumulated weight $\langle \bm{w}_j, \bm{b} \rangle$. Consequently, during the DAO verification phase, the user generates proofs $\pi_{\mathsf{IPA}}$ and $\pi_{b}$ to attest that the aggregated credential validly reflects the sum of the selected stakeholders' weights, thereby meeting the governance policy $\mathsf{pol}$.

To ensure governance integrity over time, our system implements a dynamic epoch mechanism where credentials require an issuer-assisted refresh for each new governance cycle. Instead of relying on static permissions, a user intending to vote obtains an epoch-based credential $\mathsf{cred}_{\mathsf{ep}, i, j}$ from each stakeholder $\mathsf{CI}_i$ with whom they have previously interacted, in order to locally update their credential to the current state. This on-demand synchronization directly addresses the challenge of authorization freshness within the DAO: it allows the governance contract to strictly verify that a delegate's voting power is currently active for the specific epoch. Moreover, due to the randomizability of the credentials, the system upholds user unlinkability, ensuring that the validation of a user's standing does not create a traceable history of their past governance activities.

\section{Performance Evaluation}\label{perform}
We implement and evaluate our constructions in Golang, providing complete implementations of both EB-PS and MA-ACEW. All code is available in an open Zenodo repository\footnote{\url{https://doi.org/10.5281/zenodo.17905110}}. Our implementation leverages the BLS12-381 pairing-based curve and builds upon the WTS construction from \cite{das2023threshold}. To enhance computational efficiency, we employ multi-exponentiation techniques for group elements in the aggregation process. All performance measurements were conducted on a machine equipped with an Intel Core i7-1260P CPU @2.10 GHz, 16GB RAM running Ubuntu 18.04.

\subsection{Implementation Benchmarks}
For all benchmarks, we selected parameters to achieve 128-bit security. We first evaluated the basic operations: a single exponentiation in the source groups $\mathbb{G}_1$ and $\mathbb{G}_2$ of the elliptic curve took approximately 110$\mu$s and 252$\mu$s. For storage requirements, a single compressed element in $\mathbb{G}_1$ required $48$ bytes, while an element in $\mathbb{G}_2$ required $96$ bytes.

\noindent\textbf{Communication Overhead}. We first analyze the communication overhead of EB-PS and MA-ACEW constructions during the signing/issuance phase. 

\begin{table}[h]
	\centering	\renewcommand{\arraystretch}{1.2}
    \setlength{\tabcolsep}{4pt} 
	\begin{tabular}{l|cccc|cc}
		\hline\noalign{\hrule height 0.5pt}
		\multirow{2}{*}{\textbf{Constr.}} & \multicolumn{4}{c|}{\textbf{$\mathsf{U} \rightarrow \mathsf{S}/\mathsf{CI}$}} & \multicolumn{2}{c}{\textbf{$\mathsf{U} \leftarrow \mathsf{S}/\mathsf{CI}$}} \\
		\cline{2-7}
		& $\mathbb{G}_1$ & $\mathbb{G}_2$ & $\mathbb{F}_p$ & Bytes & $\mathbb{G}_1$ & Bytes \\
		\hline
		EB-PS & $4$ & $n$ & $n$ & $192 + 128n$ & $3$ & $144$  \\
		\hline
		MA-ACEW & $4 + 3n$ & $n$ & $0$ & $192 + 240n$ & $4$ & $196$ \\
		\noalign{\hrule height 0.5pt}\hline
	\end{tabular}
\caption{Communication Overhead of Signing/Issuance}	\label{communication issuance}
\end{table}

Table \ref{communication issuance} presents the communication costs from the user's perspective. For auxiliary information $\mathsf{aux}$, EB-PS employs $\mathsf{GenAuxTag}$ while MA-ACEW utilizes $\mathsf{GenUserTag}$. Due to the encrypted messages in MA-ACEW, its communication overhead is naturally higher. Notably, Table \ref{communication issuance} excludes the zero-knowledge proof costs, which users can selectively generate for attributes being issued. For each issuer, these additional zero-knowledge proofs require $512$ bytes for an encrypted message and $256$ bytes for tag proofs.
\begin{table}[h]
	\centering
	\renewcommand{\arraystretch}{1.2}
    \setlength{\tabcolsep}{3pt} 
	\begin{tabular}{l|cccc|ccc}
		\hline\noalign{\hrule height 0.5pt}
		\multirow{2}{*}{\textbf{Constr.}} & \multicolumn{4}{c|}{\textbf{Sig./Cred. size}} & \multicolumn{3}{c}{\textbf{PK size}} \\
		\cline{2-8}
		& $\mathbb{G}_1$ & $\mathbb{G}_2$ & $\mathbb{F}_p$ & Bytes & $\mathbb{G}_1$ & $\mathbb{G}_2$ & Bytes \\
		\hline
		EB-PS & $4$ & $0$ & $0$ & $192$ & $0$ & $3n$ & $288n$  \\
		\hline
		MA-ACEW & $6$ & $5$ & $1$ & $800$ & $1$ & $3n + 7$ & $720 + 288n$ \\
		\noalign{\hrule height 0.5pt}\hline
	\end{tabular}
	\caption{Communication Overhead of Verify Phase}
\label{communication show}
\vspace{-0.5em}
\end{table}

Table \ref{communication show} presents the communication overhead for showing signatures or credentials in both EB-PS and MA-ACEW schemes. Due to signature aggregation capabilities, the user-side communication costs remain constant in both constructions regardless of the number of attributes. Compared to EB-PS, MA-ACEW requires additional computation for inner product arguments, resulting in higher overhead. For public key size, we observe a linear relationship with the number of attributes since we represent unaggregated keys. For practical storage optimization, verifiers can aggregate these public keys before verification, significantly reducing space requirements.

\noindent\textbf{Timing benchmark}. We evaluate the computational efficiency of both EB-PS and MA-ACEW by measuring the execution time of each algorithm/phase. For each algorithm, we report both the mean execution time and standard deviation. Our measurements use Go's built-in benchmarking framework, which adaptively determines iterations for statistical significance.

\begin{table}[h]
	\centering
	\setlength{\tabcolsep}{2pt}
	\begin{tabular}{@{}lccccc@{}}
		\hline\noalign{\hrule height 0.5pt}
		\textbf{Metric} & GenAuxTag & SignLt & SignEp & Verify & RndSigTag \\
		\midrule
		Time (ms)     & 0.92      & 0.34     & 0.13   & 1.29   & 0.45     \\
		Std Dev ($\pm$)    &  0.38    &  0.02    &  0.01   &  0.21  &  0.06    \\
		\midrule[0.5pt]
		\multicolumn{6}{c}{\textbf{Prims(ms):} $\mathbb{G}_1$ Exp.: 0.11 $\mid$ $\mathbb{G}_2$ Exp.: 0.25 $\mid$ Pairing: 0.85} \\
		\noalign{\hrule height0.5pt}\hline
	\end{tabular}
\caption{Performance Evaluation of EB-PS Scheme}

\label{EB-PS}
\end{table}

Table \ref{EB-PS} presents the performance metrics of the proposed EB-PS scheme detailed in Section \ref{ebpscons}, including execution time and communication overhead for key operations. Our evaluation focused on critical algorithms within the EB-PS construction. The $\mathsf{GenAuxTag}$ algorithm was benchmarked with 128 messages, while $\mathsf{SignLt}$, $\mathsf{SignEp}$, and $\mathsf{Verify}$ algorithms were evaluated in a single-signer scenario. Notably, the $\mathsf{Verify}$ algorithm demonstrates higher computational complexity due to its bilinear pairing operations.

As shown in Table \ref{tab:ma-acew-perf}, we evaluated MA-ACEW performance across different operational phases, reporting execution times and standard deviations in a network with 64 signers. The Setup phase incurs substantial overhead due to CRS generation with complex cryptographic operations. However, the Initial IKGen phase dominates the computational cost with exceptionally high execution times, as evidenced by our measurements. This overhead stems from necessary CRS preprocessing and encompasses critical procedures: generation of long-term and epoch signing keys for all issuers, verification keys for inner product arguments, and proof generation for users. Our implementation replaces original messages in $\mathsf{GenAuxTag}$ with ciphertexts, increasing overhead compared to our EB-PS implementation. Furthermore, the issuance phase introduces additional computational costs due to its blind issuance mechanism, which requires users to generate ZKPoK and issuers to verify these proofs.

\begin{table}[h]
	\centering
	\small
	\begin{tabular}{lcc|lll}
		\hline\noalign{\hrule height 0.5pt}
		\textbf{Phase} & Time (ms) & Std Dev ($\pm$) & \multicolumn{3}{l}{\textbf{Notes}} \\
		\midrule
		Setup & 46.34 & 10.34 & \multicolumn{3}{l}{ 64 issuers} \\
		Initial IKGen & 445.84 & 12.17 & \multicolumn{3}{l}{ 64 issuers} \\
		UKGen & 5.93 & 0.31 & \multicolumn{3}{l}{ 10 messages} \\
		Issuance & 2.74 & 0.86 & \multicolumn{3}{l}{User side} \\
		Issuance & 4.55 & 1.52 & \multicolumn{3}{l}{Issuer side} \\
		Epoch Update & 14.49 & 1.28 & \multicolumn{3}{l}{ 64 issuers} \\
		\noalign{\hrule height 0.5pt}\hline
	\end{tabular}
	\caption{Performance Evaluation of MA-ACEW Phases}
	\label{tab:ma-acew-perf}
\end{table}

Fig. \ref{M1show} presents the computational costs of the $\mathsf{AggVerify}$ algorithm in EB-PS, while Fig. \ref{M2show} illustrates the computational overhead of the $\mathsf{Show}$ and $\mathsf{CredVerify}$ phases in MA-ACEW. Our evaluation examines performance across varying numbers of signers/credential issuers, ranging from $4$ to $128$, with each issuer corresponding to a single disclosed attribute in our experimental setup.

The $\mathsf{Show}$ phase in MA-ACEW includes computation of Inner Product Argument (IPA) proofs. Although both schemes show increasing computational costs with more attributes, our analysis reveals that the IPA computation time grows very slowly. The observed linear growth in MA-ACEW's execution time primarily stems from the underlying EB-PS operations. Although we tested scenarios with up to 128 issuers, practical applications rarely require this many, as our MA-ACEW design allows the system to operate efficiently with fewer issuers. 
\begin{figure}[h]
	\centering
	\begin{subfigure}[b]{0.235\textwidth}
		\centering
		\includegraphics[width=\textwidth]{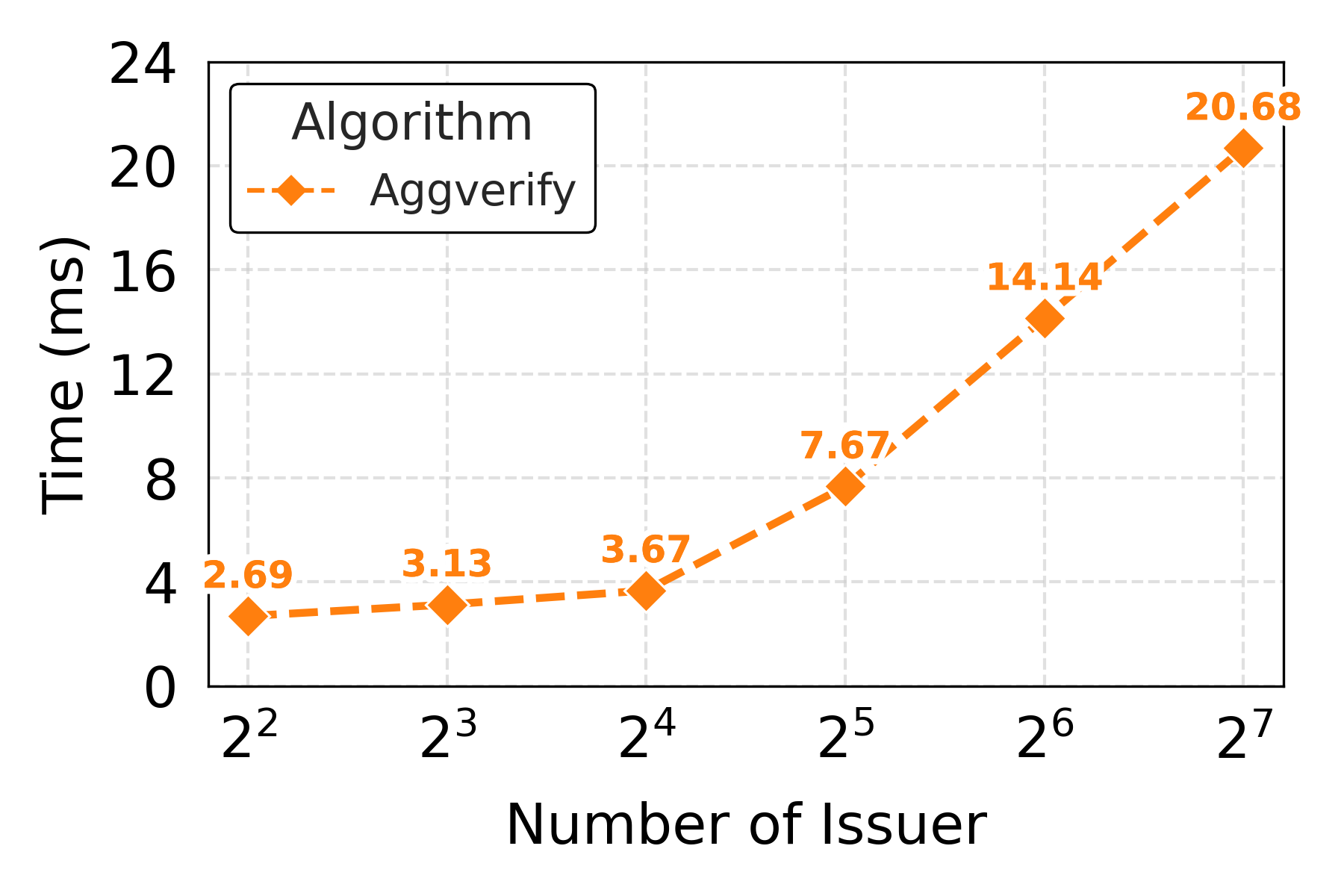}
		{\footnotesize\caption{EB-PS performance}\label{M1show}}
	\end{subfigure}
	\hfill
	\begin{subfigure}[b]{0.235\textwidth}
		\centering
		\includegraphics[width=\textwidth]{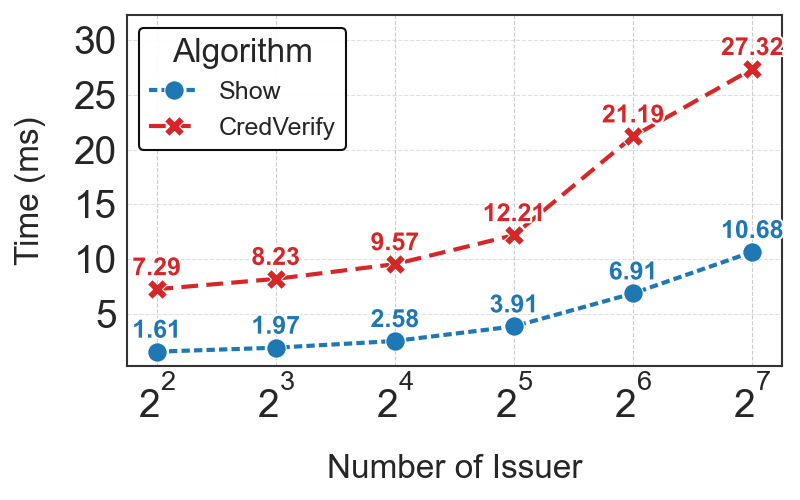}
		{\footnotesize\caption{MA-ACEW performance}\label{M2show}}
	\end{subfigure}
	{\footnotesize\caption{Execution time analysis of EB-PS and MA-ACEW}}
	\label{fig:M2}
\end{figure}
\vspace{-1em}
\subsection{Smart-Contract Evaluation}

We extended our evaluation to include a smart contract implementation of our MA-ACEW scheme, focusing on the issuance and credential verification functionalities. The implementation was developed in Solidity utilizing the BN254 asymmetric pairing curve. We selected BN254 due to its native support in Ethereum and superior efficiency among pairing-friendly curves in blockchain contexts.

\begin{table}[h]
	\centering
	\renewcommand{\arraystretch}{1.2} 
	\setlength{\tabcolsep}{5pt}       
	
	\begin{tabular}{lccc}
		\hline\noalign{\hrule height 0.5pt}
		\textbf{Operation} & \textbf{Group Ops.} & \textbf{Gas} & \textbf{Est. Cost (USD)} \\ 
		\midrule
		LT. Iss. & $2\mathbb{G}_1 \text{Exp} + 2\mathbb{P}$ & 255K & 0.11 \\
		Ep. Iss.     & $1\mathbb{G}_1 \text{Exp}$               & 70K  & 0.03 \\
		Ver.          & $19 \mathbb{G}_1 \text{Exp} + 15 \mathbb{P}$ & 1077K & 0.47 \\
		\noalign{\hrule height 0.5pt}\hline
	\end{tabular}
	
	\caption{Performance metrics on Ethereum}
	\label{tab:crypto-perf}
\end{table}

We present our detailed performance evaluation in Table \ref{tab:crypto-perf}. Our analysis focuses on two critical phases: the Issuance phase and the credential verification phase, corresponding to the credential issuer and verifier operations, respectively. The Issuance phase comprises two distinct procedures: long-term credential issuance and epoch-credential issuance. The former includes computationally intensive proof verification operations, resulting in significantly higher computational costs compared to the more efficient epoch-credential issuance procedure. For the verification phase, we implement a constant-time approach that eliminates the variable complexity associated with multiple pairing operations. Since exponentiation in $\mathbb{G}_2$ incurs significant computational overhead, we pre-compute necessary group elements and optimize the protocol to process only signature pairs and inner product arguments. Our implementation leverages the batch verification technique introduced by Das et al.~\cite{das2023threshold}, adapting their optimized pairing-based verification framework available in the open-source codebase\footnote{\url{https://github.com/sourav1547/wts.}}. This approach reduces the overall verification cost by aggregating multiple pairing operations through a randomized linear combination technique, resulting in efficient and optimized pairing checks. Finally, to provide economic context for the tabulated results, USD costs were estimated using Ethereum network parameters from November 7, 2025 (ETH $\approx$ \$3,450, median gas price $\approx$ 0.128 Gwei), noting that actual costs vary with network congestion.

\subsection{Comparison with Naive Solution}

In this subsection, we evaluate the performance of our proposed MA-ACEW scheme through both theoretical complexity analysis and experimental implementation. We utilize a virtualized Pointcheval-Sanders (PS) signature scheme as the baseline for comparison to demonstrate the efficiency gains of our approach.

\noindent{\textbf{Theoretical Analysis.}  As shown in Table \ref{tab:complexity}, the computational complexity of the baseline design depends linearly on the weight parameters. Specifically, the signing cost scales with the individual weight $w_i$ of the $i$-th issuer, while the aggregation and verification costs grow linearly with the total accumulated weight $\mathrm{W}$. In contrast, MA-ACEW decouples computational overhead from weight magnitude, achieving a constant signing cost per issuer and ensuring that aggregation and verification costs scale solely with the number of issuers $n$, irrespective of the total weight $\mathrm{W}$.

\begin{table}[t]
    \centering
    \renewcommand{\arraystretch}{1.25} 
    \caption{Theoretical Complexity Comparison with Baseline}
    \label{tab:complexity}
    
    \begin{threeparttable}
        \begin{tabularx}{\linewidth}{l >{\centering\arraybackslash}X >{\centering\arraybackslash}X >{\centering\arraybackslash}X}
            \toprule
            \textbf{Scheme} & \textbf{Sign Cost} \newline \textit{(Per Issuer)} & \textbf{Agg. Cost} \newline \textit{(Total)} & \textbf{Ver. Cost} \newline \textit{(Total)} \\
            \midrule
            Baseline & $O(w_i)$\tnote{*} & $O(\mathrm{W})$\tnote{\textdagger} & $O(\mathrm{W})$\tnote{\textdagger} \\
            MA-ACEW  & $O(1)$            & $O(n)$\tnote{\textdaggerdbl} & $O(n)$\tnote{\textdaggerdbl} \\
            \bottomrule
        \end{tabularx}
        
        \begin{tablenotes}[para, flushleft]
            \footnotesize
            \item[*] $w_i$: weight of the specific signer; 
            \item[\textdagger] $\mathrm{W}$: total weight; 
            \item[\textdaggerdbl] $n$: number of issuers.
        \end{tablenotes}
    \end{threeparttable}
\end{table}

\noindent\textbf{Performance Benchmarking.} As shown in Table \ref{tab:performance_swapped_lined}, we fix the total accumulated weight at $\mathrm{W}=256$ for our evaluation. Since the baseline does not support weighted issuers, it necessitates a number of issuers equal to the total weight (i.e., $n=256$), thereby fixing its cost to this target. We first compare the worst-case scenario for our scheme, where the weight is fragmented among 256 issuers (i.e., 256 weight-1 issuers). As expected, our solution is marginally slower than the baseline in this specific setting due to the overhead of additional proofs. However, since our performance is determined by the number of issuers rather than the total weight, we gain efficiency by achieving the same weight with fewer issuers. To demonstrate this efficiency while ensuring a fair comparison (i.e., avoiding the extreme case of a single high-weight issuer), we evaluate scenarios with 128 and 64 issuers. In these more typical configurations, our scheme significantly outperforms the baseline.
\begin{table}[htbp]
  \centering
  \small
  \renewcommand{\arraystretch}{1.25} 
  \setlength{\tabcolsep}{3pt} 
  \caption{Performance Comparison with Baseline}
  \label{tab:performance_swapped_lined}
  
  \begin{threeparttable}
    \begin{tabular}{lcccc}
      \toprule
      \multirow{2}{*}{\textbf{Scheme}} & \multicolumn{2}{c}{\textbf{Showing}} & \multicolumn{2}{c}{\textbf{Verification}} \\
      \cmidrule(lr){2-3} \cmidrule(lr){4-5}
       & \textbf{Time (ms)} & \textbf{Speedup} & \textbf{Time (ms)} & \textbf{Speedup} \\
      \midrule
      
      Baseline & 16.1 & — & 31.5 & — \\
      
      \midrule
      
      Ours ($n=256$) & 17.3 & 0.93$\times$ & 36.3 & 0.87$\times$ \\
      Ours ($n=128$) & 10.5 & \textbf{1.53$\times$} & 25.4 & \textbf{1.20$\times$} \\
      Ours ($n=64$)  & 7.0  & \textbf{2.30$\times$} & 16.7 & \textbf{1.89$\times$} \\
      \bottomrule
    \end{tabular}
    
    \begin{tablenotes}
      \footnotesize
      \item \textit{Note:} $n$ denotes the number of issuers. Speedup is relative to Baseline.
    \end{tablenotes}
  \end{threeparttable}
\end{table}
}

\section{Related Work}
We focus our analysis on decentralized anonymous credentials, aggregate signatures, and PoS verification mechanisms.

\noindent\textbf{Decentralized Anonymous Credentials}.
Garman et al. \cite{garman2013decentralized} introduced an innovative approach to decentralized credential systems without relying on traditional signing authorities. Their scheme leverages a public ledger (blockchain) where users register commitments to their attributes. Users prove possession of credentials by constructing zero-knowledge proofs based on RSA accumulators of ledger entries. However, this solution faces scalability challenges and may not be suitable for credentials that inherently require trusted issuers. To enhance security and reduce centralization in credential issuance, recent schemes have explored threshold issuance, distributing credential issuance authority among a group of entities. Such schemes require collaboration from at least a threshold number of issuers to produce a valid credential \cite{sonnino2018coconut, camenisch2020short, doerner2023threshold,DBLP:conf/asiaccs/BishtKASFK24}. A notable example is Coconut \cite{sonnino2018coconut}, a threshold-based anonymous credential scheme designed for blockchain environments, built upon threshold PS signatures \cite{PS16short}. Coconut's security is based on an interactive assumption similar to, but distinct from, the LRSW assumption \cite{lysyanskaya2000pseudonym}. Importantly, its security has been formally proven within the Universal Composability (UC) framework \cite{rial2022security}.

Multi-Authority Anonymous Credentials (MA-ACs), introduced in \cite{hebant2023traceable} and further studied in \cite{DBLP:conf/ccs/MirBGLS23}, enable the aggregation of credential showings from different issuers. The work in \cite{DBLP:conf/ccs/MirBGLS23} introduces two key cryptographic primitives: tag-based aggregate signatures with randomizable tags and public keys (AtoSa) and Aggregate Mercurial Signatures with Randomizable Tags (ATMS). Notably, their scheme achieves the issuer-hiding property, which is also studied in \cite{bobolz2021issuer, bosk2022hidden, connolly2022improved}.

\textit{While these distributed credential schemes achieve decentralized issuance, they fundamentally rely on underlying aggregate signature schemes.}

\noindent\textbf{Aggregate Signatures}. Aggregate signatures, introduced by Boneh et al. \cite{boneh2003aggregate}, support aggregation of signatures from different parties on possibly different messages. The key requirement is succinctness of the aggregated signature. As a result of this innovative concept, several variants have been studied and developed \cite{neven2008efficient, boldyreva2007ordered, hohenberger2013full, hohenberger2015universal}. Among these variants, one notable example is the synchronized aggregate signature, first proposed by Gentry and Ramzan \cite{gentry2006identity}. Subsequently, \cite{ahn2010synchronized} further advanced this idea by proposing a pairing-based scheme without random oracles.

Central to our work is the Pointcheval-Sanders (PS) signature scheme \cite{PS16short}, which provides a short signature size compared to the Camenisch-Lysyanskaya (CL) signature scheme \cite{camenisch2004signature}. Later, recent research has explored various structure-preserving properties in signature schemes. Ghadafi et al. \cite{ghadafi2021partially} introduced the notion of partially structure-preserving signatures. In fully structure-preserving signature schemes \cite{abe2010structure}, all messages, signatures, and public keys are group elements. Further advancing this line of research, Crites et al. \cite{crites2023threshold} proposed message-indexed structure-preserving signatures. Their construction, inspired by both the PS signature scheme and Ghadafi's work \cite{ghadafi2016short}, parameterizes messages with an indexing function, allowing aggregation of different signatures for different messages under different public keys if they share the same index.

\textit{We now move from theory to practice, exploring applications in the PoS setting, particularly unweighted multi-signature schemes and SPV-style proofs.}

\noindent\textbf{Unweighted Multi-signatures and SPV-style Proofs}. To reduce the substantial bandwidth and storage requirements inherent in Proof-of-Stake blockchains, Drijvers et al. \cite{drijvers2020pixel} proposed Pixel, a pairing-based forward-secure multi-signature scheme that optimizes verification costs. Following this line of work, Wei et al. \cite{wei2024pixel+} proposed Pixel+ and Pixel++, two forward-secure multi-signature schemes that optimize verification in PoS blockchains by aggregation of public keys. Addressing security tightness, \cite{bacho2024tightly} introduces a new variant of BLS multi-signatures and demonstrates how PoS protocols that currently use BLS can adopt it for fully compatible opt-in tight security. In a parallel effort to save computational resources, Baldimtsi et al. \cite{baldimtsi2024subset} propose an efficient BLS multi-signature variant for PoS blockchains that replaces per-signature randomization with a one-time public key randomization, saving significant computational resources during aggregation and verification.

Complementing signature aggregation schemes, SPV-style proofs also aim to drastically reduce verification costs. Addressing the scalability issues of traditional SPV clients, FlyClient \cite{bunz2020flyclient} utilizes an optimal probabilistic block sampling protocol and Merkle Mountain Range (MMR) commitments to enable sub-linear light client verification. Pushing efficiency further, Vesely et al. \cite{vesely2022plumo} presented Plumo. Instead of downloading full headers, it securely verifies SNARK-based state transition proofs to confirm the latest network state. Furthermore, to address security in cross-chain communication, \cite{ciobotaru2022accountable} proposed an accountable light client system designed to secure cross-chain bridges by validating incremental state updates and identifying misbehaving participants.

\noindent\textbf{Summary of Related Work}. While existing aggregate signature schemes provide elegant primitives for Anonymous Credentials, they do not directly support epoch-based validity. Moreover, previous credential schemes typically treat all issuers equally, ignoring the heterogeneity of issuer authority. In PoS settings, solutions generally rely on unweighted multi-signatures or SPV-style proofs. Unweighted multi-signatures are computationally inexpensive but misalign with the security model of PoS, as they rely on a threshold of node counts rather than the accumulated stake weight. Conversely, while recent SNARK-based light clients achieve succinct verification, they introduce significant prover-side computational overhead. Traditional SPV approaches, meanwhile, require verifiers to track a growing chain of headers and validate inclusion proofs. This introduces non-trivial bandwidth overhead and synchronization dependency. Our MA-ACEW addresses these limitations by incorporating temporal validity and flexible issuer weighting tailored for PoS-based ecosystems.

\section{Conclusion}\label{conclusion}
In this paper, we introduced the Epoch-Bound Pointcheval-Sanders (EB-PS) signature primitive, which enables temporal binding for signatures. Building on a concrete EB-PS construction, we developed Multi-Authority Anonymous Credentials with Epoch-based Weights (MA-ACEW), the first decentralized anonymous credential system that incorporates weighted authority contributions. Our approach enables efficient credential updates when authority weight distributions change across epochs, addressing a significant limitation in existing systems that treat all authorities uniformly regardless of their relative trustworthiness or significance in the system.

Furthermore, we provided formal security definitions and rigorous proofs for both EB-PS and MA-ACEW constructions. To demonstrate practical viability, we implemented our schemes in Golang and developed smart contract components to evaluate the credential issuance and verification phases. Our results confirm that MA-ACEW achieves the necessary efficiency for deployment in decentralized systems such as Proof-of-Stake networks, where authority trust levels naturally fluctuate over time.

\section*{Acknowledgements}
This work was supported in part by the National Natural Science Foundation of China (NSFC) under Grants 12441101 and 62372324. We thank the anonymous reviewers and our shepherd for their constructive feedback, and Zhiqiang Ma, Gaowei Shi, and Chenyu Zhang for their insightful suggestions.




\section*{Ethical Considerations}
\label{sec:ethics}

In this section, we conduct a stakeholder analysis and discuss the potential societal impacts of our proposed weighted anonymous credential system. While our work focuses on the formal design and security analysis of a cryptographic protocol, we recognize that the deployment of such privacy-enhancing technologies requires a robust socio-technical context to be used safely.

\noindent\textbf{Stakeholder Analysis.}
We identify four key stakeholder groups within the ecosystem of privacy-preserving proof-of-stake and governance systems. First, \textbf{Users (Credential Holders)} are the primary beneficiaries. They utilize weighted credentials to access services or participate in voting without revealing their identities, protecting them from targeting, profiling, or coercion. Their primary interest lies in the robustness of the anonymity guarantee. Second, \textbf{Service Providers and Verifiers} rely on credentials for access control or vote tallying. They benefit from the security derived from stake-backed credentials without the liability of managing user identities, though they face the operational risk of being unable to hold individual users accountable for abuse. Third, \textbf{Credential Issuers} hold the stake and issue credentials. They benefit from the \textit{issuer-hiding} property and broader adoption driven by enhanced user privacy, but face reputational risk if credentials derived from their stake are misused. Finally, \textbf{System Operators and the Community} are concerned with the overall fairness and legitimacy of the system. While they benefit from increased participation, they could be harmed if the technology facilitates large-scale, untraceable manipulation.

\noindent\textbf{Potential Impacts and Misuse Considerations.}
The core objective of our construction is to enable privacy-preserving authentication. However, the strong privacy guarantees also introduce potential risks. On the \textbf{positive side}, our system empowers users to act without fear of surveillance or retaliation, fostering more honest participation in sensitive contexts like voting. By guaranteeing user anonymity and issuer hiding, the system mitigates surveillance-based coercion risks. On the \textbf{negative side}, the primary ethical risk is abuse by malicious actors. Without adequate safeguards, malicious users could exploit anonymity to engage in untraceable, coordinated manipulation or launch credential-based Sybil attacks. If unchecked, this could lead to a tragedy of the commons, eroding accountability and deterring honest participants.

\noindent\textbf{Mitigations and Deployment Safeguards.}
To address the risks associated with real-world application, we emphasize that our cryptographic protocol should not be deployed in isolation. We recommend specific architectural safeguards.

\noindent\textit{Separation of Protocol and Deployment Controls.} It is essential to distinguish between the cryptographic security model and deployment-layer controls. Our security model captures the core properties of anonymous credentials, but practical defense against Sybil attacks requires external bounds on per-entity issuance. We treat issuance-rate control as a necessary measure at the deployment layer; effective deployments must impose strict bounds on credential issuance per epoch to prevent abuse.

\noindent\textit{Layered Accountability Mechanisms.} Real-world systems must be designed with layered accountability mechanisms. These include \textbf{rate-limiting and economic costs} to make large-scale abuse prohibitively expensive, as well as \textbf{transparent governance and circuit-breakers}. For extreme scenarios, a decentralized process should be established to investigate and, as a last resort, respond to system-wide attacks, ensuring anonymity does not become an absolute shield for actors attempting to destroy the system itself.

\noindent\textit{Cryptographic Extensions.} A natural direction for future research is to explore privacy-preserving quotas directly within the cryptographic layer. Techniques such as $k$-times anonymous credentials (also known as $k$-show) could be integrated to cryptographically enforce usage limits while maintaining user privacy, thereby bridging the gap between protocol security and operational stability.


\section*{Open Science}
In adherence to the USENIX Security Symposium's open science policy, we have deposited our implementation in an open Zenodo repository (\url{https://doi.org/10.5281/zenodo.17905110}). 
Our repository consists of two primary components: (1) the core Golang implementation of the proposed EB-PS and MA-ACEW schemes (including cryptographic logic and ZK-proof constructions) located in \texttt{maacew/src}; and (2) the smart contract implementation for on-chain operations located in \texttt{maacew/contracts}. Detailed instructions on code structure, dependencies, and reproduction steps are provided in the \texttt{README} file within the root directory.
This ensures compliance with the conference's submission guidelines and enables transparent verification and reproduction of our results by the research community.

\bibliographystyle{plain}
\bibliography{myreference}

\appendix

\section{Assumptions}\label{Apen A3}

\begin{definition}[PS Assumption \cite{PS16short}] \label{PS assu}
	Consider an asymmetric pairing setting $(p, \mathbb{G}_1, \mathbb{G}_2, \mathbb{G}_T, u, v, e)$ with $(v^x, v^y) \in \mathbb{G}_2^2$ where $x$ and $y$ are random scalars in $\mathbb{Z}_p$. The PS assumption holds if no PPT adversary $\mathcal{A}$ with unlimited access to PS oracle $\mathcal{O}^{\mathsf{PS}}(m)$ can efficiently generate a tuple $(h^*, s^*, m^*)$ such that:
	\begin{itemize}[nosep, leftmargin=*] \label {GPS ass}
		\item $\mathcal{O}^{\mathsf{PS}}(m)$: On input $m \in \mathbb{F}_p$, chooses a random $h \in \mathbb{G}_1$ and outputs $(h, h^{x+my})$.
		\item $\mathcal{Q}$ is the list of queried messages to the $\mathcal{O}^{\mathsf{PS}}(m)$ oracle.
	\end{itemize}
\end{definition}

\begin{definition}[Generalized PS Assumption \cite{kim2021practical}]
	Consider an asymmetric pairing setting, given $(v^x, v^y) \in \mathbb{G}_2^2$, the GPS assumption is defined with respect to two oracles $\mathcal{O}_0^{\mathsf{GPS}}(\cdot)$ and $\mathcal{O}_1^{\mathsf{GPS}}(\cdot)$, where:
	\begin{itemize}[nosep, leftmargin=*]
		\item $\mathcal{O}_0^{\mathsf{GPS}}(\cdot)$ outputs a uniformly distributed $h \in \mathbb{G}_1.$
		\item $\mathcal{O}_1^{\mathsf{GPS}}(m,h)$ takes input $h \in \mathbb{G}_1$, $m \in \mathbb{F}_p$ and outputs $s = h^{x+m\cdot y}$. If $h \notin \mathcal{Q}_0 \vee (h, \star) \in \mathcal{Q}_1$, it outputs $\perp.$
	\end{itemize}
	The GPS assumption holds if no PPT adversary $\mathcal{A}$ can find a tuple $(h^*, s^*, m^*)$ such that:
	\[ h^* \neq 1_{\mathbb{G}_1},\quad s^* = (h^*)^{x+m^* \cdot y},\quad \text{and}\quad m^* \notin \mathcal{Q} \]
	where $\mathcal{Q}_1 = \mathcal{Q}_1 \cup (h,m)$ is the list of queries made to $\mathcal{O}_1^{\mathsf{GPS}}$ by adversary $\mathcal{A}$.
\end{definition}

\section{Weighted Threshold Signature(WTS)}

\subsection{Formal definition}\label{WTS def}
\begin{definition}[WTS]
	A Weighted Threshold Signature scheme ($\mathsf{WTS}$) consists of the following polynomial-time algorithms:
\end{definition}
\vspace{-1em}
\begin{itemize}[nosep, leftmargin=*]
	\item $\mathsf{Setup}(1^\lambda) \rightarrow \mathsf{pp}$: On input the security parameter $\lambda$, the algorithm outputs the public parameters $\mathsf{pp}$.
	
	\item $\mathsf{KGen}(\mathsf{pp}, n) \rightarrow (\{\mathsf{sk}_i, \mathsf{pk}_i, \mathsf{ak}_i\}_{i \in [n]}, \mathsf{vk})$: On input the public parameters $\mathsf{pp}$ and the total number of signers $n$, the algorithm outputs the global verification key $\mathsf{vk}$, and for each signer $i \in [n]$, a tuple $(\mathsf{sk}_i, \mathsf{pk}_i, \mathsf{ak}_i)$ consisting of a signing key, public key, and an aggregation key for IPA proof computation.
	
	\item $\mathsf{PSign}(\mathsf{sk}_i, m) \rightarrow \sigma_i$: On input a signing key $\mathsf{sk}_i$ and a message $m$, the algorithm outputs a partial signature $\sigma_i$.
	
	\item $\mathsf{PVerify}(m, \sigma_i, \mathsf{pk}_i) \rightarrow \{0, 1\}$: On input a message $m$, a partial signature $\sigma_i$, and a public key $\mathsf{pk}_i$, the algorithm outputs 1 if the partial signature is valid, and 0 otherwise.
	
	\item $\mathsf{CombPk}(\{\mathsf{ak}_i\}_{i \in [n]}, \bm{b}, \{\sigma_i\}_{b[i] = 1}, \mathsf{vk}) \rightarrow (\pi_{\mathsf{IPA,pk}}, \pi_b, \sigma)$: On input a set of aggregation keys $\{\mathsf{ak}_i\}_{i \in [n]}$, a binary vector $\bm{b}$ indicating participating signers (where $b_i = 1$ if signer $i$ participates, and 0 otherwise), a set of partial signatures $\{\sigma_i\}_{b[i] = 1}$ from participating signers, and the global verification key $\mathsf{vk}$, the algorithm outputs: an IPA proof $\pi_{\mathsf{IPA, pk}}$ for the inner product between the public keys and vector $\bm{b}$, a proof $\pi_b$ that $\bm{b}$ is binary, and the aggregate signature $\sigma$.
	
	\item $\mathsf{CombWt}(\mathsf{vk}, \bm{b}, \bm{w}) \rightarrow \pi_\mathsf{IPA,wt}$: On input the global verification key $\mathsf{vk}$, a binary vector $\mathbf{b}$ indicating participating signers, and a weight vector $\bm{w}$ containing the weights of all signers, the algorithm outputs an IPA proof $\pi_\mathsf{IPA,wt}$ for the inner product between the participation vector $\bm{b}$ and the weight vector $\bm{w}$.
	
	\item $\mathsf{MergePf}(\pi_{\mathsf{IPA, pk}}, \pi_{\mathsf{IPA,wt}}) \rightarrow \pi_{\mathsf{IPA}}$: On input an IPA proof $\pi_{\mathsf{IPA, pk}}$ for the inner product between the public keys and the participation vector $\bm{b}$, and an IPA proof $\pi_{\mathsf{IPA,wt}}$ for the inner product between the participation vector $\bm{b}$ and the weight vector $\bm{w}$, the algorithm combines them into a single merged IPA proof $\pi_{\mathsf{IPA}}$.
	
	\item $\mathsf{Verify}(m, \sigma, \mathsf{vk}, \pi_{\mathsf{IPA}}, \mathrm{W}_{\mathsf{acc}}) \rightarrow \{0,1\}$: On input a message $m$, an aggregate signature $\sigma$, the global verification key $\mathsf{vk}$, a merged IPA proof $\pi_{\mathsf{IPA}}$, and an acceptance threshold $\mathrm{W}_{\mathsf{acc}}$, the algorithm outputs 1 if the signature is valid and the accumulated weight meets the acceptance threshold $\mathrm{W}_{\mathsf{acc}}$, and 0 otherwise.
\end{itemize}

\subsection{WTS Construction}\label{wtscons}
For comprehensive technical foundations of the polynomial identities, we refer readers to~\cite{ben2019aurora, das2023threshold}. We now present the complete construction of our weighted threshold signature (WTS) scheme. As mentioned in Section~\ref{buildingblo}, we employ asymmetric pairings and separate the signature combination process into distinct phases. The detailed algorithms are as follows:
\begin{itemize}[leftmargin=*]
	\item $\mathsf{Setup}(1^\lambda) \rightarrow \mathsf{pp}$: On input the security parameter $\lambda$, first generates bilinear group parameters $\mathsf{BG} = (p, \mathbb{G}_1, \mathbb{G}_2, \mathbb{G}_T, u, v, e) \leftarrow \mathsf{BGGen}(1^\lambda)$, where $p$ is a prime order, $u \in \mathbb{G}_1$ and $v \in \mathbb{G}_2$ are generators, and $e: \mathbb{G}_1 \times \mathbb{G}_2 \rightarrow \mathbb{G}_T$ is a bilinear pairing. Then selects a hash function $\mathrm{H}: \{0,1\}^* \rightarrow \mathbb{G}_1$ and derives $\mathrm{H}_{\mathsf{FS}}(\cdot)$ using domain separation. The algorithm then performs:
	\begin{itemize}[leftmargin=10pt,topsep=0pt,itemsep=0pt,parsep=0pt]
		
		\item Let $\omega \in \mathbb{F}_p$ be a primitive $n$-th root of unity$,$ and define $H = \{\omega, \omega^2, \ldots, \omega^n\}$ as a multiplicative subgroup of order $n$ and $L$ as a coset of $H$ of size $n-1$, where $H \cap L = \varnothing$.
		
		\item Sample random generators $\theta \in \mathbb{G}_1$, $\alpha \in \mathbb{G}_2$ and a random $\tau \in \mathbb{F}_p$, and generate the following $\mathsf{CRS}$:
		\begin{gather*}
    \bm{u} := [u, u^\tau, \ldots, u^{\tau^n}],  \bm{v} := [v, v^\tau, \ldots, v^{\tau^n}], \\
    \bm{\alpha} := [\alpha, \alpha^\tau, \ldots, \alpha^{\tau^{n-1}}]
\end{gather*}
		
		\item Preprocess $\mathsf{CRS}$ as follows, where $\mathcal{L}_{i,H}(\tau)$ and $\mathcal{L}_{i,L}(\tau)$ denote the Lagrange polynomials defined over the multiplicative subgroups $H$ and $L$ respectively:
		\begin{equation*}
\small
\begin{aligned}
    \overrightarrow{v}_{\mathcal{L}} &:= [ v^{\mathcal{L}_{1,H}(\tau)}, \ldots, v^{\mathcal{L}_{n,H}(\tau)} ] \\
    \overrightarrow{\alpha}_{\mathcal{L}} &:= [ \alpha^{\mathcal{L}_{1,H}(\tau)}, \ldots, \alpha^{\mathcal{L}_{n,H}(\tau)} ] \\
    \overrightarrow{\beta} &:= [ v^{\mathcal{L}_{1,L}(\tau)}, \ldots, v^{\mathcal{L}_{n,L}(\tau)} ]
\end{aligned}
\end{equation*}
		\item Compute $v^\eta$ using $\overrightarrow{v}_{\mathcal{L}}$, where $\eta = \sum_{i \in [n]} \frac{\mathcal{L}_i(\tau)}{\omega^i}$.
	\end{itemize} 
	Output $\mathsf{pp} = (\mathsf{BG}, \mathsf{CRS}, \overrightarrow{v}_{\mathcal{L}}, \overrightarrow{\alpha}_{\mathcal{L}}, \overrightarrow{\beta}, \alpha, \theta, v^\eta)$.

	\item $\mathsf{KGen}(\mathsf{pp}, n, \bm{w}) \rightarrow (\{\mathsf{sk}_i, \mathsf{pk}_i, \mathsf{ak}_i\}_{i \in [n]}, \mathsf{vk})$: 
	The algorithm takes as input the public parameters $\mathsf{pp}$, the total number of signers $n$ and a vector of weights $\bm{w}$. The algorithm performs:
	\begin{itemize}[leftmargin=10pt,topsep=0pt,itemsep=0pt,parsep=0pt]
		\item Sample $x_i \stackrel{\$}{\leftarrow} \mathbb{F}_p$ for each $i \in [n]$ and generate the per-signer key pairs $(\mathsf{sk}_i, \mathsf{pk}_i) = (x_i, v^{x_i})$. For each signer $i$, compute the aggregation key:
		$$\mathsf{ak}_i = (v^{x_i}, v_i^{x_i}, \alpha_i^{x_i}, \theta^{x_i}, v^{\eta x_i}, \{\beta_k^{x_i}\}_{k \in [n]})$$
		using $\overrightarrow{v}_{\mathcal{L}}, \overrightarrow{\alpha}_{\mathcal{L}}$, and $\overrightarrow{\beta}$.
		\item Compute the global verification key: $\mathsf{vk} = (v, \alpha, \beta, v^{x(\tau)}, v^{w(\tau)}, v^{\tau}, $ $ \alpha^{\tau}, u^{z_H(\tau)})$, where:
		$$v^{x(\tau)} = \prod_{i \in [n]} v_i^{x_i}, v^{w(\tau)} = \prod_{i \in [n]} v_i^{w_i}.$$
	\end{itemize}
	The vanishing polynomial over $H$ is defined as $$z_H(\tau) = \prod_{i \in [n]} (\tau - \omega^i) = \tau^n - 1.$$
	\vspace{-3ex}
	
	\item $\mathsf{CombPk}(\bm{b}, \{\mathsf{ak}_i\}_{i \in [n]}, \mathsf{vk}) \rightarrow (\pi_b, \pi_{\mathsf{IPA, pk}})$: On input a signing set vector $\bm{b}$, the aggregation keys $\{\mathsf{ak}_i\}_{i \in [n]}$, compute the following:
	\begin{itemize}[leftmargin=10pt,topsep=0pt,itemsep=0pt,parsep=0pt]
		\item Compute commitment to $\mathsf{b}$ as $u_b = u^{b(\tau)}$ and the proof $\pi_b = v^{q_b(\tau)}$ that $\mathsf{b}$ is binary using the polynomial remainder lemma:
		$$b(\tau)\cdot(1 - b(\tau)) = q_b(\tau)\cdot z_H(\tau).$$
		\item Compute the aggregated public key $\mathrm{X}$ and the IPA proof $\pi_{\mathsf{IPA},\mathsf{pk}} = \{v^{q_x(\tau)}, v^{r_x(\tau)}, \alpha^{p_x(\tau)}, \theta^x\}$, where:
		$$ \mathrm{X} = v^x = \langle\bm{pk}, \bm{b}\rangle.$$
	\end{itemize}
	
	\item $\mathsf{CombWt}(\bm{w}, \{\mathsf{ak}_i\}_{i \in [n]}) \rightarrow (\pi_{\mathsf{IPA, wt}}, \mathrm{W}_{\mathsf{claim}})$: On input a set of weight vectors $\bm{w}$ corresponding to $b[i]=1$, and the aggregation keys $\{\mathsf{ak}_i\}_{i \in [n]}$, compute $\mathrm{W}_{\mathsf{claim}} = \langle\bm{w}, \bm{b}\rangle$ and the IPA proof: $$\pi_{\mathsf{IPA},\mathsf{wt}} = \{v^{q_w(\tau)}, v^{r_w(\tau)}, \alpha^{p_w(\tau)}, \theta^{\mathrm{W}_{\mathsf{claim}}}\}$$
	
	\item $\mathsf{MergePf}(\pi_{\mathsf{IPA,pk}}, \pi_{\mathsf{IPA,wt}}, \mathrm{X}, \mathrm{W}_{\mathsf{claim}}) \rightarrow \pi_{\mathsf{IPA}}$: 
	On input the IPA proofs $\pi_{\mathsf{IPA,pk}}$ and $\pi_{\mathsf{IPA,wt}}$ for signing and weight vectors respectively, and aggregated values $(\mathrm{X}, \mathrm{W}_{\mathsf{claim}})$, compute the merged proof using $\mathrm{H}_{\mathsf{FS}}$ as follows:
	
	\begin{itemize}[leftmargin=15pt,topsep=2pt,itemsep=2pt,parsep=1pt]
		\item Compute
		\[\xi = \mathrm{H}_{\mathsf{FS}}(v^{x(\tau)}, v^{w(\tau)}, u^{b(\tau)}, \mathrm{X}, \mathrm{W}_{\mathsf{claim}})\]
		
		\item Compute
		\begin{align*}
			q_o(\tau) &= q_x(\tau) + \xi \cdot q_w(\tau) \\
			r_o(\tau) &= r_x(\tau) + \xi \cdot r_w(\tau) \\
			p_o(\tau) &= p_x(\tau)\cdot \tau + (x + \xi \cdot \mathrm{W}_{\mathsf{claim}})\cdot n^{-1}
		\end{align*}
	\end{itemize}
	
	The merged proof is:
	\[
	\pi_{\mathsf{IPA}} = (v^{q_o(\tau)}, v^{r_o(\tau)}, \alpha^{p_o(\tau)}, \theta^x, \mathrm{W}_{\mathsf{claim}})
	\]
	
	\item $\mathsf{Verify}(\pi_{\mathsf{IPA}},\mathsf{vk}) \rightarrow \{0,1\}$: 
	On input an IPA proof $\pi_{\mathsf{IPA}}$ and verification key $\mathsf{vk}$, output 1 if  all following equations hold:
	
	\begin{itemize}[leftmargin=15pt,topsep=2pt,itemsep=3pt,parsep=1pt]
		\item Check the correctness of the bit vector:
		\begin{equation}
			e(u^{b(\tau)} \cdot u^{1 - b(\tau)}, v) = e(u^{z_H(\tau)}, v^{q_b(\tau)}) \nonumber
		\end{equation}
		
		\item Compute challenge:
		\[
		\xi = \mathrm{H}_{\mathsf{FS}}(v^{x(\tau)}, v^{\mathrm{w}(\tau)}, u^{b(\tau)}, \mathrm{X}, \mathrm{W}_{\mathsf{claim}})
		\]
		
		\item Check if the following equations hold:
		\begin{align}
    \scriptsize e(u^{b(\tau)}, v^{x(\tau)}) &= \scriptsize e(u^{z_H(\tau)}, v^{q_o(\tau)}) \notag \\
    &\qquad \cdot e(u^{\tau}, v^{r_o(\tau)}) \notag \\
    &\qquad \cdot e(u^{1/n}, v^{x} \cdot v^{\xi \cdot \mathrm{W}_{\mathsf{claim}}}) \nonumber\\[0.5em]
    \scriptsize e(u, \alpha^{p_o(\tau)}) &= \scriptsize e(u^{\tau}, v^{r_o(\tau)}) \notag \\
    &\qquad \cdot e(u^{1/n}, v^{\xi \cdot \mathrm{W}_{\mathsf{claim}}} \cdot \mathrm{X}) \nonumber\\[0.5em]
    \scriptsize e(\theta^{x}, v) &= \scriptsize e(\theta, \mathrm{X}) \nonumber
\end{align}
	\end{itemize}
\end{itemize}

Here, we omit the signing phase, as in our MA-ACEW constructions, we need to replace the original BLS signature with our proposed EB-PS to satisfy the properties defined in Section \ref{sec52}.

\section{Correctness of EB-PS}\label{correctness}

\subsection{Formal Definition}\label{def corr}
 We now formalize the basic correctness of EB-PS. The correctness property ensures that honestly generated signatures will always be accepted by the verification algorithm.
Formally,
\begin{equation}
	\Pr\left[ \begin{array}{c} 
		\forall i \in [\ell], \forall j \in [\mathrm{T}]: \\(\mathsf{lsk}_i, \mathsf{lvk}_i) \leftarrow \mathsf{KGen}(\mathsf{pp}); \\ 
		(\mathsf{tsk}_{i,j}, \mathsf{tvk}_{i,j}) \leftarrow \mathsf{TKGen}(\mathsf{pp}, j); \\ 
		\sigma_{\mathsf{lt},i} \leftarrow \mathsf{SignLt}(\mathsf{lsk}_i, m_i, \mathsf{aux}, \mathsf{tg}); \\ 
		\sigma_{\mathsf{ep},i,j} \leftarrow \mathsf{SignEp}(j, \mathsf{tg}, \mathsf{tsk}_{i,j}, F(\mathsf{ctx}, j)); \\ 
		\sigma_{i,j} \leftarrow \mathsf{CombSig}(j, \mathsf{tg}, \sigma_{\mathsf{lt},i}, \sigma_{\mathsf{ep}, i, j}); \\ 
		\mathsf{Verify}(j, \mathsf{tg}, \mathsf{lvk}_i, \mathsf{tvk}_{i,j}, m_i, F(\mathsf{ctx}, j), \sigma_{i,j}) = 1
	\end{array} \right] = 1.\nonumber
\end{equation}

\noindent\textbf{Aggregation Correctness.} We next define the aggregation correctness of  EB-PS, which ensures that properly aggregated signatures remain verifiable. Formally:
\begin{equation}
\Pr\left[
    \begin{gathered}
        \forall j \in [\mathrm{T}]: \\
        \sigma_{\mathsf{agg, lt}} \leftarrow \mathsf{AggSigLt}\big(\mathsf{tg}, \{\mathsf{lvk}_i, m_i, \sigma_{\mathsf{lt},i}\}_{i=1}^\ell\big); \\
        (\sigma_{\mathsf{agg, ep},j}, \mathsf{avk}_{\mathsf{ep},j}) \leftarrow \mathsf{AggSigEp}\big(j, \mathsf{tg}, \\
        \qquad \{\mathsf{tvk}_{i,j}, \sigma_{\mathsf{ep}, i, j}\}_{i=1}^\ell, F(\mathsf{ctx},j)\big); \\
        \sigma_{\mathsf{agg},j} \leftarrow \mathsf{AggCombine}(\sigma_{\mathsf{agg, lt}}, \sigma_{\mathsf{agg, ep},j}); \\
        \mathsf{AggVerify}(j, \mathsf{tg}, \mathsf{avk}_j, \mathbb{M}, \sigma_{\mathsf{agg},j}, F(\mathsf{ctx}, j)) = 1
    \end{gathered}
\right] = 1. \nonumber
\end{equation}

\subsection{Proof of correctness }\label{proof of corr}

We now prove the correctness of our EB-PS construction as follows:

\noindent\textbf{Basic Correctness.} For basic correctness, we observe that the combined signature takes the form:
\[\sigma_{i,j} = \left(h^{\gamma}, \left(h^{\gamma}\right)^{x_i + m_{i}\cdot y_i} \cdot \left(h^{\delta}\right)^{F(\mathsf{ctx},j)\cdot z_{i,j}}\right)\]

The left-hand side of the verification equation evaluates as:
\begin{align*}
	& e(h^{\gamma}, \mathrm{X}_i \cdot \mathrm{Y}_i^{m_{i}}) \cdot e((h^{\delta})^{F(\mathsf{ctx},j)}, \mathrm{Z}_{i,j}) \\
	& = e(h^{\gamma}, v^{x_i}\cdot (v^{y_i})^{m_{i}}) \cdot e((h^{\delta})^{F(\mathsf{ctx},j)}, v^{z_{i,j}}) \\
	& = e(h,v)^{\gamma\cdot (x_i + m_{i}\cdot y_i) + \delta\cdot F(\mathsf{ctx}, j)\cdot z_{i,j}}
\end{align*}

The right-hand side evaluates as: 
\[e(s_i, v) = e(h,v)^{\gamma \cdot (x_i + m_{i}\cdot y_i) + \delta \cdot F(\mathsf{ctx}, j)\cdot z_{i,j}}\]

Since both sides yield the same expression, this establishes the correctness of the scheme.

\noindent\textbf{Aggregation Correctness.} For aggregatable correctness, we observe that the combined signature takes the form:
\[\sigma_{\mathsf{agg},j} = \left(h^{\gamma}, \left(h^{\gamma}\right)^{\sum_{i \in [\ell]} x_i + m_{i}\cdot y_i} \cdot \left(h^{\delta}\right)^{F(\mathsf{ctx},j)\cdot \sum_{i \in [\ell]} z_{i,j}}\right)\]

The left-hand side of the verification equation evaluates as:
\begin{align*}
	& e(h^{\gamma}, \prod_{i \in [\ell]} \mathrm{X}_i \cdot \mathrm{Y}_i^{m_{i}}) \cdot e((h^{\delta})^{F(\mathsf{ctx},j)}, \prod_{i \in [\ell]} \mathrm{Z}_{i, j}) \\
	& = e(h^{\gamma}, v^{\sum_{i \in [\ell]} x_i} \cdot \prod_{i \in [\ell]} (v^{y_i})^{m_{i}}) \cdot e((h^{\delta})^{F(\mathsf{ctx},j)}, v^{\sum_{i \in [\ell]} z_{i,j}}) \\
	& = e(h,v)^{\gamma\cdot \sum_{i \in [\ell]}(x_i + m_{i}\cdot y_i) + \delta\cdot F(\mathsf{ctx}, j) \sum_{i \in [\ell]} z_{i,j}}
\end{align*}

The right-hand side evaluates as: 
\[e(s, v) = e(h,v)^{\gamma \cdot \sum_{i \in [\ell]} (x_i + m_{i}\cdot y_i) + \delta \cdot F(\mathsf{ctx}, j)\cdot \sum_{i \in [\ell]}z_{i,j}}\]

Since both sides yield the same expression, this establishes the correctness of the scheme.

\section{Security Proofs} 

\subsection{Proof of Theorem \ref{theo}}\label{hardp}
Let us assume an adversary $\mathcal{A}$ produces a valid forgery $(j^*, m^*, h^*, \mathsf{ctx}^*, s^*)$ after making a total of $q$ queries to the assumption's oracles ($\mathcal{O}_h, \mathcal{O}_{\text{sign}}, \mathcal{O}_{\text{update}}, \mathcal{O}_{\text{corrupt}}$) and $q_G$ queries to the group operation oracles. By the definition of a valid forgery, the target epoch $j^*$ has not been corrupted, i.e., $j^* \notin \mathcal{Q}_{\text{corrupt}}$. We now analyze the structure of this forgery in the Generic Group Model (GGM).

\textbf{Setup and Oracle Simulation.}
In the GGM, we associate group elements with formal polynomials over a set of indeterminates. $\mathcal{A}$ is given handles to the public parameters and oracle outputs. The oracles are simulated as follows:
\begin{itemize}
    \item  $\mathcal{A}$ queries the oracle $\mathcal{O}_h$ and receives a group element $h_k \in \mathbb{G}_1$ represented by a new random polynomial $r_k$.

    \item $\mathcal{A}$ queries the oracle $\mathcal{O}_{\text{sign}}$ on a tuple $(j, m, \mathsf{ctx}, h)$ and receives handles to the two polynomials representing the signature components: $P_{s_{\mathsf{lt}}} = P_{h} \cdot (x + m y)$ and $P_{s_{\mathsf{ep}, j}} = P_h \cdot (F(\mathsf{ctx}, j) \cdot z_j)$.

\item $\mathcal{A}$ queries the oracle $\mathcal{O}_{\text{update}}$ on a tuple $(j, m, h, \mathsf{ctx})$ and receives a handle to the polynomial for the new epoch component: $P_{s_{\mathsf{ep}, j+1}} = P_h \cdot (F(\mathsf{ctx}, j+1) \cdot z_{j+1})$.

    \item $\mathcal{A}$ queries the oracle $\mathcal{O}_{\text{corrupt}}$ on an epoch $j$ and receives the secret polynomial $z_j$.
\end{itemize}

\textbf{Polynomial Representation of the Forgery.}
The GGM principle dictates that any group element computed by $\mathcal{A}$ corresponds to a polynomial that is a linear combination of the polynomials for all elements it already possesses. The adversary's goal is to output a forged signature pair $(s_{\mathsf{lt}}^*, s_{\mathsf{ep},j^*}^*)$. This means $\mathcal{A}$ must construct two corresponding polynomials, $P_{s_{\mathsf{lt}}}^*$ and $P_{s_{\mathsf{ep},j^*}}^*$.

Therefore, for $\mathcal{A}$ to construct the two components of its forgery, their respective polynomials, $P_{s_{\mathsf{lt}}}^*$ and $P_{s_{\mathsf{ep},j^*}}^*$, must be expressible as linear combinations of this basis. We use distinct coefficients for each component to reflect that they are constructed independently:

\begin{align}
    P_{s_{\mathsf{lt}}}^* &= \alpha_{\mathsf{lt}} + \beta_{\mathsf{lt}} y + \sum_{k=1}^{q_h} \delta_k P_{h_k} \nonumber \\
    &\qquad + \sum_{i=1}^{q_s} \gamma_i P_{s_{\mathsf{lt}},i} + \sum_{i=1}^{q_s+q_u} \eta_i P_{s_{\mathsf{ep}},i} + \sum_{j \in \mathcal{Q}_{\text{corrupt}}} c_j z_j \label{eq:adv-poly-lt} \\
    \nonumber \\
    P_{s_{\mathsf{ep},j^*}}^* &= \alpha_{\mathsf{ep}} + \beta_{\mathsf{ep}} y + \sum_{k=1}^{q_h} \delta'_k P_{h_k} \nonumber \\
    &\qquad + \sum_{i=1}^{q_s} \gamma'_i P_{s_{\mathsf{lt}},i} + \sum_{i=1}^{q_s+q_u} \eta'_i P_{s_{\mathsf{ep}},i} + \sum_{j \in \mathcal{Q}_{\text{corrupt}}} c'_j z_j \label{eq:adv-poly-ep}
\end{align}

\textbf{The Verification Constraint.}
For the output pair $(s_{\mathsf{lt}}^*, s_{\mathsf{ep},j^*}^*)$ to be a valid forgery for the tuple $(m^*, \mathsf{ctx}^*, h^*)$ at the uncorrupted target epoch $j^*$, it must satisfy the verification equation:
$e(s_{\mathsf{lt}}^* \cdot s_{\mathsf{ep},j^*}^*, v) = e(h^*, v^x \cdot v^{m^*y} \cdot v^{F(\mathsf{ctx}^*, j^*)z_{j^*}})$.

In the GGM, where group operations correspond to polynomial additions, this verification equation translates into a specific condition on the sum of the corresponding polynomials, $P_{s_{\mathsf{lt}}}^*$ and $P_{s_{\mathsf{ep},j^*}}^*$. The target polynomial they must sum to is:
\begin{equation} \label{eq:target-poly-final}
P_{\text{target}} = P_{h^*} \cdot (x + m^*y + F(\mathsf{ctx}^*, j^*)z_{j^*})
\end{equation}
Thus, the core constraint for a valid forgery is the polynomial identity:
\begin{equation} \label{eq:main-identity-new}
P_{s_{\mathsf{lt}}}^* + P_{s_{\mathsf{ep},j^*}}^* = P_{h^*} \cdot (x + m^*y + F(\mathsf{ctx}^*, j^*)z_{j^*})
\end{equation}

As argued previously, for the verification to be possible, the handle $h^*$ must be an explicit output from $\mathcal{O}_h$. Let us assume $h^*$ was the output of the $k^*$-th distinct hash query. Its polynomial is therefore $P_{h^*} = r_{k^*}$, where $r_{k^*}$ is a fresh, random polynomial. By definition of a valid forgery, the target epoch $j^*$ is uncorrupted, i.e., $j^* \notin \mathcal{Q}_{\text{corrupt}}$. Consequently, $z_{j^*}$ is an indeterminate unknown to $\mathcal{A}$, and its value is independent of all other polynomials known by $\mathcal{A}$.

\textbf{The Algebraic Contradiction.}
To prove that no adversary $\mathcal{A}$ can produce a valid forgery, we will now substitute the general forms of the adversary's constructed polynomials into the main verification identity (Eq.~\eqref{eq:main-identity-new}) and analyze the resulting algebraic constraints.

\paragraph{Polynomial Forms for the Forgery.}
As established, for a forgery to be verifiable, the adversary must choose a handle $h^*$ that was an explicit output of the hash oracle $\mathcal{O}_h$. Let us assume $h^* = h_{k^*}$ for some $k^* \in \{1, \dots, q_h\}$. This constrains its polynomial to the simple, non-composite form:
\begin{equation} \label{eq:h-constrained}
P_{h^*} = P_{h_{k^*}} = r_{k^*}
\end{equation}
where $r_{k^*}$ is a fresh random polynomial whose value is unknown to $\mathcal{A}$.

The two forged signature components, $s_{\mathsf{lt}}^*$ and $s_{\mathsf{ep},j^*}^*$, on the other hand, are constructed by $\mathcal{A}$ from all available information. Their corresponding polynomials, $P_{s_{\mathsf{lt}}}^*$ and $P_{s_{\mathsf{ep},j^*}}^*$, can therefore be expressed as two distinct linear combinations of all polynomials known to the adversary. These are precisely the general forms we defined previously, which we restate here for clarity:
\begin{align}
    P_{s_{\mathsf{lt}}}^* &= \alpha_{\mathsf{lt}} + \beta_{\mathsf{lt}} y + \sum_{k=1}^{q_h} \delta_k P_{h_k} \nonumber \\
    &\qquad + \sum_{i=1}^{q_s} \gamma_i P_{s_{\mathsf{lt},i}} + \sum_{i=1}^{q_s+q_u} \eta_i P_{s_{\mathsf{ep},i}} + \sum_{j \in \mathcal{Q}_{\text{corrupt}}} c_j z_j \label{eq:adv-poly-lt-restate} \\
    \nonumber \\
    P_{s_{\mathsf{ep},j^*}}^* &= \alpha_{\mathsf{ep}} + \beta_{\mathsf{ep}} y + \sum_{k=1}^{q_h} \delta'_k P_{h_k} \nonumber \\
    &\qquad + \sum_{i=1}^{q_s} \gamma'_i P_{s_{\mathsf{lt},i}} + \sum_{i=1}^{q_s+q_u} \eta'_i P_{s_{\mathsf{ep},i}} + \sum_{j \in \mathcal{Q}_{\text{corrupt}}} c'_j z_j \label{eq:adv-poly-ep-restate}
\end{align}
The coefficients in these combinations (the $\alpha, \beta, \delta, \gamma, \eta, c$ values) are all chosen by the adversary $\mathcal{A}$.

\paragraph{The Main Polynomial Identity.}
A valid forgery, consisting of the tuple $(m^*, \mathsf{ctx}^*, h^*)$ and the signature pair $(s_{\mathsf{lt}}^*, s_{\mathsf{ep},j^*}^*)$, must satisfy the core verification identity from Eq.~\eqref{eq:main-identity-new}:
$P_{s_{\mathsf{lt}}}^* + P_{s_{\mathsf{ep},j^*}}^* = P_{h^*} \cdot (x + m^*y + F(\mathsf{ctx}^*, j^*)z_{j^*})$.

As established, this requires the handle polynomial to be of the simple form $P_{h^*} = r_{k^*}$. We now substitute this constraint, along with the general expressions for the two forged signature components (Eqs.~\eqref{eq:adv-poly-lt-restate} and \eqref{eq:adv-poly-ep-restate}), into the verification identity. By grouping the coefficients of like terms from the two forged polynomials, we arrive at the central relation for our analysis:
\begin{equation} \label{eq:main-identity-combined}
\begin{split}
    & (\alpha_{\mathsf{lt}} + \alpha_{\mathsf{ep}}) + (\beta_{\mathsf{lt}} + \beta_{\mathsf{ep}}) y + \sum_{k=1}^{q_h} (\delta_k + \delta'_k) P_{h_k} \\
    & \qquad + \sum_{i=1}^{q_s} (\gamma_i + \gamma'_i) P_{s_{\mathsf{lt},i}} \\
    & \qquad + \sum_{i=1}^{q_s+q_u} (\eta_i + \eta'_i) P_{s_{\mathsf{ep},i}} \\
    & \qquad + \sum_{j \in \mathcal{Q}_{\text{corrupt}}} (c_j + c'_j) z_j \\
    & = r_{k^*} \cdot (x + m^*y + F(\mathsf{ctx}^*, j^*)z_{j^*})
\end{split}
\end{equation}
This equation represents the fundamental constraint that the adversary must satisfy. The left-hand side represents the total polynomial the adversary can construct, while the right-hand side represents the target polynomial required for a valid forgery.

\paragraph{Analysis of Coefficients.}
The analysis proceeds by comparing the coefficients of the indeterminates on both sides of our main polynomial identity, Eq.~\eqref{eq:main-identity-combined}. To do this, we must first substitute the definitions of the signature component polynomials issued by the oracles:
\begin{itemize}
    \item Long-term components: $P_{s_{\mathsf{lt},i}} = P_{h_i} \cdot (x + m_i y )$
    \item Epoch-specific components: $P_{s_{\mathsf{ep}},i} = P_{h_i} \cdot F(\mathsf{ctx}_i, j) z_{j}$
\end{itemize}
Substituting these into the left-hand side (LHS) of Eq.~\eqref{eq:main-identity-combined} allows us to group terms by the indeterminates $x, y,$ and the various $z_j$. We then equate the resulting coefficients with those on the right-hand side (RHS), which are determined by the target polynomial $r_{k^*} \cdot (x + m^*y + F(\mathsf{ctx}^*, j^*)z_{j^*})$.

The main identity is:
\begin{equation} \label{eq:master-identity-expanded}
\begin{split}
    & \underbrace{
    \begin{aligned}
        &(\alpha_{\mathsf{lt}} + \alpha_{\mathsf{ep}}) + (\beta_{\mathsf{lt}} + \beta_{\mathsf{ep}}) y + \sum_{k=1}^{q_h} (\delta_k + \delta'_k) P_{h_k} \\
        &\quad + \sum_{i=1}^{q_s} (\gamma_i + \gamma'_i) \big[ P_{h_i} x \big] \\
        &\quad + \sum_{i=1}^{q_s+q_u} (\eta_i + \eta'_i) \big[ P_{h_i} (m_i y + F_i z_{j_i}) \big] \\
        &\quad + \sum_{j \in \mathcal{Q}_{\text{corrupt}}} (c_j + c'_j) z_j
    \end{aligned}
    }_{\text{LHS: Adversary's Constructed Polynomial (Sum of Forged Components)}} \\
    & \qquad\qquad = \underbrace{
    \begin{aligned}
        r_{k^*} x + r_{k^*} m^* y + r_{k^*} F^* z_{j^*}
    \end{aligned}
    }_{\text{RHS: Target Forgery Polynomial}}
\end{split}
\end{equation}

We now equate the coefficients of each indeterminate on both sides of this identity.

\textbf{1. Coefficients of the indeterminate $x$:}
The indeterminate $x$ is special because it only appears in the long-term signature components, $P_{s_{\mathsf{lt},i}}$.
\begin{itemize}
    \item \textbf{LHS coefficient of $x$}: This term arises exclusively from the expansion of $\sum (\gamma_i + \gamma'_i) P_{s_{\mathsf{lt},i}} = \sum (\gamma_i + \gamma'_i) P_{h_i} x$. The resulting coefficient of $x$ is thus $\sum_{i=1}^{q_s} (\gamma_i + \gamma'_i) P_{h_i}$.
    \item \textbf{RHS coefficient of $x$}: The coefficient is clearly $r_{k^*}$.
\end{itemize}
Equating these gives the formal identity: $\sum_{i=1}^{q_s} (\gamma_i + \gamma'_i) P_{h_i} = r_{k^*}$.
Since each $P_{h_i} = r_i$ is a distinct random indeterminate (and $r_{k^*}$ is one of them), this equality can only hold if the coefficients match exactly. This forces:
\begin{itemize}
    \item $\gamma_{k^*} + \gamma'_{k^*} = 1$
    \item $\gamma_i + \gamma'_i = 0$ for all $i \neq k^*$
\end{itemize}

\textbf{2. Coefficients of the uncorrupted secret $z_{j^*}$:}
This is the most critical step, as it involves the secret $z_{j^*}$ for the uncorrupted target epoch $j^*$, which is unknown to the adversary.
\begin{itemize}
    \item \textbf{LHS coefficient of $z_{j^*}$}: The term $z_{j^*}$ can only arise from the expansion of epoch-specific components $\sum (\eta_i + \eta'_i) P_{s_{\mathsf{ep},i}}$ for those queries $i$ that were made in the target epoch (i.e., where $j_i = j^*$). The resulting coefficient is $\sum_{i \text{ s.t. } j_i=j^*} (\eta_i + \eta'_i) F_i P_{h_i}$. Note that the term $\sum (c_j+c'_j)z_j$ does not contribute, as by definition of a valid forgery, $j^* \notin \mathcal{Q}_{\text{corrupt}}$.
    \item \textbf{RHS coefficient of $z_{j^*}$}: The coefficient is $F(\mathsf{ctx}^*, j^*) r_{k^*} = F^* P_{h_{k^*}}$.
\end{itemize}
Equating these gives: $\sum_{i \text{ s.t. } j_i=j^*} (\eta_i + \eta'_i) F_i P_{h_i} = F^* P_{h_{k^*}}$.
Again, by comparing the coefficients of the indeterminates $P_{h_i}$, we are forced to conclude:
\begin{itemize}
    \item $(\eta_{k^*} + \eta'_{k^*}) F_{k^*} = F^*$, which requires that the query $k^*$ must have been for the target epoch, so $j_{k^*} = j^*$.
    \item $(\eta_i + \eta'_i) F_i = 0$ for all other queries $i \neq k^*$ made in epoch $j^*$.
\end{itemize}

\textbf{3. Coefficients of the indeterminate $y$:}
\begin{itemize}
    \item \textbf{LHS coefficient of $y$}: This comes from two places: the standalone term $(\beta_{\mathsf{lt}} + \beta_{\mathsf{ep}})$ and the expansion of the epoch-specific components, which contributes $\sum (\eta_i + \eta'_i) m_i P_{h_i}$. The total coefficient is $(\beta_{\mathsf{lt}} + \beta_{\mathsf{ep}}) + \sum_{i=1}^{q_s+q_u} (\eta_i + \eta'_i) m_i P_{h_i}$.
    \item \textbf{RHS coefficient of $y$}: The coefficient is $m^* r_{k^*} = m^* P_{h_{k^*}}$.
\end{itemize}
Equating these and comparing coefficients of the $P_{h_i}$ indeterminates (and the constant term) yields:
\begin{itemize}
    \item $(\eta_{k^*} + \eta'_{k^*}) m_{k^*} = m^*$
    \item $(\eta_i + \eta'_i) m_i = 0$ for all $i \neq k^*$
    \item $\beta_{\mathsf{lt}} + \beta_{\mathsf{ep}} = 0$
\end{itemize}

\textbf{4. The Final Contradiction:}
Let's assemble our findings. The coefficient analysis has forced the adversary's choices to satisfy the following conditions simultaneously for some query index $k^*$:
\begin{enumerate}
    \item The query $k^*$ must have been made for the target epoch: $j_{k^*} = j^*$.
    \item The forged message $m^*$ and context $F^*$ must be related to the query's message $m_{k^*}$ and context $F_{k^*}$ via a single scaling factor $C = (\eta_{k^*} + \eta'_{k^*})$. Specifically, $m^* = C \cdot m_{k^*}$ and $F^* = C \cdot F_{k^*}$.
\end{enumerate}
By definition, a forgery requires that the adversary produces a signature on a message-context pair $(m^*, \mathsf{ctx}^*)$ for which it has not previously requested a signature.

However, from our analysis, if we assume $F$ behaves as a random oracle, then $F^* = C \cdot F_{k^*}$ implies that either $C=1$ and $\mathsf{ctx}^* = \mathsf{ctx}_{k^*}$, or the adversary has found a non-trivial linear dependency in the outputs of the random oracle, which is impossible except with negligible probability.

If we must have $C=1$, then the conditions become:
\begin{itemize}
    \item $m^* = m_{k^*}$
    \item $F^* = F_{k^*} \implies F(\mathsf{ctx}^*, j^*) = F(\mathsf{ctx}_{k^*}, j_{k^*})$. Given $j^*=j_{k^*}$, this implies $\mathsf{ctx}^* = \mathsf{ctx}_{k^*}$.
\end{itemize}
This means the forged tuple $(m^*, \mathsf{ctx}^*)$ is identical to the tuple $(m_{k^*}, \mathsf{ctx}_{k^*})$, which was the subject of the $k^*$-th signature query. This directly contradicts the condition that the forgery must be for a new, un-queried message-context pair.

Therefore, no such set of adversary-chosen coefficients can exist, and the adversary cannot construct a valid forgery. The remaining coefficients must all be zero to satisfy the rest of the main identity, reinforcing that no deviation is possible.

\paragraph{Bounding the Probability of Accidental Collisions.}
The final way for $\mathcal{A}$ to win is if two formally distinct polynomials, $P_A$ and $P_B$, happen to evaluate to the same value over $\mathbb{F}_p$ once the secret variables are instantiated, causing the algebraic constraints to break down. We can bound the probability of such an "accidental collision" using the Schwartz-Zippel lemma.

First, we count the total number of distinct polynomials available to $\mathcal{A}$. These come from:
\begin{itemize}
    \item The 2 base polynomials $\{1, y\}$ derived from the public parameters. The secret polynomial $x$ is never known to $\mathcal{A}$ in isolation.
    \item The polynomials obtained from oracle queries. This pool includes $q_h$ hash polynomials $\{P_{h_k}\}$, $q_s$ long-term signature polynomials $\{P_{s_{\mathsf{lt},i}}\}$, $q_s+q_u$ epoch-specific signature polynomials $\{P_{s_{\mathsf{ep},i}}\}$, and $q_c$ corrupted secret key polynomials $\{z_j\}$. Let $q = q_h + 2q_s + q_u + q_c$ be the total number of distinct polynomials from oracles.
    \item Up to $q_G$ additional polynomials generated by $\mathcal{A}$ through its own computations (linear combinations of existing polynomials).
\end{itemize}

By the Schwartz-Zippel lemma, the probability of any single non-trivial polynomial identity $P_A - P_B = 0$ holding true is at most $d_{\text{max}}/p$. We apply a union bound over all possible pairs of distinct polynomials to bound the total probability of any such collision occurring:
\begin{align}
\Pr[\mathcal{A} \text{ wins }] &\le \binom{2 + q + q_G}{2} \cdot \frac{d_{\text{max}}}{p} \nonumber \\
& = \frac{(2 + q + q_G)(2 + q + q_G-1)}{p} < \frac{(2 + q + q_G)^2}{p} \label{eq:sz-win-prob-final}
\end{align}

\subsection{Proof of Theorem \ref{unfeb}}\label{unforge of ebps}
We now present the formal security analysis of the unforgeability property of EB-PS.

\begin{proof}
	We prove the EUF-eCMA security of our scheme by constructing a reduction algorithm $\mathcal{R}$ that uses any adversary $\mathcal{A}$ against the scheme to break the STB-GPS assumption. The reduction $\mathcal{R}$ interacts with the adversary $\mathcal{A}$ on one side, and the challenger $\mathcal{C}$ of the STB-GPS security game on the other.  The security of our scheme is established using techniques similar to those in the security proofs for multi-message signatures, such as in~\cite{PS16short} and~\cite{DBLP:conf/ccs/MirBGLS23}. Specifically, our proof follows the chosen-key simulation paradigm, where $\mathcal{R}$ uses the challenge from $\mathcal{C}$ to generate a simulated key for $\mathcal{A}$, and later extracts a solution to the STB-GPS problem from $\mathcal{A}$'s forgery.

\noindent\textit{Setup.}
The simulator $\mathcal{R}$ receives a challenge from the STB-GPS challenger $\mathcal{C}$. This consists of the public parameters $\mathsf{pp}$ (containing groups $\mathbb{G}_1, \mathbb{G}_2$ with generators $u,v$), a public key $\mathsf{PK_C} = (\mathrm{X}_C, \mathrm{Y}_C, \{\mathrm{Z}_{C,j}\}_{j=1}^T)$, where $\mathrm{X}_C=v^x, \mathrm{Y}_C=v^y, \mathrm{Z}_{C,j}=v^{z_j}$, and access to oracles for signing, updating, and corruption. The secrets $(x, y, \{z_j\})$ are unknown to $\mathcal{R}$.

To simulate the environment for the adversary $\mathcal{A}$, $\mathcal{R}$ defines a target public key $\mathsf{PK}^*$ under the chosen-key model. First, it samples a random value $\nu \stackrel{\$}{\leftarrow} \mathbb{F}_p$. It then provides $\mathcal{A}$ with $\mathsf{pp}$ and the simulated public key $\mathsf{PK}^*$ defined as follows:
\begin{itemize}
    \item The long-term public key is set to $\mathsf{lvk}^* = (\mathrm{X}^*, \mathrm{Y}^*)$, where:
    \begin{align*}
        \mathrm{X}^* &\leftarrow \mathrm{X}_C = v^x \\
        \mathrm{Y}^* &\leftarrow \mathrm{Y}_C \cdot v^\nu = v^{y+\nu}
    \end{align*}
    \item For each epoch $j \in [1, T]$, the epoch public key is set to $\mathsf{epk}^*_j \leftarrow \mathrm{Z}_{C,j} = v^{z_j}$.
\end{itemize}
This implicitly defines the adversary's target secrets as $x^* = x$ and $y^* = y+\nu$, while the epoch secrets remain unchanged, $z^*_j = z_j$. The generator for $\mathbb{G}_1$ is set to $u$. Finally, $\mathcal{R}$ simulates two programmable random oracles, $H: \{0,1\}^* \to \mathbb{G}_1$ and $F: \{0,1\}^* \to \mathbb{F}_p$. $\mathcal{R}$ maintains lists $L_H$ and $L_F$ of query-response pairs to ensure consistency. The oracle $F$ will be used as the pivot for the Forking Lemma.

\medskip
\noindent\textit{Queries.}
The simulator $\mathcal{R}$ responds to the adversary $\mathcal{A}$'s queries as follows.

\medskip
\noindent\textit{Signing Queries.}
When the adversary $\mathcal{A}$ requests a signature on a message $m \in \mathbb{F}_p$ for an epoch $j$ (with $\mathsf{aux}, \mathsf{ctx}$), the simulator $\mathcal{R}$ must produce a valid signature $\sigma^*$ under the target public key $\mathsf{PK}^*$. $\mathcal{R}$ proceeds as follows:

\begin{enumerate}
    \item \textbf{Randomness Generation.} $\mathcal{R}$ generates two fresh random exponents, $\gamma, \delta \stackrel{\$}{\leftarrow} \mathbb{F}_p$. It also obtains $h$ from the random oracle on input $\mathsf{aux}$.

    \item \textbf{Challenger Query.} $\mathcal{R}$ queries the challenger's signing oracle $\mathcal{O}_{\text{sign}}$ with $(m, h, \gamma, \delta, \mathsf{ctx}, j)$. The challenger $\mathcal{C}$, using its secrets $(x, y, z_j)$, computes and returns the two signature components:
    \begin{align*}
        s_{\mathsf{lt}} &\leftarrow (h^\gamma)^{x+ym} \\
        s_{\mathsf{ep},j} &\leftarrow (h^\delta)^{z_j \cdot F(\mathsf{ctx},j)}
    \end{align*}

    \item \textbf{Signature Transformation.} $\mathcal{R}$ receives $s_{\mathsf{lt}}$ and $s_{\mathsf{ep},j}$ from the challenger. To make the signature valid under the target key $\mathsf{PK}^*$, $\mathcal{R}$ must adjust the long-term component. The epoch-specific component requires no change since $z^*_j = z_j$.
    
    The simulator computes the final signature components for the adversary, denoted with a star:
    \begin{itemize}
        \item The long-term part $s_{\mathsf{lt}}^*$ is computed by applying a corrective term using the secret offset $\nu$:
        \[ s_{\mathsf{lt}}^* \leftarrow s_{\mathsf{lt}} \cdot (h^\gamma)^{\nu m} \]
        \item The epoch-specific part remains unchanged:
        \[ s_{\mathsf{ep},j}^* \leftarrow s_{\mathsf{ep},j} \]
    \end{itemize}

    \item \textbf{Return Signature.} $\mathcal{R}$ assembles the final signature $\sigma^* = (h^\gamma, h^\delta, s_{\mathsf{lt}}^*, s_{\mathsf{ep}}^*)$ and returns it to $\mathcal{A}$.
\end{enumerate}

\smallskip
\noindent\textit{Correctness of the Simulation.} The signature $\sigma^*$ is valid under $\mathsf{PK}^* = (\mathrm{X}^*, \mathrm{Y}^*, \{\mathrm{Z}_j^*\})$ because the components correctly align with the implicitly defined secrets $(x^*, y^*, z^*_j)$.

For the long-term part, the verifier checks against $x^* = x$ and $y^* = y+\nu$:
\begin{align*}
    (h^\gamma)^{x^*+y^*m} &= (h^\gamma)^{x+(y+\nu)m} \\
    &= (h^\gamma)^{x+ym} \cdot (h^\gamma)^{\nu m} \\
    &= s_{\mathsf{lt}} \cdot (h^\gamma)^{\nu m} = s_{\mathsf{lt}}^*
\end{align*}
For the epoch-specific part, the check is against $z^*_j = z_j$, which is trivially correct:
\begin{align*}
    (h^\delta)^{z^*_j \cdot F(\mathsf{ctx},j)} &= (h^\delta)^{z_j \cdot F(\mathsf{ctx},j)} = s_{\mathsf{ep},j} = s_{\mathsf{ep},j}^*
\end{align*}
Thus, the signature $\sigma^*$ is a perfect simulation from the adversary's perspective.

\medskip
\noindent\textit{Update Queries.}
The adversary $\mathcal{A}$ provides a previously obtained signature tuple $(j, m, h, \mathsf{ctx}, h', h'')$ to request the epoch component for $j+1$. The simulator $\mathcal{R}$ must respond using the challenger's oracle, as it does not know the epoch secret $z_{j+1}$.

\begin{enumerate}
    \item \textbf{Forwarding the Query.} $\mathcal{R}$ takes the request from $\mathcal{A}$, which includes the randomized base $h''$ (where $h''=h^\delta$ for some $\delta$ chosen by the challenger during the initial signing). $\mathcal{R}$ forwards the entire valid request to the challenger's update oracle, $\mathcal{O}_{\text{update}}$.

    \item \textbf{Challenger Interaction.} The challenger $\mathcal{C}$ uses the provided $h''$ and its own secret $z_{j+1}$ to compute and return the next epoch component:
    \[ s_{\mathsf{ep}, j+1} \leftarrow (h'')^{F(\mathsf{ctx},j+1) \cdot z_{j+1}} \]

    \item \textbf{Return to Adversary.} $\mathcal{R}$ receives $s_{\mathsf{ep}, j+1}$ and returns it directly to $\mathcal{A}$ as $s_{\mathsf{ep}, j+1}^*$. No transformation is needed. The crucial point is that the returned value is already consistent with the randomized base $h''$ that the adversary possesses.
\end{enumerate}

\smallskip
\noindent\textit{Correctness of the Simulation.} The returned component $s_{\mathsf{ep}, j+1}^*$ is perfectly simulated. It is valid under the target key $\mathsf{PK}^*$ because the epoch secrets are identical ($z^*_{j+1} = z_{j+1}$), and the component correctly corresponds to the randomized base $h'' = h^\delta$. The verification check is:
\begin{align*}
    e(s_{\mathsf{ep}, j+1}^*, v) &= e( (h^\delta)^{F(\mathsf{ctx}, j+1) \cdot z_{j+1}}, v) \\
    &= e(h^\delta, v^{F(\mathsf{ctx}, j+1) \cdot z_{j+1}}) \\
    &= e(h'', v^{F(\mathsf{ctx}, j+1) \cdot z^*_{j+1}})
\end{align*}
This is exactly the equation the adversary would use to verify the component, thus the simulation is flawless.

\medskip
\noindent\textit{Corrupt Queries.}
The simulator $\mathcal{R}$'s response depends on which secret $\mathcal{A}$ requests.

\begin{itemize}
    \item \textbf{Corruption of the Target Long-Term Key.} If $\mathcal{A}$ asks for the secret key $(x^*, y^*)$ for $\mathsf{PK}^*$, $\mathcal{R}$ must abort.

    \item \textbf{Corruption of an Epoch Secret ($z_j^*$).} If $\mathcal{A}$ requests the secret for epoch $j$, $\mathcal{R}$ forwards this query to the challenger's oracle $\mathcal{O}_{\text{corrupt}}(j)$. The challenger returns its secret $z_j$. Since the simulation defines $z_j^* = z_j$, $\mathcal{R}$ passes this value directly to $\mathcal{A}$. This is a perfect simulation. 
\end{itemize}

\noindent\textit{Output.}
Eventually, the adversary $\mathcal{A}$ outputs a valid, non-trivial forgery tuple $(j^*, m^*, h^*, h'^*, h''^*, \mathsf{ctx}^*, s_{\mathsf{lt}}^*, s_{\mathsf{ep},j^*}^*)$.

\noindent\textit{Forgery Extraction via the Forking Lemma.}
The simulator $\mathcal{R}$'s goal is to leverage this forgery to break the underlying STB-GPS assumption. The correct approach is to apply the Forking Lemma by programming the random oracle, which we assume is the function $F(\cdot, \cdot)$.

\smallskip
\noindent\textbf{1. Identifying the Forking Point.}
The adversary $\mathcal{A}$, in order to compute the epoch-specific signature component $s_{\mathsf{ep},j^*}^*$, must query the random oracle for the value $c = F(\mathsf{ctx}^*, j^*)$. The simulator $\mathcal{R}$, which controls the oracle, can therefore choose the output of this query. 

\smallskip
\noindent\textbf{2. Applying the Forking Lemma.}
$\mathcal{R}$ executes the following steps:
\begin{enumerate}
    \item When $\mathcal{A}$ makes the crucial query for $F(\mathsf{ctx}^*, j^*)$, $\mathcal{R}$ records the state of $\mathcal{A}$ and provides a randomly chosen value $c_1 \in \mathbb{F}_p$ as the oracle's output.
    \item With non-negligible probability, $\mathcal{A}$ continues and produces a valid forgery: $(j^*, m^*, h'^*, h''^*, \mathsf{ctx}^*, s_{\mathsf{lt}}^*, s_{\mathsf{ep},1}^*)$. The epoch signature component satisfies $s_{\mathsf{ep},1}^* = (h''^*)^{z_{j^*} \cdot c_1}$.
    \item $\mathcal{R}$ rewinds $\mathcal{A}$ to the recorded state just before the oracle query. It then provides a different, randomly chosen value $c_2 \in \mathbb{F}_p$ ($c_2 \neq c_1$) as the output for the same query $F(\mathsf{stx}^*, j^*)$.
    \item Since the rest of $\mathcal{A}$'s random tape is unchanged, with a significant probability, it will follow a similar execution path and produce a second valid forgery: $(j^*, m^*, h'^*, h''^*, \mathsf{ctx}^*, s_{\mathsf{lt}}^*, s_{\mathsf{ep},2}^*)$. Note that the long-term component $s_{\mathsf{lt}}^*$ and the randomized bases $h'^*, h''^*$ will be the same, as they were determined by $\mathcal{A}$ before the forking point. The new epoch signature component satisfies $s_{\mathsf{ep},2}^* = (h''^*)^{z_{j^*} \cdot c_2}$.
\end{enumerate}

\noindent\textbf{3. Extracting the Secret and Constructing the Final Forgery.}
The simulator $\mathcal{R}$ now possesses two distinct, valid epoch signature components and the corresponding oracle outputs it chose:
\begin{align*}
    s_{\mathsf{ep},1}^* &= (h''^*)^{z_{j^*} \cdot c_1} \\
    s_{\mathsf{ep},2}^* &= (h''^*)^{z_{j^*} \cdot c_2}
\end{align*}
Let the unknown value be $K = (h''^*)^{z_{j^*}}$. The equations become:
\begin{align*}
    s_{\mathsf{ep},1}^* &= K^{c_1} \\
    s_{\mathsf{ep},2}^* &= K^{c_2}
\end{align*}
Since the simulator knows the distinct exponents $c_1$ and $c_2$, it can easily solve for $K$. For instance, from the first equation $s_{\mathsf{ep},1}^* = K^{c_1}$, $\mathcal{R}$ can compute the modular inverse of $c_1$ in $\mathbb{F}_p$ and find $K$ directly:
$$ K \leftarrow (s_{\mathsf{ep},1}^*)^{c_1^{-1}}$$
The simulator has thus successfully computed $K = (h''^*)^{z_{j^*}}$, which is the core of the epoch-specific secret for the un-corrupted epoch $j^*$.
The simulator has thus successfully computed $K = (h''^*)^{z_{j^*}}$, which is the core of the epoch-specific secret for the un-corrupted epoch $j^*$.

\noindent\textit{Conclusion.}
With the extracted value $K=(h''^*)^{z_{j^*}}$, the simulator $\mathcal{R}$ can now construct a valid solution to present to the STB-GPS challenger. The challenger expects a solution for the tuple $(m^*, j^*, \mathsf{ctx}^*)$ that consists of secret components corresponding to its own public key $\mathsf{PK_C}$. $\mathcal{R}$ constructs this solution as follows.

First, $\mathcal{R}$ must construct the long-term component for the challenger's key $(x,y)$. The component $s_{\mathsf{lt}}^*$ provided by the adversary is valid for the simulated key $(x^*, y^*)$, where $y^* = y+\nu$. Specifically, $s_{\mathsf{lt}}^* = (h'^*)^{x + (y+\nu)m^*} = (h'^*)^{x+ym^*} \cdot (h'^*)^{\nu m^*}$. To obtain the component valid for the challenger's key, $\mathcal{R}$ performs a reverse transformation:
$$ s_{\mathsf{lt}, \text{final}} \leftarrow s_{\mathsf{lt}}^* \cdot (h'^*)^{-\nu m^*} $$
This yields $(h'^*)^{x+ym^*}$, which is the correct long-term secret component relative to the base $h'^*$ that solves the first part of the STB-GPS challenge.

Second, let $c_{\text{chal}} = F(\mathsf{ctx}^*, j^*)$ be the value that the challenger's random oracle would output. $\mathcal{R}$ uses the extracted value $K$ to compute the final epoch component:
$$ s_{\mathsf{ep}, \text{final}} \leftarrow K^{c_{\text{chal}}} = \left((h''^*)^{z_{j^*}}\right)^{c_{\text{chal}}} = (h''^*)^{z_{j^*} \cdot c_{\text{chal}}} $$
This is the correct epoch secret component relative to the base $h''^*$.

By presenting a valid forgery tuple $(j^*, m^*, h^*, h'^*, h''^*, \mathsf{ctx}^*, s_{\mathsf{lt}, \text{final}}, s_{\mathsf{ep}, \text{final}})$ containing these correctly formed components to the challenger, $\mathcal{R}$ breaks the STB-GPS assumption. 

\end{proof}

\section{Entities and Oracles}\label{oracles}
The challenger $\mathcal{C}$ maintains several lists to track the game state:
\begin{itemize}[nosep, leftmargin=*]
	\item $\mathcal{HU}, \mathcal{CU}$: honest and corrupted users, respectively.
	\item $\mathcal{HCI}, \mathcal{CCI}$: honest and corrupted credential issuers, respectively.
	\item $\mathcal{L}_{\mathsf{uk}}$: a list of users’ keys.
	\item $\mathcal{L}_{\mathsf{cred}}$: credential records, where each entry contains $(\mathsf{cred}, attr, uid)$ representing the issued credential, its attributes, and the user's identifier. An entry may be $\bot$ if the corresponding credential has not yet been issued.
\end{itemize}
\begin{itemize}[nosep, leftmargin=*]
	\item $\mathcal{O}^{\mathsf{HCI}}(i)$: For a given identifier $i$, this oracle creates a new honest credential issuer. It first checks whether $i$ exists in $\mathcal{HCI} \cup \mathcal{CCI}$, outputting $\bot$ if true. Otherwise, it generates issuer keys through $\mathsf{KGen}(\mathsf{pp},i)$ and $\mathsf{TKGen}(\mathsf{pp}, i)$, registers the complete key set $(i, \mathsf{lsk}_i, \mathsf{lvk}_i, \mathsf{tsk}_i, \mathsf{tvk}_i)$ to $\mathcal{HCI}$, and outputs $(\mathsf{lvk}_i, \mathsf{tvk}_i)$.
	
	\item $\mathcal{O}^{\mathsf{CCI}}(i)$: This oracle corrupts issuer $i$, subject to $(\mathsf{lvk'}, \mathsf{tvk}_j')$ remaining uncorrupted. For non-existent issuers ($i \notin \mathcal{HCI} \cup \mathcal{CCI}$), it creates a new entry in $\mathcal{CCI}$. For honest issuers ($i \in \mathcal{HCI}$), it transfers $i$ to $\mathcal{CCI}$ and exposes $(\mathsf{lsk}_i, \mathsf{tsk}_{i,j})$.
	
	\item $\mathcal{O}^{\mathsf{User}}(id, \mathrm{S})$: This oracle initializes a new user with identity $id$ and a set $\mathrm{S} = {(m_i, \mathsf{lvk}_i)}_{i \in [\ell]}$, where $m_i$ denotes the attribute to be signed. For non-existent users ($id \notin \mathcal{HU} \cup \mathcal{CU}$), it creates a fresh entry via $ (\mathsf{usk}, \mathsf{uvk}, \mathsf{aux}) \leftarrow \mathsf{UKGen}$, registers the user in $\mathcal{HU}$, stores the key information in $\mathcal{L}_{\mathsf{uk}}$, and returns $\mathsf{uvk}$. Otherwise, outputs $\bot$.
	
	\item $\mathcal{O}^{\mathsf{CU}}(id)$: This oracle corrupts user $id$. For unregistered users ($id \notin \mathcal{HU}$), it creates a new entry in $\mathcal{CU}$. For honest users ($id \in \mathcal{HU}$), it transfers $id$ to $\mathcal{CU}$ and exposes $\mathsf{usk}$ along with tuples $(id, m_i, \mathsf{cred}_i)$ from $\mathcal{L}_{\mathsf{cred}}[id]$.
	
	\item $\mathcal{O}^{\mathsf{ObtIss}}(id, i, m_i)$: This is an honest issuing oracle accepts a user identity $id$, an issuer identity $i$, and attribute $m_i$ as input. It first verifies that $id \in \mathcal{HU}$ and $i \in \mathcal{HCI}$, returning $\bot$ if either check fails. Upon successful verification, it retrieves the user's secret key $\mathsf{usk}$ from $\mathcal{L}_{\mathsf{uk}}[id]$ and the issuer's secret key $\mathsf{lsk}_i$ from $\mathcal{HCI}$. The oracle then executes the issuing protocol between the user and issuer for the specified attributes $m_i$, where: 
	\begin{equation}
		\footnotesize
		\begin{aligned}
			&[\mathsf{CredObtain}(id, \mathsf{aux}, m_i) \leftrightarrow \mathsf{CredIssue}(j, \mathsf{lsk}_i, \mathsf{tsk}_{i, j})] \\
			&\quad\quad\quad\quad \rightarrow (\mathsf{cred}_{\mathsf{lt}, i}, \mathsf{cred}_{\mathsf{ep},i,j}) \nonumber
		\end{aligned}
	\end{equation}
	$\text{Add the entry }(id, m_i, \mathsf{cred}_{\mathsf{lt}, i}, \mathsf{cred}_{\mathsf{ep},i,j})\text{ to }\mathcal{L}_{\mathsf{cred}}$.
	
	\item $\mathcal{O}^{\mathsf{Obtain}}(id, i, m_i)$: This is an honest obtaining oracle with malicious issuer. The oracle accepts a user identity $id$, an issuer identity $i$, and an attribute $m_i$ as input. It first verifies that $id \in \mathcal{HU}$ and $i \in \mathcal{CCI}$, returning $\bot$ if either check fails. Upon successful verification, it retrieves the user's secret key $\mathsf{usk}$ from $\mathcal{L}_{\mathsf{uk}}[id]$. The oracle then executes the obtaining protocol between the honest user and the malicious issuer for the specified attributes $m_i$. The user side of the protocol is simulated honestly, while the adversary controls the issuer's actions, where:
	\[\begin{aligned}
		\langle\mathsf{CredObtain}(id, \mathsf{aux}, m_i) \leftrightarrow \mathcal{A}\rangle \rightarrow (\mathsf{cred}_{\mathsf{lt}, i}, \mathsf{cred}_{\mathsf{ep},i,j})
	\end{aligned}\]
	If $\mathsf{cred}_{\mathsf{lt}, i} = \bot$ or $\mathsf{cred}_{\mathsf{ep},i,j} = \bot$, return $\bot$. Otherwise, append $(id, m_i, \mathsf{cred}_{\mathsf{lt}, i}, \mathsf{cred}_{\mathsf{ep},i,j})$ to $\mathcal{L}_{\mathsf{cred}}$. 
	
	\item $\mathcal{O}^{\mathsf{Issue}}(id, i, m_i)$: This oracle executes the credential obtaining protocol between a malicious user and an honest issuer. It accepts a user identity $id$, an issuer identity $i$, and an attribute $m_i$ as input. The oracle first verifies that $id \in \mathcal{HU}$ and $i \in \mathcal{CCI}$, returning $\bot$ if either check fails. Upon successful verification, it retrieves the issuer's secret key $\mathsf{lsk}$ from $\mathcal{HCI}$. The oracle then executes the obtaining protocol between the honest user and the malicious issuer for the specified attribute $m_i$. The user side of the protocol is simulated honestly using $\mathsf{usk}$, while the adversary controls the issuer's actions.
	\[\begin{aligned}
		\langle\mathcal{A} \leftrightarrow \mathsf{CredIssue}(j, \mathsf{lsk}, \mathsf{tsk}_j)\rangle \rightarrow (\mathsf{cred}_{\mathsf{lt}, i}, \mathsf{cred}_{\mathsf{ep},i,j})
	\end{aligned}\]
	Add the entry $(id, m_i, \mathsf{cred}_{\mathsf{lt}, i}, \mathsf{cred}_{\mathsf{ep},i,j})$ to $\mathcal{L}_{\mathsf{cred}}$.
	
	\item $\mathcal{O}^{\mathsf{Anch}\text{-}b}(id_0, id_1, \mathbb{M})$: This oracle takes as inputs the identities of two honest users who have the same user attribute set $\mathcal{D}$. If $(id_0, id_1) \notin \mathcal{HU} \vee \mathbb{M}_{id_0} \neq \mathbb{M}_{id_1}$, return $\bot$. The oracle parses $\mathcal{L}_{\mathsf{cred}}[id_0] = (id_0,\{m_{0,i}, \mathsf{cred}_{0, \mathsf{lt}, i}, \mathsf{cred}_{0, \mathsf{ep}, j}\}_{i \in [\ell]})$, $\mathcal{L}_{\mathsf{cred}}[id_1] = (id_1,\{m_{1,i}, \mathsf{cred}_{1, \mathsf{lt}, i}, \mathsf{cred}_{1, \mathsf{ep}, j}\}_{i \in [\ell]})$. The oracle then computes $\mathsf{cred}_{\mathsf{agg},j,b}$ (definition of aggregation process). Finally, the oracle runs the interactive protocol \begin{align*}
    \langle \mathsf{CredShow}(&\mathsf{usk}_b, \mathsf{pol}, \{\mathsf{lvk}_i, \mathsf{tvk}_{i,j}, m_{b,i}\}_{i \in [\ell]}, \\&\mathsf{cred}_{\mathsf{agg},j,b}, \mathbb{M}) \leftrightarrow \mathcal{A}\rangle
\end{align*}, where $b \in \{0,1\}$, and outputs the result $b'$.
	
	\item $\mathcal{O}^{\mathsf{Blch}\text{-}b}(id, \{m_{0,i}\}_{i \in [\ell]}, \{m_{1,i}\}_{i \in [\ell]})$: This oracle takes as inputs two different sets of attributes $\{m_{0,i}\}_{i \in [\ell]}$, $\{m_{1,i}\}_{i \in [\ell]}$ for the honest user identity $id$. If $id \notin \mathcal{HU}$, return $\bot$. The oracle runs the interactive protocol:
\begin{align*}
    \langle\mathsf{CredObtain}(id, \mathsf{aux}, m_i) & \leftrightarrow \mathsf{CredIssue}(j, \mathsf{lsk}_i, \mathsf{tsk}_{i, j})\rangle \\
    & \rightarrow (\mathsf{cred}_{\mathsf{lt}, i}, \mathsf{cred}_{\mathsf{ep},i,j})
\end{align*}
	Then computes $\mathsf{cred}_{\mathsf{agg},j,b}$ (definition of aggregation process), and outputs $b'$.

	\item $\mathcal{O}^{\mathsf{UpdEp}}(\mathsf{st}_j)$: This oracle facilitates the transition from epoch $j$ to epoch $j+1$. It accepts the current system state $\mathsf{st}_j$ as input and executes the following procedure: For each issuer $i \in [n]$, it generates new epoch-specific key pairs $(\mathsf{tsk}_{i,j+1}, \mathsf{tvk}_{i,j+1}) \leftarrow \mathsf{TKGen}$. Subsequently, the oracle updates the epoch-specific keys in both $\mathcal{HCI}$ and $\mathcal{CCI}$, replacing $(\mathsf{tsk}_{i,j}, \mathsf{tvk}_{i,j})$ with $(\mathsf{tsk}_{i,j+1}, \mathsf{tvk}_{i, j+1})$. If any operation within this process fails, the oracle returns $\bot$. 
	
	\item $\mathcal{O}^{\mathsf{ObtIssEp}}(id, i, j+1)$: This is an honest issuing oracle accepts a user identity $id$, an issuer identity $i$, and an epoch index $j$ as input. It first verifies that  $i \in \mathcal{HCI}$ and $ \mathcal{L}_{\mathsf{cred}}[id] \neq \bot$, returning $\bot$ if either check fails. Upon successful verification, it retrieves the issuer's epoch-specific secret key $\mathsf{tsk}_{i,j+1}$ from $\mathcal{HCI}$. The oracle then executes an epoch-specified credential for update, where:
	\[\begin{aligned}
		\mathsf{CredIssueEp}(\mathsf{tsk}_j, id) \rightarrow \mathsf{cred}_{\mathsf{ep}, i, j+1} \nonumber
	\end{aligned}\]
	Update the entry $(id, m_i, \mathsf{cred}_{\mathsf{lt}, i}, \mathsf{cred}_{\mathsf{ep},i,j})$ in $\mathcal{L}_{\mathsf{cred}}$ to $(id, m_i, \mathsf{cred}_{\mathsf{lt},i}, \mathsf{cred}_{\mathsf{ep},i, j+1})$.
	
	\item $\mathcal{O}^{\mathsf{IssEp}}(id, i, j)$: This is an malicious issuing oracle accepts a user identity $id$, an issuer index $i$, and an epoch index $j$ as input. It first verifies that  $i \in \mathcal{CCI}$ and $ \mathcal{L}_{\mathsf{cred}}[id] \neq \bot$, returning $\bot$ if either check fails. Upon successful verification, it retrieves the issuer's epoch-specific secret key $\mathsf{tsk}_{i,j+1}$ from $\mathcal{CCI}$. The oracle then executes an epoch-specified credential for update, where:
	$$\mathsf{CredIssueEp}(\mathsf{tsk}_j, id) \rightarrow \mathsf{cred}_{\mathsf{ep}, i, j+1}$$
	
	Update the entry $(id, m_i, \mathsf{cred}_{\mathsf{lt}, i}, \mathsf{cred}_{\mathsf{ep},i,j})$ in $\mathcal{L}_{\mathsf{cred}}$ to $(id, m_i, \mathsf{cred}_{\mathsf{lt},i}, \mathsf{cred}_{\mathsf{ep},i, j+1})$.
	
	\item $\mathcal{O}^{\mathsf{CredShow}}(k, j, \mathsf{pol}, \mathbb{M})$: This oracle takes as input an issuance index $k$, an epoch index $j$, a policy $\mathsf{pol}$, and an attributes-subset $\mathbb{M}$. It first parses $\mathcal{L}_{\mathsf{cred}}[k]$ to obtain $(id, m_k, \mathsf{cred}_{\mathsf{lt}, k}, \mathsf{cred}_{\mathsf{ep},k,j})$. If $id \notin \mathcal{HU}$, it returns $\bot$. Otherwise, it executes the credential showing protocol $\mathsf{CredShow}$ between the honest user (with identity $id$) and the adversary $\mathcal{A}$, where:
	\[\begin{aligned}
		\mathsf{CredShow}(\mathsf{tg}, \{m_i, \mathsf{lvk}_i, \mathsf{tvk}_{i,j}\}_{i \in [\ell]}, \mathsf{cred}_{\mathsf{agg}, j}, \mathbb{M}, \pi) \leftrightarrow \mathcal{A}
	\end{aligned}\]
\end{itemize}

\section{Security Analysis of MA-ACEW}\label{sec Ma-acew} 
Building on the security model described previously, we provide formal proofs for the three security properties of MA-ACEW: unforgeability, anonymity, and blindness.
\subsection{Proof of Theorem  \ref{masec1}}
\begin{proof}
	Intuitively, an adversary $\mathcal{A}$ could attempt to break the unforgeability of MA-ACEW by forging an EB-PS signature on the challenge public key, which would allow verification without possessing the required attributes. We prove that if there exists an adversary $\mathcal{A}$ that wins the unforgeability game (Fig. ~\ref{Unforgeability}) with non-negligible probability $\epsilon$, then we can construct a reduction $\mathcal{R}$ that breaks the unforgeability of the underlying EB-PS scheme. The reduction proceeds as follows:
	
	\noindent\textit{Setup.} $\mathcal{R}$ interacts with a challenger $\mathcal{C}$ in the unforgeability game of EB-PS while simultaneously simulating the MA-ACEW unforgeability game for adversary $\mathcal{A}$. Initially, $\mathcal{R}$ receives from $\mathcal{C}$ the values $(\mathsf{lvk}, \mathsf{tvk}_1)$, where $\mathsf{lvk} = (\mathrm{X} = v^x, \mathrm{Y} = v^y)$ and $\mathsf{tvk}_1 = \mathrm{Z}_1 = z_1$ represents the initial time epoch, along with the public parameters $\mathsf{pp}$ of the bilinear group BG. $\mathcal{R}$ then constructs the challenge key as $\mathsf{vk}' = (\mathrm{X},\mathrm{Y},\mathrm{Z}_1)$ and forwards $(\mathsf{pp},\mathsf{vk}')$ to $\mathcal{A}$. All oracle queries are handled as in the real game, with the following exception: instead of using the challenge signing key $\mathsf{sk}'$, $\mathcal{R}$ forwards relevant signing queries to the signing oracle provided by the EB-PS unforgeability game:
	
	\smallskip
	\noindent$\mathcal{O}^{\mathsf{User}}(id)$: On input a user identity $id$, $\mathcal{R}$ first checks if $id \in \mathcal{HU}$ or $id \in \mathcal{CU}$. If so, return $\bot$. Otherwise, $\mathcal{R}$ generates a fresh user key pair $(\mathsf{usk}, \mathsf{uvk}) \leftarrow \mathsf{UKGen}$ and creates the auxiliary information $\mathsf{aux}$ using commitments and ElGamal paris. $\mathcal{R}$ then adds the tuple $(id,(\mathsf{usk}, \mathsf{uvk},\mathsf{aux}))$ to both $\mathcal{HU}$ and $\mathcal{L}_{\mathsf{uk}}$ respectively, and returns $\mathsf{uvk}$.
    \vspace{0.1in}
	
\noindent$\mathcal{O}^{\mathsf{Obtlss}}(i, id, m_i)$: On input issuer index $i$, user identity $id$, and attribute $m_i$, if $id \notin \mathcal{HU}$ or $i \notin \mathcal{HCI} \cup \{\mathsf{vk}'\}$, return $\bot$. Otherwise, if $\mathsf{vk}_i \neq \mathsf{vk}'$, $\mathcal{R}$ retrieves $(\mathsf{usk}, \mathsf{uvk}, \mathsf{aux})$ from $\mathcal{L}_{\mathsf{uk}}$ and $(\mathsf{lsk}, \mathsf{tsk}_1)$ from $\mathcal{HCI}$, then computes $\sigma_{\mathsf{lt}, i} \leftarrow \mathsf{SignLt}(\mathsf{lsk}, \mathsf{uvk}, \mathsf{aux}, m_i)$ and $\sigma_{\mathsf{ep},i, 1} \leftarrow \mathsf{SignEp}(\mathsf{tsk}_1, \mathsf{uvk}, \mathsf{aux})$. If $\mathsf{lvk}_i' = \mathsf{vk}'$, $\mathcal{R}$ forwards the query to the signing oracle of EB-PS, obtaining $\sigma_{\mathsf{lt},i'} \leftarrow \mathcal{O}^{\mathsf{SignLt}}(m_i', \mathsf{aux}, \mathsf{uvk}), \sigma_{\mathsf{ep}, i', 1} \leftarrow \mathcal{O}^{\mathsf{SignEp}}(\mathsf{uvk}, \mathsf{aux})$, and adds $(id, m_i, \mathsf{cred}_i)$ to $\mathcal{L}_{\mathsf{cred}}$, where $\mathsf{cred}_i = (\sigma_{\mathsf{lt}, i}, \sigma_{\mathsf{ep}, i, 1}, \mathsf{uvk})$.
\vspace{0.1in}
	
\noindent$\mathcal{O}^{\mathsf{Issue}}(i, id, m_i)$: On input issuer index $i$, user identity $id$, and attribute $m_i$, if $id \notin \mathcal{CU}$ or $i \notin \mathcal{HCI} \cup \{\mathsf{vk}'\}$, return $\bot$. Otherwise, if $\mathsf{lvk}_i \neq \mathsf{vk}'$, compute $\sigma_{\mathsf{lt},i} \leftarrow \mathsf{SignLt}(\mathsf{lsk}, \mathsf{uvk}, \mathsf{aux}, m_i)$ and $\sigma_{\mathsf{ep},i,1} \leftarrow \mathsf{SignEp}(\mathsf{tsk}_1, \mathsf{uvk}, \mathsf{aux})$. Else, ask the queries $\sigma_{\mathsf{lt},i'} \leftarrow \mathcal{O}^{\mathsf{SignLt}}(m_i', \mathsf{aux}, \mathsf{uvk})$ and $\sigma_{\mathsf{ep},i',1} \leftarrow \mathcal{O}^{\mathsf{SignEp}}(\mathsf{uvk}, \mathsf{aux})$ of EB-PS, add the entry $(id,m_i,\mathsf{cred}_i)$ to $\mathcal{L}_{\mathsf{cred}}$, where $\mathsf{cred}_i = (\sigma_{\mathsf{lt},i}, \sigma_{\mathsf{ep},i,1}, \mathsf{uvk})$.
\vspace{0.1in}
	
\noindent$\mathcal{O}^{\mathsf{ObtlssEp}}(id, i, \mathsf{st}_{j+1})$: On input a user identity $id$, an issuer index $i$, and the new epoch state $\mathsf{state}_{j+1}$, if $id \notin \mathcal{HU}$ or $i \notin \mathcal{HCI} \cup {\mathsf{vk}'_{i}}$, return $\bot$. Otherwise, if $\mathsf{tvk}_{i} \neq \mathsf{tvk}'_{i}$, retrieve $(\mathsf{usk}, \mathsf{uvk}, \mathsf{aux})$ from $\mathcal{L}_{\mathsf{uk}}$ and $\mathsf{tsk}_{i,j+1}$ from $\mathcal{HCI}$, then compute $\sigma_{\mathsf{ep}, i, j+1} \leftarrow \mathsf{SignEp}(\mathsf{tsk}_{i, j+1}, \mathsf{uvk}, \mathsf{aux})$. If $\mathsf{tvk}_{i} = \mathsf{tvk}'_{i}$, query the epoch-specific signing oracle to obtain $\sigma_{\mathsf{ep}, i', j+1} \leftarrow \mathcal{O}^{\mathsf{SignEp}}(\mathsf{uvk}, \mathsf{aux})$. Finally, securely erase the previous epoch signature $\sigma_{\mathsf{ep}, i, j}$ and update the credential to $\mathcal{L}_{\mathsf{cred}}$, reset $\mathsf{cred}_{i} = (\sigma_{\mathsf{lt},i}, \sigma_{\mathsf{ep}, i, j+1}, \mathsf{uvk})$.
	\vspace{0.1in}
	
	\noindent$\mathcal{O}^{\mathsf{IssEp}}(id, i, \mathsf{st}_{j+1})$: On input a user identity $id$, an issuer index $i$, and the new epoch state $\mathsf{st}_{j+1}$, if $id \notin \mathcal{CU}$ or $i \notin \mathcal{HCI} \cup \{\mathsf{vk}'\}$, return $\bot$. Otherwise, if $\mathsf{tvk}_i \neq \mathsf{tvk}'$, compute $\sigma_{\mathsf{ep},i,j+1} \leftarrow \mathsf{SignEp}(\mathsf{tsk}_{i,j+1}, \mathsf{uvk}, \mathsf{aux})$. If $\mathsf{tvk}_i = \mathsf{tvk}'$, query the epoch-specific signing oracle to obtain $\sigma_{\mathsf{ep},i',j+1} \leftarrow \mathcal{O}^{\mathsf{SignEp}}(\mathsf{uvk}, \mathsf{aux})$. Finally, securely erase the previous epoch signature $\sigma_{\mathsf{ep}, i, j}$ and update the credential to $\mathcal{L}_{\mathsf{cred}}$, reset $\mathsf{cred}_{i} = (\sigma_{\mathsf{lt},i}, \sigma_{\mathsf{ep}, i, j+1}, \mathsf{uvk})$.
	\vspace{0.1in}
	
	\noindent Upon receiving a valid showing proof $(\mathsf{avk}, \mathsf{cred}^{\ast}, \mathbb{M}, \mathsf{tg}^{\ast})$ with its associated proofs $(\pi_b, \pi_{\mathsf{IPA}}, \pi_{\mathsf{tg}^{\ast}})$, the reduction $\mathcal{R}$ leverages the knowledge soundness of the underlying proof systems to extract the necessary witnesses.
Specifically, from the successful verification of $\pi_{\mathsf{IPA}}$, whose knowledge soundness is proven in \cite{das2023threshold}, $\mathcal{R}$ extracts a bit vector $\bm{b}$ satisfying $\langle \bm{pk}_x, \bm{b}\rangle = \mathrm{X}$ and $\langle \bm{w}_{j'}, \bm{b} \rangle \geq \mathrm{W}_{\mathsf{acc}}$.
Similarly, as $\pi_{\mathsf{tg}^{\ast}}$ is a Zero-Knowledge Proof of Knowledge, $\mathcal{R}$ extracts the exponents $(\gamma^{\ast}, \delta^{\ast})$ from it.
	
	By the unforgeability definition, no credentials held by corrupt users can be valid for the attribute set $\mathbb{M}$. Formally, for all credentials $\mathsf{cred}_{i,id}$ on $m_{i,id}$ and $\mathsf{uvk}$ with $id \in \mathcal{CCU}$, we have $\mathbb{M} \not\subseteq \bigcup_{i\in[\ell]}m_{id,i}$. Consequently, at least one key in $\mathsf{vk}_i \in \mathsf{avk}$ must be the challenge key, with its corresponding attribute in $m_i \in \mathbb{M}$. $\mathcal{R}$ then retrieves all $(\mathsf{lsk}_i, \mathsf{tsk}_i) \in \mathcal{HCI} \cup \mathcal{CCI}$ corresponding to $(\mathsf{lvk}_i,\mathsf{tvk}_i) \in \mathsf{CI}'$ for $i \in [\ell]$, and constructs $\mathsf{ask} = \{\mathsf{lsk}_i, \mathsf{tsk}_i\}_{i\in[\ell]}$. This process yields a valid forgery $(\mathsf{avk}, (\mathsf{usk}^*, \mathsf{tg}^*), \mathbb{M}, \mathsf{ask}, \sigma^*)$ against our signature scheme, thereby breaking the unforgeability of EB-PS and concluding our proof.   
\end{proof}

\subsection{Proof of Theorem  \ref{masec2}}
\begin{proof}
	The anonymity of our scheme stems from two key properties: the unlinkability of EB-PS signatures due to the DDH assumption and the zero-knowledge property of the proof system. The former ensures that a credential tuple $(\sigma, \mathsf{tg})$ can be perfectly randomized as $(\sigma' = ((h')^{r},s^{r}), \mathsf{tg}^{r})$, where $r \stackrel{\$}{\leftarrow} \mathbb{F}_p^{\ast}$, statistically obfuscating all information about the original signature-tag pair. Here, $\mathsf{tg}$ serves as a pseudonym during interactions. The latter guarantees that the proof $\pi_{\mathsf{tg}}$ leaks no information about the witness.
	
	In the anonymity experiment, the witnesses used in computing $\pi_{\mathsf{tg}}$ are valid for both $b \in \{0,1\}$, and the signature-tag pairs undergo proper randomization. Consequently, the adversary's view comprises solely of random elements that are identically distributed, independent of $b$. This implies that no probabilistic polynomial-time (PPT) adversary can distinguish between the two cases with non-negligible advantage.
	
	We formalize this intuition through a sequence of games. Let $\mathcal{S}$ denote the event that the adversary correctly guesses bit $b$, with $\mathcal{S}_i$ representing this event in Game$_i$. The proof evolves through the following key transformations:
	
	\noindent\textbf{Game}$_\mathbf{0}$: This is the anonymity game as given in Fig. \ref{Anonymity}.
	
	\noindent\textbf{Game}$_\mathbf{1}$: We change the way we generate proofs in the original anonymity game. Instead of using real proofs, we use simulated proofs for all $\mathsf{NIZK}(\mathsf{tg})$ in $\mathsf{CredObtain}$ and $\mathsf{CredShow}$ respectively.
	
	\noindent\textbf{Game}$_\mathbf{2}$: We change the way we run queries in the experiment. Let $q_u$ be the number of $\mathcal{O}^{\mathsf{User}}$ queries. At the beginning of \textbf{Game}$_\mathbf{2}$, we pick $k \leftarrow [q_u]$ to guess when the challenge user (who owns the $i_b$-th credential) is registered. Modify oracles as follows:
	\begin{itemize}
		\item $\mathcal{O}^{\mathsf{User}}(id, \mathrm{S})$: As in Game$_1$, but if this is the $k$-th call then, setting $id^* \leftarrow id$.
		\item $\mathcal{O}^{\mathsf{CU}}(id)$: If $id \in \mathcal{CU}$, it returns $\bot$ (as in the previous games). If $id = id^*$ then the experiment stops and outputs a random bit $b' \stackrel{\$}{\leftarrow} \{0,1\}$. Otherwise, if $id \in \mathcal{HU}$, it returns user $\mathsf{usk}$ and credentials and moves $id$ from $\mathcal{HU}$ to $\mathcal{CU}$.
		\item $\mathcal{O}^{\mathsf{Anch}\text{-}b}(id_0,id_1, \mathbb{M})$: If $id^\ast$ is not in the credential list $\mathcal{L}_{\mathsf{cred}}[id_b]$, the oracle game terminates and outputs the guess $b'$.
	\end{itemize}
	
	\noindent$\textbf{Game}_\textbf{3}$: We modify the scheme by altering the sampling method for the tag $\mathsf{tg}$. Instead of deriving it from the credential, we randomly generate $\mathsf{tg}$.
	
	\noindent\textbf{Game}$_\mathbf{4}$: We no longer use the stored credential (signature) and tag $(id_b,m_b, \mathsf{cred}_b)$ from the list $\mathcal{L}_{\mathsf{cred}}[id_b]$ to perform RndSigTag. Instead, we directly generate new random signatures.

    \noindent\textbf{Game}$_{\mathsf{5}}$: When the challenge oracle $\mathcal{O}^{\mathsf{Anch}\text{-}b}$ is called, it no longer generates the credential pair based on the identity $id_b$. Instead, it constructs a randomized credential pair $(\mathsf{cred}_{\mathsf{sim}}, \mathsf{tg}_{\mathsf{sim}})$  independent of the bit~$b$. All other components of the experiment remain unchanged from \textbf{Game}$_\mathbf{4}$.

	\noindent\textbf{Game}$_\mathbf{0} \rightarrow$ \textbf{Game}$_\mathbf{1}$: By perfect zero-knowledge of NIZK, we have that:
	\[\Pr[\mathcal{S}_1] = \Pr[\mathcal{S}_0]\]

	\noindent\textbf{Game}$_\mathbf{1} \rightarrow$ \textbf{Game}$_\mathbf{2}$: There exists at least one anonymity query with input $(id_0,id_1,\mathbb{M})$ where $id_0, id_1 \in \mathcal{HU}$. When $id^{\ast} = id_b$, which occurs with probability $\frac{1}{q_u}$, the game continues without abortion. Moreover, $id_b$ must remain uncorrupted prior to this query ($id_b \in \mathcal{HU}$), and any subsequent corruption attempts on $id^{\ast}$ return $\bot$. Therefore, combining both cases yields: When $id^{\ast}$ is selected independently of the challenge bit $id_b$ (which occurs with probability $1 - \frac{1}{q_u}$), $\mathcal{S}_2$ maintains at least probability $\frac{1}{2}$. In the case where $id^{\ast} = id_b$ (occurring with probability $\frac{1}{q_u}$), the view of $\mathcal{A}$ is identical to Game$_1$, and thus the success probability is exactly $\Pr[\mathcal{S}_1]$. Combining these cases yields, we have:
	
	\[\Pr[\mathcal{S}_2] \geq \frac{1}{2}(1-\frac{1}{q_u}) + \frac{1}{q_u} \cdot \Pr[\mathcal{S}_1]\]
	
	\noindent\textbf{Game}$_\mathbf{2} \rightarrow$ \textbf{Game}$_\mathbf{3}$: The difference between these two games is that we use freshly generated tag $\mathsf{tg}_{\mathsf{fresh},b}$, which are indirectly guaranteed by the DDH assumption. The oracles are simulated as in \textbf{Game}$_2$, except for the following oracle:
	
	$\mathcal{O}^{\mathsf{User}}(id, \mathrm{S})$: As in \textbf{Game}$_1$, but if this is the $k$-th call then, setting $id^* \leftarrow id$, it sets $\mathsf{usk}[id] \leftarrow \bot$ and $\mathsf{uvk} \leftarrow \mathsf{tg}_{\mathsf{fresh}}$.
	
	$\mathcal{O}^{\mathsf{Anch}\text{-}b}(id_0,id_1,\mathbb{M})$: This oracle works as in \textbf{Game}$_\mathsf{2}$, except that for $id^{\ast} = \mathcal{L}_{\mathsf{cred}}[id_b]$, the game generates a random tag $\mathsf{tg}_{\mathsf{fresh},b}$ instead of using the stored one. The game selects $b$ and sends $(\mathsf{tg}_b,\mathsf{cred}_b,\pi)$ to $\mathcal{A}$, then receives $b'$ from $\mathcal{A}$. 
    
    Let $\epsilon_{\mathsf{DDH}}$ denote the advantage of solving the DDH problem and $q_{A}$ be the number of queries to the $\mathcal{O}^{\mathsf{Anch}\text{-}b}$ oracle. Apart from the adversary's advantage in solving the DDH problem, their success probability might also increase if certain bad events cause the simulation to fail. The total probability of such events is bounded by $(1+2q_A)/p$, where $p$ is the order of the underlying finite field $\mathbb{F}_p$ from which random elements are chosen. The $1$ term corresponds to a one-time failure probability during the initial setup, while the $2q_A$ term accounts for potential failures across the $q_A$ oracle queries, where each query involves two random choices that could break the simulation. Thus we have:
	\[|\Pr[\mathcal{S}_2] - \Pr[\mathcal{S}_3]| \leq \epsilon_{\mathsf{DDH}}(\lambda) + (1+2q_A)
	\frac{1}{p}\]

\noindent\textbf{Game}$_\mathbf{3} \rightarrow$ \textbf{Game}$_\mathbf{4}$: Credentials obtained from $\mathsf{RndSigTag}$ are identically distributed for all valid tuples $(\mathbb{M},\mathsf{tg^*},\mathsf{vk},\mathsf{cred^*})$. We thus have:
	
	\[\Pr[\mathcal{S}_3] = \Pr[\mathcal{S}_4]\]

    \noindent\textbf{Game}$_\mathsf{4}\rightarrow$ \textbf{Game}$_\mathsf{5}$:
The only modification is in the generation of the challenge credential. In \textbf{Game}$_\mathsf{5}$, the credential, which was honestly generated for identity $id_b$ \textbf{Game}$_\mathsf{4}$, is now produced by a simulation algorithm, rendering its distribution independent of the bit $b$. The computational indistinguishability between \textbf{Game}$_\mathsf{4}$ and \textbf{Game}$_\mathsf{5}$ is based on the DDH assumption. It follows that for any adversary $\mathcal{A}$, $$|\Pr[\mathcal{S}_4] - \Pr[\mathcal{S}_5]| \leq \epsilon_{\text{DDH}}(\lambda)$$

\noindent\textbf{Analysis of Game}$_\mathbf{5}$: In this final game, the challenge credential pair $(\mathsf{cred}_{\mathsf{sim}}, \mathsf{tg}_{\mathsf{sim}})$ is generated independently of the challenge bit $b$. Therefore, the adversary's entire view is statistically independent of $b$. This implies the adversary has no advantage over a random guess. Thus, \(\Pr[\mathcal{S}_5] = \frac{1}{2}\).

\noindent\textbf{Conclusion:} We can now bound the adversary's advantage in the original game, $\mathsf{Adv}_{\mathcal{A}}(\lambda) = |\Pr[\mathcal{S}_0] - 1/2|$. By combining the (in)equalities from the game sequence, we have:
\begin{align*}
	\mathsf{Adv}^{\mathsf{ANO\text{-}}b}_{\mathcal{A}}(\lambda) 
	&= |\Pr[\mathcal{S}_0] - 1/2| \\
	&= |\Pr[\mathcal{S}_1] - 1/2| \\
	&= \left| q_u \left(\Pr[\mathcal{S}_2] - \frac{1}{2}\left(1-\frac{1}{q_u}\right)\right) - \frac{1}{2} \right| \\
	&\leq q_u \Biggl( |\Pr[\mathcal{S}_2] - \Pr[\mathcal{S}_3]| + |\Pr[\mathcal{S}_3] - \Pr[\mathcal{S}_4]| \\
	& \qquad + |\Pr[\mathcal{S}_4] - \Pr[\mathcal{S}_5]| \Biggr) + \text{negl}(\lambda) \\
	&\leq q_u \left( \epsilon_{\mathsf{DDH}}(\lambda) + \frac{1+2q_A}{p} + 0 + \epsilon_{\text{DDH}}(\lambda) \right) + \text{negl}(\lambda) \\
	&\leq 2q_u \cdot \epsilon_{\mathsf{DDH}}(\lambda) + \text{negl}(\lambda)
\end{align*}
	The preceding analysis follows a standard game-hopping argument. For a more rigorous and detailed treatment of the probability bounds in such proofs, we refer the reader to \cite{fuchsbauer2019structure}. Since $\epsilon_{\mathsf{DDH}}(\lambda)$ is a negligible function in the security parameter $\lambda$, and $q_u$ is polynomial, the adversary's total advantage is negligible. This completes the proof.
\end{proof}

\subsection{Proof of Theorem \ref{masec3}}
\begin{proof}
	The blindness of our scheme stems from two key properties: the IND-CPA of ElGamal encryption and the zero-knowledge property of the proof system. We formalize the process with the following sequence of games:

    \noindent\textbf{\text{Game}$_{0}$:} This is the blindness game as given in Fig. \ref{Blindness}. 

    \noindent\textbf{\text{Game}$_{1}$:} We modify the behavior of the challenge oracle $\mathcal{O}^{\mathsf{Blch}\text{-}b}$. Instead of generating real proofs, we simulate proofs. All other computations remain unchanged.

\noindent\textbf{\text{Game}$_{2}$:} In this game, we further alter the challenge oracle $\mathcal{O}^{\mathsf{Blch}\text{-}b}$. Instead of encrypting the actual attributes $\{m_{b,i}\}_{i \in [\ell]}$, it now encrypts freshly sampled random messages $\{r_i\}_{i \in [\ell]}$ from the message space.

\noindent\textbf{ \textbf{Game}$_{\mathsf{0}} \rightarrow$  \textbf{Game}$_{\mathsf{1}}$:} The only difference between \textbf{Game}$_{\mathsf{0}}$ and \textbf{Game}$_{\mathsf{1}}$ is the use of simulated proofs instead of real ones in the challenge oracle. Because the NIZK system provides perfect zero-knowledge, the distribution of real proofs is identical to the distribution of simulated proofs. Therefore, the adversary's view in \textbf{Game}$_{\mathsf{1}}$ is statistically identical to its view in \textbf{Game}$_{\mathsf{0}}$. Thus, we have:
$$ \Pr[\mathcal{S}_0] = \Pr[\mathcal{S}_1] $$

\noindent\textbf{ \text{Game}$_{\mathsf{1}} \rightarrow$  \textbf{Game}$_{\mathsf{2}}$:} The transition from $\textbf{Game}_{\mathsf{1}}$ to $\textbf{Game}_{\mathsf{2}}$ is based on the $\text{IND-CPA}$ security of the encryption scheme.
In $\textbf{Game}_{\mathsf{1}}$, the oracle encrypts a real message $\{m_{b,i}\}$, while in $\textbf{Game}_{\mathsf{2}}$, it is modified to encrypt a random message $r_i$ (sampled from the message space) instead.
The $\text{IND-CPA}$ property ensures that these two games are computationally indistinguishable.
Therefore, the difference in the adversary's success probabilities is bounded by the advantage of the encryption scheme:
\[
|\Pr[\mathcal{S}_1] - \Pr[\mathcal{S}_2]| \leq \epsilon_{\mathsf{Enc}}^{\mathsf{IND\text{-}CPA}}(\lambda)
\]

\noindent\textbf{Analysis of \text{Game}$_{2}$:} In this game, the output of the challenge oracle $\mathcal{O}^{\mathsf{Blch}\text{-}b}$ (both the ciphertext and the simulated proofs) is computed based on a random message that is completely independent of the challenge bit $b$. Consequently, the adversary's entire view contains no information about $b$. Thus, 
$ \Pr[\mathcal{S}_2] = \frac{1}{2}.$

\noindent\textbf{Conclusion:}
By combining the above steps, we can bound the adversary's advantage in the original experiment:
\begin{align*}
    \mathsf{Adv}_{\mathcal{A}}^{\mathsf{BLI\text{-}}b}(\lambda) &= |\Pr[\mathcal{S}_0] - 1/2| \\
    &= |\Pr[\mathcal{S}_1] - 1/2| \\
    &\leq |\Pr[\mathcal{S}_1] - \Pr[\mathcal{S}_2]| + |\Pr[\mathcal{S}_2] - 1/2|  \\
    &\leq \epsilon^{\mathsf{IND\text{-}CPA}}_{\mathsf{Enc}}(\lambda) + |1/2 - 1/2|  \\
    &= \epsilon^{\mathsf{IND\text{-}CPA}}_{\mathsf{Enc}}(\lambda)
\end{align*}
Since we assume the encryption scheme is IND-CPA secure, the advantage $\epsilon^{\mathsf{IND\text{-}CPA}}_{\mathsf{Enc}}(\lambda)$ is negligible for any probabilistic polynomial-time adversary. We conclude that the adversary's advantage in the blindness game is also negligible. This completes the proof.
	
\end{proof}

\section{Additional Properties} \label{additi}
\subsection{Multi-Attribute Credential Issuance For a Single Issuer} We can note that in MA-ACEW, the secret key of issuer $\mathsf{CI}_i$ consists of $(x_i, y_i, z_{i,j})$. Currently, only $y_i$ is used for signing attributes, due to the underlying structure of the signature scheme. However, this can be easily extended to a multi-message setting by expanding the issuer's secret keys to $(x_i, y_{1,i}, \ldots, y_{t,i}, z_{i,j})$, corresponding to the verification keys $(\mathrm{X}_i, \mathrm{Y}_{1,i}, \ldots, \mathrm{Y}_{t,i}, \mathrm{Z}_{i,j})$.

With this extension, credential issuer $\mathsf{CI}_i$ can issue credentials for the attribute vector $\vec{m} = (m_1, \ldots, m_t)$. The process is adapted as follows:

First, the user generates a single ElGamal key pair $(\mathsf{esk}, \mathsf{evk}) = (d, \zeta = u^{d})$. For each attribute $m_j$ where $j \in [1,t]$, the user samples random values $k_j, o_j \in \mathbb{F}_p$ and computes:
\begin{itemize}
    \item An ElGamal encryption of the attribute: $c_j = \mathsf{Enc}((h')^{m_j},\mathsf{evk}) = (u^{k_j}, \zeta^{k_j}\cdot (h')^{m_j})$.
    \item A Pedersen commitment to the attribute: $c_{m_j} = u^{m_j} h_1^{o_{j}}$.
\end{itemize}

The user then generates a single, comprehensive proof $\pi_{\mathsf{CI}_i}$ that attests to the well-formedness of all encrypted attributes for issuer $\mathsf{CI}_i$:
\begin{align*}
	\pi_{\mathsf{CI}_i} = \mathsf{ZKPOK}\{ & (d, m_1, \dots, m_t, o_1, \dots, o_t, k_1, \dots, k_t): \\
	& \zeta = u^d \land {} \\
	& \forall j \in [1,t]: c_{m_j}= u^{m_j} h_1^{o_{j}} \land {} \\
	& \forall j \in [1,t]: c_j = (u^{k_j}, \zeta^{k_j}\cdot (h')^{m_j}) \}
\end{align*}

The issuing phase between the user and issuer $\mathsf{CI}_i$ is updated for the vector of attributes:
\begin{itemize}[leftmargin=*]
	\item The user transmits the relevant parts of $(\mathsf{tg}, \mathsf{aux}, \pi_{\mathsf{CI}_i}, \pi_{\mathsf{tg}})$ to the credential issuer $\mathsf{CI}_i$. Here, $\mathsf{aux}$ now contains the set of all ciphertexts $\{c_j\}_{j \in [1,t]}$ and commitments $\{c_{m_j}\}_{j \in [1,t]}$ intended for this issuer.
	
	\item The issuer $\mathsf{CI}_i$ parses $\mathsf{aux}$ and verifies the proofs. Upon success, it computes the blinded credential $\tilde{\mathsf{cred}}_{\mathsf{lt},i} = (\tilde{\sigma}_1, \tilde{\sigma}_2)$, where:
    \begin{align*}
        \tilde{\sigma}_1 &= (h')^{x_i} \prod_{j=1}^t (c_{j,2})^{y_{j,i}} \\
        \tilde{\sigma}_2 &= \prod_{j=1}^t (c_{j,1})^{y_{j,i}}
    \end{align*}
    The issuer then sends $\tilde{\mathsf{cred}}_{\mathsf{lt},i}$ to the user.

	\item The user receives the blinded credential $\tilde{\mathsf{cred}}_{\mathsf{lt},i} = (\tilde{\sigma}_1, \tilde{\sigma}_2)$ and unblinds it to construct the final credential $\mathsf{cred}_{\mathsf{lt},i} = (\sigma_1, \sigma_2)$. The first component is set as $\sigma_1 = h'$, and the second component is computed by removing the blinding factor:
    $$ \sigma_2 = \tilde{\sigma}_1 \cdot (\tilde{\sigma}_2)^{-d} $$
\end{itemize}

The final long-term credential $\mathsf{cred}_{\mathsf{lt},i}$ from issuer $\mathsf{CI}_i$ for the message vector $\vec{m}$ is correctly constructed as:
$$\mathsf{cred}_{\mathsf{lt},i} = (\sigma_1, \sigma_2) = \left( h', \; (h')^{x_i + \sum_{j=1}^t m_j y_{j,i}} \right)$$
This approach has been previously done in works such as~\cite{PS16short} and~\cite{sonnino2018coconut}.

\subsection{Issuer Hiding}
Issuer Hiding refers to the ability of a user to show a credential to a verifier without revealing which specific issuer has issued which credential. Instead, they only demonstrate whether a defined policy on the acceptance set of issuers is satisfied. In traditional schemes, this process may expose the credential issuer's public key, which could leak information about the user's privacy. For instance, if a user's credential is issued by a specific authority (e.g., a local government), revealing this information may expose the user's location.

The issuer-hiding feature \cite{connolly2022improved, DBLP:conf/ccs/MirBGLS23} , works roughly as follows: Each verifier generates a policy defining a set of acceptable issuers, identified by their verification keys. This policy is represented as a collection of Structure-Preserving Signatures on Equivalence Classes (SPSEQ)\cite{fuchsbauer2019structure} on the verification keys of the EB-PS. To accommodate the property, we need to consider incorporate the SPSEQ scheme.

In the $\mathsf{Gen\text{-}Policies}$ phase, the verifier runs $(\mathsf{vsk}, \mathsf{vvk})\leftarrow\mathsf{SPSEQ.KGen}(\mathsf{pp})$. Additionally, the verifier uses $\mathsf{vsk}$ to generate $\sigma_{\mathsf{SPS},i} \leftarrow \mathsf{SPSEQ.Sign}(\mathsf{vsk}, \mathsf{lvk}_i)$ for $i \in \mathbb{I}$, where $\mathbb{I}$ is the set of acceptable issuer's public keys. The policy is set as $\mathsf{pol} = (\mathrm{W}_{\mathsf{acc}},j, \{\mathsf{lvk}_i, \sigma_{\mathsf{SPS},i}\}_{i \in \mathbb{I}})$.

In the showing phase, the user selects a disclosed set $\mathbb{M}$, where $|\mathbb{M}| \leq |\mathbb{I}|$, to form a credential $\mathsf{cred}$. The user then applies $\mathsf{RndSigTag}$ to randomize $\mathsf{cred}$ and $\mathsf{tg}$ simultaneously using randomness $r$. The same $r$ is used to randomize $\sigma_{\mathsf{SPS},i}$ and $\mathsf{lvk}_i$ for $i \in \mathbb{M}$, resulting in $\mathsf{lvk}'_i, \sigma'_{\mathsf{SPS},i}$ which remain valid signatures, but prevent the verifier from learning the original issuer. More details on this process can be found in \cite{DBLP:conf/ccs/MirBGLS23}.

\noindent\textbf{Remark.} We note that if the participant set is small or the weight space is limited, the aggregated weight $\mathrm{W}_\mathsf{claim}$ may correspond to a unique combination of issuers, potentially revealing their identities. However, we can effectively prevent this leakage by:

\begin{itemize}[nosep, leftmargin=*]
    \item \textbf{Padding with zero‑weight issuers.} The issuer universe can be expanded by appending dummy issuers with weight 0. During aggregation, users may randomly include some of these dummy issuers in proofs. Since a verifier cannot distinguish real and dummy issuers, the anonymity set grows from $n$ to $n + k$,  the number of consistent subsets grows combinatorially in $k$ ($\approx 2^k$), which greatly complicates inference of the exact issuer set. For instance, with $n = 50$ active issuers and $k = 100$ zero‑weight fillers, even when $\mathrm{W}_\mathsf{claim}$ equals the sum of genuine weights, the verifier cannot tell which subset among the $n + k$ participants contributed. 
    \item \textbf{Range Proof Disclosure.} The system can disclose only that  $\mathrm{W}_\mathsf{claim} \geq$  $\mathrm{W}_\mathsf{acc}$ or $\mathrm{W}_\mathsf{claim} \in [\mathrm{W}_\mathsf{acc} - \Delta, \mathrm{W}_\mathsf{acc} + \Delta]$ instead of the exact $\mathrm{W}_\mathsf{claim}$. This disclosure can be realized using zero‑knowledge range‑proof techniques, requiring only minor modifications.  
\end{itemize}

\subsection{Selective Disclosure of Arbitrary Attributes}
Recall that for clarity of presentation in Section \ref{madef} and Section \ref{macons}, we assumed that all attributes are disclosed to the verifier in the $\mathsf{Show}$ protocol. We now detail how our MA-ACEW scheme supports the more general policy of selective disclosure of arbitrary attributes. This is achieved with minor modifications to the $\mathsf{Gen\text{-}Policies}$ and $\mathsf{Show}$  protocols.

Concretely, the interaction is extended as follows. First, in the $\mathsf{Show}$  protocol, the verifier specifies a disclosure policy. This policy defines:
\begin{itemize}
    \item A set of attribute indices $\mathcal{D}$ that the user is required to disclose.
    \item A predicate $\Phi$ that must be satisfied by the attributes corresponding to the hidden indices in $\mathcal{H} = \mathcal{U} \setminus \mathcal{D}$, where $\mathcal{U}$ is the universe of all attribute indices.
\end{itemize}
Correspondingly, the user, holding the full attribute message set $\mathbb{M}$, partitions it based on the verifier's policy into two disjoint subsets:
\begin{itemize}
    \item The set of disclosed messages, $\mathbb{M}_{\mathcal{D}} = \{m_i \in \mathbb{M} \mid i \in \mathcal{D}\}$.
    \item The set of hidden messages, $\mathbb{M}_{\mathcal{H}} = \{m_i \in \mathbb{M} \mid i \in \mathcal{H}\}$.
\end{itemize}
The user then reveals the messages in $\mathbb{M}_{\mathcal{D}}$ to the verifier, while simultaneously proving in zero-knowledge that their hidden messages in $\mathbb{M}_{\mathcal{H}}$ satisfy the predicate $\Phi$.

Accordingly, the $\mathsf{Gen\text{-}Policies}$ protocol is updated to output a policy $\mathsf{pol}$ defined as the tuple:
$$ \mathsf{pol} = (j+1, \mathrm{W}_{\mathsf{acc}}, \mathcal{D}, \Phi) $$
where $\mathcal{D}$ is the set of disclosed attribute indices and $\Phi$ is the predicate that the hidden attributes (indexed by $\mathcal{H} = \mathcal{U} \setminus \mathcal{D}$) must satisfy. We now detail how a user, holding their attribute set $\mathbb{M}$, satisfies the requirements specified by a policy $\mathsf{pol}$.

To support selective disclosure, we adapt the baseline verification protocol. Recall that the original $\mathsf{AggVerify}$ algorithm validates a signature $\sigma_{\mathsf{agg}, j} = (h', s)$ on attributes $\mathbb{M} = \{m_i\}_{i=1}^\ell$ by checking:
\begin{equation} \label{eq:original_verification}
e\left(h', \prod_{i=1}^\ell \left(\mathrm{X}_i \cdot \mathrm{Y}_i^{m_i}\right)\right) \cdot 
e\left((h^{\delta})^{F(\mathsf{ctx},j)}, \prod_{i=1}^\ell \mathrm{Z}_{i,j}\right) = e(s,v) 
\end{equation}
We now partition $\mathbb{M}$ into a disclosed subset $\mathbb{M}_{\mathcal{D}}$ (indexed by $\mathcal{D}$) and a hidden subset $\mathbb{M}_{\mathcal{H}}$ (indexed by $\mathcal{H}$). Instead of revealing attributes in $\mathbb{M}_{\mathcal{H}}$, the user computes a commitment $C$ to them using a random blinding factor $r \in \mathbb{F}_p$:
\begin{equation} \label{eq:commitment}
C = \textsf{Com}(\mathbb{M}_{\mathcal{H}}) = v^r \cdot \prod_{k \in \mathcal{H}}\left(\mathrm{X}_k \cdot \mathrm{Y}_k^{m_k}\right)
\end{equation}
The prover then derives a modified signature component $s'$ that allows the verifier to check the credential's validity using only the disclosed attributes $\mathbb{M}_{\mathcal{D}}$ and the commitment $C$. The new verification equation is:
\begin{align} \label{eq:new_verification}
&e\!\left(h', \left(\prod_{i \in \mathcal{D}} (\mathrm{X}_i \cdot \mathrm{Y}_i^{m_i})\right) \cdot C \right) \notag \\
&\quad \cdot\; e\!\left((h^{\delta})^{F(\mathsf{ctx},j)}, \prod_{i=1}^\ell \mathrm{Z}_{i,j}\right) 
\stackrel{?}{=} e(s',v)
\end{align}
To determine the correct form of $s'$, we substitute Eq.~\eqref{eq:commitment} into the left-hand side of Eq.~\eqref{eq:new_verification}:
\begin{align*}
    \text{LHS}_{\text{new}} 
    &= e\left(h', \left(\prod_{i \in \mathcal{D} \cup \mathcal{H}} (\mathrm{X}_i \cdot \mathrm{Y}_i^{m_i})\right) \right) \cdot e(h', v^r) \\
    &\qquad \cdot e\left((h^{\delta})^{F(\mathsf{ctx},j)}, \prod_{i=1}^\ell \mathrm{Z}_{i,j}\right) \\
    &= \underbrace{e\left(h', \prod_{i=1}^\ell \left(\mathrm{X}_i \cdot \mathrm{Y}_i^{m_i}\right)\right) \cdot e\left((h^{\delta})^{F(\mathsf{ctx},j)}, \prod_{i=1}^\ell \mathrm{Z}_{i,j}\right)}_{\text{Original LHS from Eq.~\eqref{eq:original_verification}}} \\
    &\qquad \cdot e((h')^r, v) \\
    &= e(s, v) \cdot e((h')^r, v) \\
    &= e(s \cdot (h')^r, v)
\end{align*}
This derivation shows that the equality holds if the prover sets $s' = s \cdot (h')^r$. The $\mathsf{Show}$ protocol thus requires the prover to convince the verifier of two things: that the modified signature is valid, and that the committed attributes satisfy some policy.

To achieve the latter, the prover generates a comprehensive NIZK proof, $\pi_{\text{NIZK}}$. This proof must simultaneously establish the correctness of a commitment to the hidden attributes $\mathbb{M}_{\mathcal{H}}$ and the satisfaction of a predicate $\Phi$ over these same attributes. Specifically, $\pi_{\text{NIZK}}$ demonstrates knowledge of the hidden attributes $\mathbb{M}_{\mathcal{H}}$ and the randomness $r$ such that:
\begin{enumerate}
    \item The commitment is correctly formed: $C = v^r \cdot \prod_{k \in \mathcal{H}}\left(\mathrm{X}_k \cdot \mathrm{Y}_k^{m_k}\right)$.
    \item The attributes $\mathbb{M}_{\mathcal{H}}$ satisfy the predicate $\Phi$.
\end{enumerate}

Crucially, to support arbitrary predicates, our design handles the second part by integrating a general-purpose ZK system (e.g., the efficient zk-SNARK \textbf{PLONK}~\cite{DBLP:journals/iacr/GabizonWC19}) in a black-box fashion. The proof for $\Phi(\mathbb{M}_{\mathcal{H}})$ is generated, and $\pi_{\text{NIZK}}$ then proves that the witness used for the SNARK is the same set of attributes $\mathbb{M}_{\mathcal{H}}$ committed to in $C$. This design choice makes the underlying proof system for predicates a modular component, orthogonal to our core protocol, allowing it to be selected based on application-specific requirements.

In summary, the prover sends the tuple $(\mathbb{M}_{\mathcal{D}}, C, (h', s'), \pi_{\text{NIZK}})$. The verifier accepts if and only if the signature check in Eq.~\eqref{eq:new_verification} passes and the proof $\pi_{\text{NIZK}}$ is valid.

While these additional properties provide enhanced functionalities, they also introduce increased computational costs. Therefore, practitioners can adjust the implementation of these features according to their specific practical scenarios and performance requirements.

\end{document}